\documentclass{article}
\usepackage[auth-sc]{authblk}
\usepackage[toc,page]{appendix}
\usepackage{natbib}

\RequirePackage[colorlinks,citecolor=blue,urlcolor=blue]{hyperref}

\usepackage{amssymb}
\usepackage{amsthm}
\usepackage{amsmath}
\usepackage{booktabs}
\usepackage{dsfont}
\usepackage{graphicx}
\usepackage[dvipsnames]{xcolor}
\usepackage{tabu}
\usepackage{here}

\usepackage[utf8]{inputenc}

\DeclareMathOperator{\argmin}{argmin}

\newcommand{\bPsi}{\boldsymbol{\Psi}}
\newcommand{\bEta}{\boldsymbol{\eta}}
\newcommand{\eps}{\varepsilon}

\newtheorem{proposition}{Proposition}
\newtheorem{assumption}[proposition]{Assumption}
\newtheorem{theorem}[proposition]{Theorem}
\theoremstyle{remark}
\newtheorem{remark}[proposition]{Remark}
\newtheorem{lemma}[proposition]{Lemma}

\begin{document}

\title{The validity of bootstrap testing in the threshold framework}
\author[1]{Simone Giannerini}
\author[1,2]{Greta Goracci}
\author[3]{Anders Rahbek}

\affil[1]{Department of Statistical Sciences, University of Bologna, Italy}
\affil[2]{Faculty of Economics and Management, Free University of Bozen-Bolzano, Italy}
\affil[3]{Department of Economics, University of Copenhagen, Denmark}

\maketitle
\begin{abstract}
We consider bootstrap-based testing for threshold effects in non-linear threshold autoregressive (TAR) models. It is well-known that classic tests based on asymptotic theory tend to be oversized in the case of small, or even moderate sample sizes, or when the estimated parameters indicate non-stationarity, as often witnessed in the analysis of financial or climate data. To address the issue we propose a supremum Lagrange Multiplier test statistic (sLMb), where the null hypothesis specifies a linear autoregressive (AR) model against the alternative of a TAR model. We consider a recursive bootstrap applied to the sLMb statistic and establish its validity. This result is new, and requires the proof of non-standard results for bootstrap analysis in time series models; this includes a uniform bootstrap law of large numbers and a bootstrap functional central limit theorem. These new results can also be used as a general theoretical framework that can be adapted to other situations, such as regime-switching processes with exogenous threshold variables, or testing for structural breaks. The Monte Carlo evidence shows that the bootstrap test has correct empirical size even for small samples, and also no loss of empirical power when compared to the asymptotic test. Moreover, its performance is not affected if the order of the autoregression is estimated based on information criteria. Finally, we analyse a panel of short time series to assess the effect of warming on population dynamics.
\end{abstract}



\section{Introduction}\label{sec:intro}
The problem of testing for a linear time series model versus its threshold extension has attracted considerable attention for a number of reasons. First and foremost, threshold autoregressive models (TAR) are among the simplest nonlinear specifications and retain a good interpretability. Second, they can encompass many complex features such as jumps, limit-cycles, time irreversibility and chaos, see e.g. \cite{Ton90,Ton11}. \cite{Pet92} proved that TAR models approximate a wide range of nonlinear autoregressive processes.
Moreover, they have been proven to describe successfully many real-world phenomena in economics and finance, see e.g. \cite{Cha17b,Han11,Ton17}. For population biology and climate studies see e.g. \cite{Ste98,Yao00}.
\par
Seminal works on asymptotic quasi-likelihood ratio tests for threshold autoregression include \citet{Cha90a}, \citet{Cha90b}, \citet{Cha91}. Other contributions include those of \citet{Pet86}, \cite{Su17} and that of \citet{Tsa98} for the multivariate case. Tests based upon Lagrange Multipliers were proposed in \citet{Luu88} for the smooth transition case and \citet{Won97, Won00} for TAR models with conditional heteroscedasticity, see also \cite{Ton11} for a review.
\par
The main theoretical problem associated with testing for threshold autoregression is the nuisance parameter (the threshold) being present only under the alternative hypothesis, as adduced in \cite{Dav77,Dav87} and \cite{And93}. In the present context, one solution is to derive the test statistic as a random function of the nuisance parameter. Then, the overall test statistic is the supremum (or some other convenient function) of the statistic over the grid of values of the nuisance parameter. The derivation of the null distribution of the overall test statistic requires proving the stochastic equicontinuity (tightness) of the sequence of random functions, see e.g. \cite{Vaa98}, and this is often the most challenging task.
\par
One key issue with asymptotic tests is the sample size requirement to deliver a good performance. Typically, the rate of convergence towards the asymptotic null distribution depends upon the true parameters values of the data generating process and might produce a size bias that can be severe, for instance when the processes are close to non-stationarity and/or non-invertibility, see e.g. \cite{Gor21}. Furthermore, the null distribution, which has no closed analytical form, depends both upon the threshold range and the number of tested parameters, so that one has to make use of simulated critical values for each combination of the threshold grid and number of parameters, see \cite{And03}. One way to overcome the aforementioned problems is to resort to resampling methods. \cite{Han96} proposes tests based on a stochastic permutation device where the score function is randomly perturbed through an auxiliary random variable. The same approach has been deployed in \cite{Li11} to test a linear model against its threshold ARMA extension by means of a quasi likelihood ratio statistic. More recently, \cite{Hil21} adopts a similar approach to introduce robust conditional moment tests of omitted nonlinearity. To the best of our knowledge, to date, there are no available results on the validity of the classical bootstrap (both parametric and nonparametric) for testing a linear AR model against a TAR model.
\par
 In this paper we fill this gap and provide a proof of the validity of the test based on a residual bootstrap. In particular, we consider a supremum Lagrange Multiplier test statistic (sLMb) where the null hypothesis specifies a linear AR$(p)$ model against the alternative of a TAR$(p)$ model. One of the main advantages of Lagrange multiplier tests over likelihood ratio tests is that the former only need estimating the model under the null hypothesis and avoid direct estimation of the TAR model.
 \par
 We prove that, under the null hypothesis, the bootstrap distribution of the test statistic coincides with the asymptotic distribution derived in \cite{Cha90a} for the likelihood ratio test, namely, a functional of a centered Gaussian process. Note that, as also shown for instance in \cite{Han96}, the Wald, the supLM and the likelihood-ratio test statistics share the same asymptotic distribution. The inherent difficulties associated with working in the bootstrap framework, i.e. simultaneously coping with the two kinds of randomness (the first one is the sampling variability and the second one is the bootstrap variability) are amplified by the discontinuity of the threshold function and the absence of the nuisance parameter under the null hypothesis. To this end, we provide a \emph{uniform} bootstrap law of large numbers and a functional bootstrap central limit theorem that can be used as a general theoretical framework that can be adapted to other situations, such as regime-switching processes with exogenous threshold variables or testing for structural breaks.
\par
The simulation study shows that the bootstrap test (sLMb) has a correct empirical size for a series' length as small as 50, even when the data generating process is close to non-stationarity, a situation that produces a severe oversize in the asymptotic version of the test. Moreover, the behaviour of the bootstrap test is not influenced by treating the order of the tested process as unknown and estimating it by means of the AIC. Again, this is not the case with the asymptotic test, which results oversized. The good performance of the bootstrap test in small samples makes it applicable to many applied situations where either data collection/production is expensive, as in the experimental context, or longer series are simply not available, as in e.g. \cite{Yao00}.
\par
We apply our test to a panel of 12 short time series of populations of larvae of the pyralid moth \emph{Plodia interpunctella} under different experimental conditions. The data come from \cite{Lau19} where the aim was to assess the effect of warming, which is one of the consequences of climate change, in connection to age structure, density, and generation cycles. We find a significant threshold effect and an appropriate specification that accounts for the 6-week characteristic asymmetric cycle.
\par
The rest of the paper is organised as follows. Section~\ref{sec:pre} introduces the problem and describes the theory behind the standard asymptotic sLM test. In Section~\ref{sec:boot} we present the bootstrap version of the test, together with the main results on its validity. Section~\ref{sec:MCstudy} shows the finite sample behaviour of the tests where our bootstrap test (sLMb) is compared to the asymptotic test (sLMa), also when the order of the tested process is unknown and has to be estimated. Section~\ref{sec:applic} is devoted to the real application. All the proofs are detailed in Section~\ref{sec:proof}.

\subsection{Notation}\label{sec:notation}
We write $P^*(\cdot)$, $E^*[\cdot]$ to indicate probability and expectation conditional on the data, respectively; $\xrightarrow[n\to\infty]{w^*}_p$ denotes the weak convergence in probability and $Y^*_n\xrightarrow[n\to\infty]{p^*}_p Y$ or, equivalently,  $Y^*_n-Y=o_{p^*}(1)$,  means that, for any $\delta>0$, $P^*(\|Y^*_n-Y\|>\delta)\xrightarrow[n\to\infty]{p} 0$; lastly, $Y_n^*=O_{p^*}(1)$ means that, for any $\delta>0$, there exists $M>0$ such that $P(P^*(\|Y_n^*\|>M)<\delta)$ is arbitrarily close to one for sufficiently large $n$. Here, $\|\cdot\|$ is the $\mathcal{L}^2$ matrix norm (the Frobenius' norm, i.e. $\|A\|=\sqrt{\sum_{i=1}^{n}\sum_{j=1}^{m}|a_{ij}|^2}$, where $A$ is a $n\times m$ matrix); $\|A\|=(E[A]^r)^{1/r}$ is the $\mathcal{L}^r$ norm of a random matrix. Moreover, let $\mathcal{D}_\mathds{R}(a,b)$, $a<b$ be the space of functions from $(a,b)$ to $\mathds{R}$ that are right continuous with left-hand limits.

\section{Preliminaries}\label{sec:pre}

Let the time series $\{X_t\}$ follow the threshold autoregressive TAR$(p)$ model defined by the difference equation:
\begin{align}\label{eq:TAR}
  X_t&=\phi_{0}+\sum_{i=1}^{p}\phi_{i}X_{t-i}+\left(\Psi_{0}+\sum_{i=1}^{p}\Psi_{i}X_{t-i}\right)I(X_{t-d}\leq r)+\eps_t.
\end{align}
\noindent
The positive integers $p$ and $d$ are the autoregressive order and the delay parameter, respectively; we assume $p$ and $d$ to be known. Moreover $I(\cdot)$ indicates the indicator function and $r\in\mathds{R}$ is the threshold parameter. The innovations $\{\eps_t\}$ are independent and identically distributed (iid) with $E[\eps_t]=0$ and $E[\eps_t^2]=\sigma^2<\infty$. For each $t$, $\eps_t$ is independent of $X_{t-1}$, $X_{t-2}$, \dots. Clearly, Eq.~(\ref{eq:TAR}) specifies a regime-switching process where each regime follows a linear autoregressive process. The parameters are given by
\begin{align*}
\boldsymbol\phi&=(\phi_{0},\phi_{1},\ldots,\phi_{p})^\intercal\in\Theta_{\phi};\\
\boldsymbol\Psi&=(\Psi_{0},\Psi_{1},\ldots,\Psi_{p})^\intercal\in\Theta_{\Psi};\\
\boldsymbol\eta&=(\boldsymbol\phi^\intercal,\boldsymbol\Psi^\intercal,\sigma^2)^\intercal\in\Theta=\Theta_{\phi}\times\Theta_{\Psi}\times(0,+\infty),
\end{align*}
with $\Theta_{\phi}$ and $\Theta_{\Psi}$ being subsets of $\mathds{R}^{p+1}$. We use $\boldsymbol\eta=(\boldsymbol\phi^\intercal,\boldsymbol\Psi^\intercal,\sigma^2)^\intercal$ to refer to unknown parameters, whereas the true parameters are indicated by
$$\boldsymbol\eta_0=(\boldsymbol\phi_0^\intercal,\boldsymbol\Psi_0^\intercal,\sigma^2_0)^\intercal=(\phi_{0,0},\phi_{0,1},\ldots,\phi_{0,p}, \Psi_{0,0},\Psi_{0,1},\ldots,\Psi_{0,p},\sigma^2_0)^\intercal.$$
We test whether the TAR model fits the data significantly better than its linear counterpart. As $\boldsymbol\Psi$  contains the differences of the autoregressive parameters in the two regimes, the system of hypotheses reduces to
$$
\begin{cases}
  H_0&:\boldsymbol\Psi=\boldsymbol 0 \\
  H_1&: \boldsymbol\Psi\neq\boldsymbol 0,
\end{cases}
$$
\noindent
where $\boldsymbol 0$ is the vector of zeros. Suppose we observe $\{X_t,t=1,\dots,n\}$. We develop the Lagrange Multiplier (hereafter LM) test based on the quasi Gaussian log-likelihood conditional on the initial values $X_0,X_{-1},\ldots,X_{-p+1}$:
\begin{equation}\label{log-like}
\ell_n(\boldsymbol\eta,r)= -\frac{1}{2\sigma^2}\sum_{t=1}^n \eps_t^2(\boldsymbol\eta,r),
\end{equation}
where
\begin{align}\label{eq:eps}
\eps_t(\boldsymbol\eta,r)=X_t&-\left\{\phi_{0}+\sum_{i=1}^{p}\phi_{i}X_{t-i}\right\}-\left\{\Psi_{0}+\sum_{i=1}^{p}\Psi_{i}X_{t-i}\right\}I\left(X_{t-d}\leq r\right).
\end{align}
\noindent
Under the null hypothesis, model (\ref{eq:TAR}) reduces to an AR$(p)$ model:
\begin{equation}\label{eq:AR}
  X_t=\phi_0+\sum_{i=1}^{p}\phi_iX_{t-i}+\eps_t,
\end{equation}
and let $\tilde{\boldsymbol{\phi}}=(\tilde{\phi}_{0},\tilde{\phi}_{1},\ldots,\tilde{\phi}_{p})^\intercal$ be the Maximum Likelihood Estimator (hereafter MLE) of the autoregressive parameters in Eq.~(\ref{eq:AR}) based upon the Gaussian likelihood, i.e.:
$$
\tilde{\boldsymbol{\phi}}=\underset{\boldsymbol{\phi}\in\Theta_\phi}{\argmin}\;-\frac{1}{2\sigma^2}\sum_{t=1}^{n}\eps_t^2((\boldsymbol{\phi},\boldsymbol 0,\sigma^2),r).
$$
The associated residuals (\textit{restricted} residuals) are
\begin{align}
\tilde{\eps}_t&=X_t-\tilde{\phi}_0-\sum_{i=1}^{p}\tilde{\phi}_iX_{t-i} =(\phi_{0,0}-\tilde{\phi}_0)+\sum_{i=1}^{p}(\phi_{0,i}-\tilde{\phi}_i)X_{t-i}+\eps_t.\label{eps_tilde}
\end{align}
Moreover, $\sigma^2$ is estimated by
\begin{equation}\label{eq:sigma}
 \tilde{\sigma}^2=\frac{1}{n-p-1}\sum_{t=1}^{n}\tilde{\eps}^2_t.
\end{equation}
\noindent
Lastly, define $\tilde{\bEta}=(\tilde{\boldsymbol{\phi}}^\intercal,\boldsymbol 0^\intercal,\tilde\sigma^2)$, i.e. $\tilde{\bEta}$ is the resticted MLE under the null hypothesis.
\par
In order to test the null hypothesis define:
$$\left.\frac{\partial\tilde\ell_n}{\partial\boldsymbol\eta}(r)= \left(\frac{\partial\tilde\ell_n}{\partial\boldsymbol\phi}, \frac{\partial\tilde\ell_n}{\partial\boldsymbol\Psi}(r)\right) =\frac{\partial\ell_n(\boldsymbol\eta,r)}{\partial\boldsymbol\eta}\right|_ {\boldsymbol\eta=\tilde{\bEta}}$$
i.e. the score function, evaluated at the restricted estimators. The supremum Lagrange multipliers test statistic (hereafter supLM) is
  \begin{align}
  T_n&=\sup_{r\in[r_L,r_U]}T_{n}(r)\label{eq:Tn},\\
  \text{where }&\nonumber\\
  T_{n}(r)&=\left(\frac{\partial \tilde{\ell}_n }{\partial \boldsymbol\Psi}(r)\right)^\intercal \left({I}_{n,22}(r)-{I}_{n,21}(r){I}_{n,11}^{-1} {I}_{n,12}(r)\right)^{-1}\frac{\partial \tilde{\ell}_n }{\partial \boldsymbol\Psi}(r)\label{eq:Tnr}
  \end{align}
with $[r_L,r_U]$ being a data driven interval, e.g. $r_L$ and $r_U$ can be some percentiles of the observed data. Define the information matrix as follows:
\begin{equation}\label{eq:information_fun}
  I_n(r)=\begin{pmatrix}
     I_{n,11} & I_{n,12}(r) \\
     I_{n,21}(r) & I_{n,22}(r)
   \end{pmatrix}:=\begin{pmatrix}
                   -\frac{\partial^2\ell_n(\boldsymbol\eta,r)}{\partial\boldsymbol\phi\partial\boldsymbol\phi^\intercal} & -\frac{\partial^2\ell_n(\boldsymbol\eta,r)}{\partial\boldsymbol\phi\partial\boldsymbol\Psi^\intercal} \\
                   -\frac{\partial^2\ell_n(\boldsymbol\eta,r)}{\partial\boldsymbol\Psi\partial\boldsymbol\phi^\intercal} & -\frac{\partial^2\ell_n(\boldsymbol\eta,r)}{\partial\boldsymbol\Psi\partial\boldsymbol\Psi^\intercal}
                 \end{pmatrix},
\end{equation}
and
  \begin{equation}\label{eq:information_asym}
    I_\infty(r)=\begin{pmatrix}
     I_{\infty,11} & I_{\infty,12}(r) \\
     I_{\infty,21}(r) & I_{\infty,22}(r)
   \end{pmatrix}
  \end{equation}
  where $I_{\infty,22}(r)=I_{\infty,12}(r)=I^\intercal_{\infty,21}(r)$ are $(p+1)\times (p+1)$ symmetric matrices whose $(i+1,j+1)$th element is

  \begin{align*}
    &E[I(X_{t-d}\leq r)], & \mbox{if } i=0,j=0 \\
    &E[X_{t-j}I(X_{t-d}\leq r)], & \mbox{if } i=0,j\neq0 \\
    &E[X_{t-i}X_{t-j}I(X_{t-d}\leq r)], & \mbox{if } i\neq0,j\neq0
  \end{align*}
  and $I_{\infty,11}=I_{\infty,22}(\infty)$. Here and in the following, $P(\cdot)$ and $E[\cdot]$ are, respectively,  the probability and expectation taken under the true probability distribution for which the null hypothesis holds. As in \citet{Cha90a}, the null distribution of the supLM test statistic is a functional of the centered Gaussian process $\left\{\xi(r),\;r\in\mathds{R}\right\}$ with covariance kernel
  $$\Sigma(r_1,r_2)=\sigma_0^{-2}\left\{I_{\infty,22}(r_1\wedge r_2)-I_{\infty,21}(r_1)I_{\infty,11}^{-1}I_{\infty,12}(r_2)\right\},$$
  \noindent
   where $a_1\wedge a_2=\min\{a_1,a_2\}$, for any $a_1,a_2\in\mathds{\mathds{R}}$.
   Under standard regularity conditions as in \cite{Cha90a}, it holds that
   \begin{equation}\label{eq:Tinf}
   T_n\xrightarrow[n\to\infty]{w}\sup_{r\in[r_L,r_U]}\xi(r)^\intercal\Sigma(r,r)^{-1}\xi(r):= T_\infty,
   \end{equation}
\noindent
where $\xrightarrow[n\to\infty]{w}$ means the convergence in distribution as the sample size $n$ increases.

\section{The bootstrap}\label{sec:boot}
We focus on the following residual-based bootstrap approach. Let $\{\eps_t^{*}\}$ be sampled with replacement from the re-centred residuals $\tilde\eps_t^c:=\tilde\eps_t-n^{-1}\sum_{t=1}^{n}\tilde\eps_t$, where $\tilde\eps_t$ are defined in Eq.~(\ref{eps_tilde}). Consider the recursively defined bootstrap process generated by the bootstrap parameters $\boldsymbol\phi^*=(\phi_0^*,\phi_1^*,\dots,\phi_p^*)^\intercal$:
\begin{equation}\label{eq:boot_proc_gen}
  {X}_t^{*}=\phi^*_{0}+\sum_{i=1}^{p}\phi^*_{i}X^{*}_{t-i}+\eps_t^{*},
\end{equation}
where the initial values $X^*_{0},X^*_{1},\dots,X^*_{-p+1},$ are equal to the sample counterpart. We consider the case where the bootstrap parameters are the restricted MLE, i.e. $\boldsymbol\phi^*=\tilde{\boldsymbol\phi}$; therefore the process defined in Eq.~(\ref{eq:boot_proc_gen}) equals:
\begin{equation}\label{eq:boot_proc}
  {X}_t^{*}=\tilde\phi_{0}+\sum_{i=1}^{p}\tilde\phi_{i}X^{*}_{t-i}+\eps_t^{*},
\end{equation}
\noindent
 which is an example of the so-called restricted bootstrap, see \citet{Cav21b}. Given the bootstrap sample in Eq.~(\ref{eq:boot_proc}), $\{X_t^*,t=1,\dots,n\}$, the bootstrap log-likelihood function results:
\begin{equation}\label{eq:loglikb}
\ell_n^*(\boldsymbol\eta,r)= -\frac{1}{2\sigma^2}\sum_{t=1}^n \eps_t^{*2}(\boldsymbol\eta,r),
\end{equation}
where $\eps_t^{*}(\boldsymbol\eta,r)$ is defined as in Eq.~(\ref{eq:eps}) with $X$ being replaced by $X^*$:
\begin{align}\label{eq:eps_boot}
\eps^*_t(\boldsymbol\eta,r)=X^*_t&-\left\{\phi_{0}+\sum_{i=1}^{p}\phi_{i}X^*_{t-i}\right\} -\left\{\Psi_{0}+\sum_{i=1}^{p}\Psi_{i}X^*_{t-i}\right\}I\left(X^*_{t-d}\leq r\right).
\end{align}
\noindent
Moreover, let $D_{t}^*(r)$ denote the first-order derivative of $\eps^*_{t+1}(\boldsymbol\eta,r)$ with respect to $\boldsymbol\eta$. It follows that:
\begin{align}\label{eq:epsder}
 D_{t}^*(r)&= \left(-1,-X_t^*,\dots,-X_{t-p+1}^*,-I(X^*_{t-d+1}\leq r),\right.\nonumber\\
 &\quad\left.-X_t^*I(X^*_{t-d+1}\leq r),\dots,-X_{t-p+1}^*I(X^*_{t-d+1}\leq r)\right)^\intercal.
\end{align}
Similar to Eq.~(\ref{eq:information_fun}), the bootstrap observed information matrix is defined by:
\begin{align}
  I_n^*(r)&=\begin{pmatrix}
     I^*_{n,11}    & I^*_{n,12}(r) \\
     I^*_{n,21}(r) & I^*_{n,22}(r)
   \end{pmatrix}\nonumber\\
   &=\begin{pmatrix}
                   -\frac{\partial^2\ell_n^*(\boldsymbol\eta,r)}{\partial\boldsymbol\phi\partial\boldsymbol\phi^\intercal} & -\frac{\partial^2\ell_n^*(\boldsymbol\eta,r)}{\partial\boldsymbol\phi\partial\boldsymbol\Psi^\intercal} \\
                   -\frac{\partial^2\ell_n^*(\boldsymbol\eta,r)}{\partial\boldsymbol\Psi\partial\boldsymbol\phi^\intercal} & -\frac{\partial^2\ell_n^*(\boldsymbol\eta,r)}{\partial\boldsymbol\Psi\partial\boldsymbol\Psi^\intercal}
                 \end{pmatrix}=\frac{1}{\sigma^{*2}}\sum_{t=1}^{n}D^*_{t-1}(r)D^{*\intercal}_{t-1}(r).\label{eq:inf_matrix_boot}
\end{align}
 Let $\tilde{\boldsymbol\phi}^*=(\tilde{\phi}_0^*,\tilde{\phi}_1^*,\dots,\tilde{\phi}_p^*)$ be the MLE computed upon $\{X_t^*,t=1,\dots,n\}$ defined in Eq.~(\ref{eq:boot_proc_gen}) and $\tilde\sigma^{*2}=(n-p-1)^{-1}\sum_{t=1}^{n-p-1}\tilde\eps_t^{*2}$, with $\tilde\eps_t^{*2}$ being the corresponding bootstrap restricted residuals. In analogy with standard asymptotic theory, we define $\tilde\bEta^*=(\tilde{\boldsymbol\phi}^{*\intercal}, \boldsymbol0^\intercal,\tilde\sigma^{*2})$ to be the bootstrap estimator maximising the bootstrap loglikelihood function in Eq.~(\ref{eq:loglikb}). Let
\begin{equation}\label{eq:boot_score}
\left.\frac{\partial\ell_n^*}{\partial\boldsymbol\eta}(r) =\frac{\partial\ell_n^*(\boldsymbol\eta,r)}{\partial\boldsymbol\eta}\right|_ {\boldsymbol\eta=\tilde\bEta},\quad
\left.\frac{\partial\tilde{\ell}_n^*}{\partial\boldsymbol\eta}(r) =\frac{\partial\ell_n^*(\boldsymbol\eta,r)}{\partial\boldsymbol\eta}\right|_ {\boldsymbol\eta=\tilde\bEta^*}
\end{equation}
be the bootstrap score function evaluated in $\tilde\bEta$ and $\tilde\bEta^*$, respectively.
The partial derivative of $\ell^*(\boldsymbol\eta,r)$ with respect to $\boldsymbol\phi$ does not depend upon $r$ and, as before, we partition $\frac{\partial\ell_n^*}{\partial\boldsymbol\eta}(r)$ and $\frac{\partial\tilde{\ell}_n^*}{\partial\boldsymbol\eta}(r)$ according to $\boldsymbol\phi$ and $\boldsymbol\Psi$:
\begin{equation}\label{eq:boot_score2}
\frac{\partial\ell_n^*}{\partial\boldsymbol\eta}(r)=\left(\frac{\partial\ell_n^*}{\partial\boldsymbol\phi}, \frac{\partial\ell_n^*}{\partial\boldsymbol\Psi}(r)\right),\qquad \frac{\partial\tilde{\ell}_n^*}{\partial\boldsymbol\eta}(r)=\left(\frac{\partial\tilde{\ell}_n^*}{\partial\boldsymbol\phi}, \frac{\partial\tilde{\ell}_n^*}{\partial\boldsymbol\Psi}(r)\right).
\end{equation}
Let
\begin{equation}\label{eq:boot_score_process}
 \left\{\frac{\partial\ell_n^*}{\partial\boldsymbol\eta}(r)\right\}= \left\{\frac{\partial\ell_n^*}{\partial\boldsymbol\eta}(r), r_L\leq r\leq  r_U\right\}
\end{equation}
be the bootstrap score process as a function of $r$, evaluated in $\boldsymbol\eta=\tilde\bEta$. We compute the bootstrap supLM statistic $T_n^{*}$ as:
 \begin{align}
  T^*_n&=\sup_{r\in[r_L,r_U]}T^*_{n}(r);\label{eq:Tn_boot}\\
 T^*_{n}(r)&=\left(\frac{\partial \tilde{\ell}_n^* }{\partial \boldsymbol\Psi}(r)\right)^\intercal \left({I}^*_{n,22}(r)-{I}^*_{n,21}(r)({I}^*_{n,11})^{-1} {I}^*_{n,12}(r)\right)^{-1}\frac{\partial \tilde{\ell}_n^* }{\partial \boldsymbol\Psi}(r).\label{eq:Tnr_boot}
  \end{align}
   Finally,  the bootstrap $p$-value is given by $$B^{-1} \sum_{b=1}^{B} I(T^{*b}_n \geq T_n),$$
    where $T_n^{*b},\; b = 1,\dots B$ is the bootstrap test statistics and $T_n$ is the value of the supLM statistic computed on the original sample, defined in Eq.~(\ref{eq:Tn}).


\subsection{Bootstrap asymptotic theory}\label{sec:asyb}

In order to derive the bootstrap asymptotic theory we rely on the following assumption, which is customary in this setting.
\begin{assumption}\label{ass}
 $\{\eps_t\}$ is a sequence of independent and identically distributed (hereafter iid) random variables with $E[\eps_t]=0$, $E[\eps_t^2]=\sigma^2<\infty$ and $E[\eps_t^4]=\kappa<\infty$; $\{X_t\}$ is stationary and ergodic under the null hypothesis.
\end{assumption}
\noindent
Under Assumption~\ref{ass}, in Theorem~\ref{thm:boot} we prove  that $T^{*}_n$ converges weakly in probability to $T_\infty$, namely, the proposed bootstrap is valid. To this aim, in Proposition~\ref{prop:UBLLN_TARp} we derive a \textit{new} uniform bootstrap law of large numbers (hereafter UBLLN) that allows us to $(i)$ verify that $n^{-1}I^*_n(r)$ converges in probability (in probability) to $I_\infty(r)$ uniformly on $r$ (Proposition~\ref{prop:matrix_TARp}) and $(ii)$  derive an approximation of $\partial\tilde{\ell}_n^*/\partial \boldsymbol\Psi(r)$ in terms of $\partial{\ell}_n^*/\partial \boldsymbol\eta(r)$ (Proposition~\ref{prop:score_TARp}). We next state the UBLLN in Proposition~\ref{prop:UBLLN_TARp} that establishes a new result which is of independent interest since it is the first proof of the validity of the bootstrap when testing for a regime switching mechanism where a nuisance parameter is absent under the null hypothesis. The main difficulty here resides in the indicator function $I(y\leq r)$ being not differentiable. Hence, standard methods based upon Taylor's expansion cannot be applied. Notice that in \cite{Han96}, the problem is circumvented by adopting a stochastic permutation of the score vector for which no UBLLN is required. Our proof of the bootstrap validity also extends the approach of \cite{Cha20}. We approximate the step function with a parameterised sequence of continuous and differentiable functions.
\begin{remark}
  The results provide a general theoretical framework that can be adapted to other kinds of nonlinear processes such as regime-switching processes with exogenous threshold variables or testing for structural breaks.
\end{remark}
\begin{proposition}\label{prop:UBLLN_TARp}
 \textbf{(UBLLN)} Let $\{X_t\}$ and $\{X_t^*\}$ be defined in Eq.~(\ref{eq:AR}) and Eq.~(\ref{eq:boot_proc}), respectively. Under Assumption~\ref{ass}, it holds that:
   \begin{enumerate}
     \item If $E[|X_t|^u]<\infty$, for $u\geq 0$, then:
     \begin{equation}\label{eq:UBLLN_TAR1}
       \sup_{r \in [r_{L},r_{U}]}\left| \frac{1}{n}\sum_{t=1}^{n}X^{*u}_tI(X^*_{t}\leq r)-E[X^u_tI(X_{t}\leq r)]\right|\xrightarrow[n\to\infty]{p^{*}}_p0.
     \end{equation}
     \item If $E[|X_t|^u]<\infty$, for $u=1,2$, then, for every $i,j,d$:
     \begin{equation}\label{eq:UBLLN_TAR2}
     \hspace*{-15pt}
        \sup_{r \in [r_{L},r_{U}]}\left|\frac{1}{n}\sum_{t=1}^{n}X^{*}_{t-i}X^{*}_{t-j}I(X^*_{t-d}\leq r)-E[X_{t-i}X_{t-j}I(X_{t-d}\leq r)]\right|\xrightarrow[n\to\infty]{p^{*}}_p0.
     \end{equation}
   \end{enumerate}
\end{proposition}
\begin{remark}
  Under the null hypothesis and Assumption~\ref{ass}, $E[\eps_t]=0$ and $E[\eps_t]=\sigma^2 < \infty$ imply that $E[|X_t|^u]<\infty$, for $u=1,2$.
\end{remark}
\begin{proposition}\label{prop:matrix_TARp}
Let $\{X_t\}$ and $\{X_t^*\}$ be defined in (\ref{eq:AR}) and Eq.~(\ref{eq:boot_proc}), respectively. Under the null hypothesis and Assumption~\ref{ass}, it holds
 $$\sup_{r \in [r_{L},r_{U}]}\left|\frac{1}{n}\frac{1}{\sigma^{*2}}\sum_{t=1}^{n}D^*_{t-1}(r)D^{*\intercal}_{t-1}(r)-I_\infty(r)\right|\xrightarrow[n\to\infty]{p^{*}}_p0,$$
 with $D^*_t(r)$ and $I_\infty(r)$ being defined in Eq.~(\ref{eq:epsder}) and Eq.~(\ref{eq:information_asym}), respectively.
\end{proposition}
\begin{proposition}\label{prop:score_TARp}
Let $\{X_t\}$ and $\{X_t^*\}$ be defined in Eq.~(\ref{eq:AR}) and Eq.~(\ref{eq:boot_proc}), respectively. Under the null hypothesis, it holds that the bootstrap score defined in Eq.~(\ref{eq:boot_score2}) satisfies:
$$\frac{\partial\tilde\ell_n^*}{\partial\boldsymbol{\Psi}}(r)= \frac{\partial\ell_n^*}{\partial\boldsymbol{\Psi}}(r)- I_{n,21}(r)I_{n,11}^{-1}\frac{\partial\ell_n^*}{\partial\boldsymbol{\phi}}.$$
\end{proposition}
\begin{remark}
By analogy with standard, non-bootstrap, asymptotics \citep{Cha90a,Lin05,Gor21}, thanks to Proposition~\ref{prop:score_TARp} the asymptotic null distribution of $T^{*}_n$ is predominantly determined by the asymptotic behaviour of $\left\{\frac{\partial\ell_n^*}{\partial\boldsymbol\eta}(r)\right\}$ rather than $\left\{\frac{\partial\tilde\ell_n^*}{\partial\boldsymbol\eta}(r)\right\}$ defined in Eq.~(\ref{eq:boot_score}) and this simplifies substantially the derivations.
\end{remark}
\noindent
Next, in Proposition~\ref{prop:BCLT_TARp} we prove a bootstrap central limit theorem (hereafter BCLT) for $\left\{\frac{\partial\ell_n^*}{\partial\boldsymbol\eta}(r)\right\}$, the bootstrap score process defined in Eq.~(\ref{eq:boot_score_process}).
\begin{proposition}\label{prop:BCLT_TARp}
 \textbf{(BCLT)} Under the null hypothesis and Assumption~\ref{ass}, for any fixed $r$, it holds that
  $$\frac{1}{\sqrt{n}}\frac{\partial\ell_n^*}{\partial\boldsymbol\eta}(r)\xrightarrow[n\to\infty]{w^*}_p \mathcal{Z}(r),$$
where $\frac{\partial\ell_n^*}{\partial\boldsymbol\eta}(r)$ is defined in Eq.~(\ref{eq:boot_score}), $\mathcal{Z}(r)$ is a Gaussian distributed $2(p+1)$-dimensional random vector with zero-mean and variance-covariance matrix equal to $\sigma_0^{-2}I_\infty(r)$, defined in Eq.~(\ref{eq:information_asym}).
\end{proposition}
The next theorem contains the main result, namely, the bootstrap functional central limit theorem (BFCLT), where we prove that the conditional asymptotic null distribution of the bootstrap test statistic $T_n^*$ is the same of the unconditional asymptotic null distribution of the non-bootstrap test statistic $T_n$. This guarantees  the validity of the proposed bootstrap method.
\begin{theorem}\label{thm:boot}
\textbf{(BFCLT)}
Let $T^{*}_n$ be the supLM statistic defined in Eq.~(\ref{eq:Tn_boot}). Under the null hypothesis and Assumption~\ref{ass}, it holds that
$T^{*}_n\xrightarrow[n\to\infty]{w^{*}}_p T_\infty,$
with $T_\infty$ being defined in Eq.~(\ref{eq:Tinf}).
\end{theorem}

\section{Finite sample performance}\label{sec:MCstudy}
In this section we investigate the finite sample performance of the bootstrap sLM test and compare it with the asymptotic counterpart for series whose length is $n=50,100,200$. These are small to moderate sample sizes that are quite common in many fields, especially when the cost of producing the data is not negligible. Hereafter $\eps_t$, $t=1,\dots,n$ is generated from a standard Gaussian white noise, the nominal size is $\alpha=5\%$ and the number of Monte Carlo replications is 1000. For the asymptotic tests we have used the tabulated values of \cite{And03}, whereas the bootstrap $p$-values are based on $B=1000$ resamples. The threshold is searched from percentile 25th to 75th of the sample distribution. In Section~\ref{sec:size} and Section~\ref{sec:power} we study the size and the power of the tests. Then, in Section~\ref{sec:ordsel} we assess the behaviour of the tests when the order of the AR process tested is treated as unknown and is selected through the AIC. All the results are presented as percentages as to enhance the readability of the tables.
\subsection{Empirical size of the tests}\label{sec:size}
To study the size of the tests, we generate time series from 21 different simulation settings of the following AR$(1)$ model:

\begin{equation}\label{eq:AR1}
  X_t= \phi_{0} + \phi_{1} X_{t-1}+\eps_t
\end{equation}
where $\phi_{0}=-1,0,+1$ and $\phi_{1}=0,\pm 0.3, \pm 0.6, \pm 0.9$. Table~\ref{tab:s1} shows the rejection percentages for the three sample sizes in use. As expected, the intercept $\phi_{0}$ has no impact upon the size of the tests and the variability observed reflects the joint sampling and simulation fluctuation. Our bootstrap test sLMb has a good size even for a sample size as small as 50 and is not influenced by the value of the autoregressive parameter close to non-stationarity. On the contrary, the asymptotic test results severely oversized as $\phi_1$ approaches unity and the bias persists for $n=200$.

\begin{table}
  \centering
\begin{tabular}{rrrrrrrr}
\multicolumn{2}{c}{ } & \multicolumn{2}{c}{$n=50$} & \multicolumn{2}{c}{$n=100$} & \multicolumn{2}{c}{$n=200$} \\
\cmidrule(l{3pt}r{3pt}){3-4} \cmidrule(l{3pt}r{3pt}){5-6} \cmidrule(l{3pt}r{3pt}){7-8}
$\phi_{0}$ & $\phi_{1}$ & sLMa & sLMb & sLMa & sLMb & sLMa & sLMb\\
\cmidrule(lr){3-8}
-1 & -0.9 & 14.7 & 4.2 & 9.6 & 4.8 & 7.4 & 6.1\\
-1 & -0.6 & 4.6 & 4.9 & 4.6 & 5.5 & 4.6 & 6.0\\
-1 & -0.3 & 4.6 & 5.7 & 2.8 & 3.8 & 4.0 & 5.3\\
-1 & 0.0 & 4.4 & 4.3 & 4.3 & 5.1 & 5.0 & 5.3\\
-1 & 0.3 & 7.3 & 5.4 & 4.4 & 4.3 & 4.6 & 5.0\\
-1 & 0.6 & 16.6 & 6.9 & 8.5 & 5.3 & 6.5 & 5.6\\
-1 & 0.9 & 42.2 & 7.5 & 30.9 & 6.2 & 21.8 & 7.3\\
\addlinespace
0 & -0.9 & 17.0 & 5.5 & 9.0 & 3.8 & 6.9 & 4.4\\
0 & -0.6 & 4.4 & 5.2 & 3.7 & 4.8 & 3.1 & 4.3\\
0 & -0.3 & 3.6 & 4.6 & 3.7 & 4.7 & 4.7 & 5.9\\
0 & 0.0 & 4.4 & 5.1 & 4.3 & 5.1 & 4.2 & 5.2\\
0 & 0.3 & 6.6 & 4.6 & 3.4 & 2.7 & 4.8 & 5.3\\
0 & 0.6 & 13.9 & 5.6 & 8.5 & 4.6 & 6.0 & 5.1\\
0 & 0.9 & 42.6 & 6.3 & 28.8 & 4.9 & 18.5 & 4.7\\
\addlinespace
1 & -0.9 & 16.4 & 3.5 & 9.7 & 4.0 & 8.6 & 5.7\\
1 & -0.6 & 3.8 & 3.8 & 3.9 & 4.6 & 4.9 & 5.7\\
1 & -0.3 & 2.6 & 3.6 & 3.6 & 4.6 & 5.0 & 6.2\\
1 & 0.0 & 3.5 & 3.8 & 4.2 & 5.6 & 4.2 & 5.1\\
1 & 0.3 & 6.5 & 4.4 & 5.9 & 5.3 & 4.4 & 4.9\\
1 & 0.6 & 14.7 & 4.9 & 8.6 & 5.0 & 5.7 & 4.2\\
1 & 0.9 & 40.9 & 6.6 & 33.1 & 5.0 & 19.9 & 5.6\\
\cmidrule(lr){3-8}
\end{tabular}
  \caption{Empirical size at nominal level $\alpha=5\%$ for the AR(1) process of Eq.~(\ref{eq:AR1}) for the asymptotic test statistic sLMa of Eq.~(\ref{eq:Tn}) and the bootstrap test statistic sLMb of Eq.~(\ref{eq:Tn_boot}).}\label{tab:s1}
\end{table}


\subsection{Empirical power of the tests}\label{sec:power}
%
In this section we study the power of the supLM tests and highlight the differences between them. We simulate from the following TAR(1) model:
\begin{equation}\label{eq:TAR1}
 X_t = \phi_{1,0} + \phi_{1,1} X_{t-1} + \left(\Psi + \Psi X_{t-1}\right)I(X_{t-1}\leq 0) + \eps_t.
\end{equation}
where $\phi_{1,0}$ and $\phi_{1,1}$ are as follows:
\begin{center}
$
  \begin{array}{crr}
    & \phi_{1,0} & \phi_{1,1} \\
\cmidrule(lr){2-3}
\text{M1} & -0.1&-0.8\\
\text{M2} &  0.8&-0.2\\
\cmidrule(lr){2-3}
  \end{array}
$
\end{center}
and $\Psi = (0.0, 0.3, 0.6, 0.9)$ as to obtain 8 different parameter settings. Note that the parameter $\Psi$ represents the departure from the null hypothesis and in all the simulations below we take sequences of increasing distance from $H_0$ in all of its components. The case $\Psi=0$ corresponds to $H_0$. Table~\ref{tab:p1c} presents the size corrected power at nominal level $5\%$ where the first and the fifth rows correspond to the null hypothesis and reflect the size of the tests, while the subsequent three rows represent increasing departures from $H_0$ and reflect the power of the tests. Small deviations from the 5\% nominal level are due to discretization effects when size-correcting the bootstrap $p$-values. The table shows clearly that the bootstrap sLM test has superior power. As the sample size increases the power of the two tests is similar but the bootstrap version has always a small margin of advantage.
The uncorrected power is reported in Table~\ref{tab:p1} of the Supplement and shows that our bootstrap test has correct size for the first row and $n=50$ where the TAR(1) model reduces to a AR(1) model with parameter -0.8 and the asymptotic test shows some oversize.

\begin{table}

  \centering
\begin{tabular}{crrrrrrr}
&\multicolumn{1}{c}{ } & \multicolumn{2}{c}{$n=50$} & \multicolumn{2}{c}{$n=100$} & \multicolumn{2}{c}{$n=200$} \\
\cmidrule(l{3pt}r{3pt}){3-4} \cmidrule(l{3pt}r{3pt}){5-6} \cmidrule(l{3pt}r{3pt}){7-8}
&$\Psi$ & sLMa & sLMb & sLMa & sLMb & sLMa & sLMb\\
\cmidrule(l{3pt}r{3pt}){3-4} \cmidrule(l{3pt}r{3pt}){5-6} \cmidrule(l{3pt}r{3pt}){7-8}
\text{M1}& 0.0 & 5.0 & 5.0 & 5.0 & 4.9 & 5.0 & 5.0\\
& 0.3 & 6.1 & 9.5 & 11.2 & 15.0 & 41.0 & 43.1\\
& 0.6 & 15.7 & 28.2 & 51.3 & 58.2 & 93.8 & 93.8\\
& 0.9 & 37.4 & 55.5 & 88.2 & 91.0 & 100.0 & 100.0\\
\addlinespace
\text{M2}& 0.0 & 5.0 & 5.0 & 5.0 & 5.0 & 5.0 & 5.0\\
& 0.3 & 6.3 & 6.0 & 7.5 & 7.8 & 8.6 & 8.6\\
& 0.6 & 9.5 & 9.2 & 14.9 & 14.5 & 36.3 & 36.5\\
& 0.9 & 14.0 & 14.1 & 33.4 & 33.5 & 67.7 & 67.7\\
\cmidrule(l{3pt}r{3pt}){3-4} \cmidrule(l{3pt}r{3pt}){5-6} \cmidrule(l{3pt}r{3pt}){7-8}
\end{tabular}
\caption{Size corrected power at nominal level $\alpha=5\%$ for the TAR(1) process of Eq.~(\ref{eq:TAR1}) for the asymptotic test statistic sLMa of Eq.~(\ref{eq:Tn}) and the bootstrap test statistic sLMb of Eq.~(\ref{eq:Tn_boot}).}\label{tab:p1c}
\end{table}

\subsection{The impact of order selection}\label{sec:ordsel}

In practical situations, the order of the autoregressive model to be tested is unknown and has to be estimated. This can impinge on the performance of the tests so that we assess the impact of treating the order $p$ of the AR model as unknown and selecting it by means of the AIC.
We study the impact on the size of the tests by simulating from the following AR(2) model.
\begin{equation}\label{eq:AR2}
  X_t= \phi_{0} + \phi_{1} X_{t-1}+\phi_{2} X_{t-2} + \eps_t
\end{equation}
where $\phi_{0}=0$, whereas $\phi_{1}$ and $\phi_{2}$ are presented on the first two columns of Table~\ref{tab:s2} that shows the rejection percentages when using the true order of the autoregression in the tests. The results confirm that the bootstrap test has correct size also for $n=50$ irrespective of the parameters, while the asymptotic test is biased when the parameters are close to the non-stationary region. Table~\ref{tab:s2OS} is as Table~\ref{tab:s2} but in such a case the order of the autoregression is treated as unknown and selected through the AIC. This produces a noticeable oversize in the asymptotic test but has no effects upon the bootstrap version of the sLM test, no matter the sample size.

\begin{table}
  \centering
\begin{tabular}{rrrrrrrr}
\multicolumn{2}{c}{ } & \multicolumn{2}{c}{$n=50$} & \multicolumn{2}{c}{$n=100$} & \multicolumn{2}{c}{$n=200$} \\
\cmidrule(l{3pt}r{3pt}){3-4} \cmidrule(l{3pt}r{3pt}){5-6} \cmidrule(l{3pt}r{3pt}){7-8}
$\phi_{1}$ & $\phi_{2}$ & sLMa & sLMb & sLMa & sLMb & sLMa & sLMb\\
\cmidrule(l{3pt}r{3pt}){3-4} \cmidrule(l{3pt}r{3pt}){5-6} \cmidrule(l{3pt}r{3pt}){7-8}
-0.65 & 0.25 & 25.2 & 5.0 & 15.3 & 4.0 & 10.1 & 4.7\\
-0.95 & -0.25 & 8.1 & 4.9 & 6.1 & 6.0 & 5.2 & 5.0\\
-0.35 & -0.45 & 3.6 & 3.7 & 3.6 & 4.4 & 4.6 & 5.2\\
1.15 & -0.55 & 9.2 & 4.4 & 6.0 & 4.6 & 5.9 & 4.7\\
0.45 & 0.25 & 21.8 & 4.7 & 14.6 & 5.4 & 10.5 & 5.4\\
0.45 & -0.55 & 5.8 & 4.9 & 4.6 & 6.1 & 5.6 & 6.4\\
-0.90 & -0.25 & 6.6 & 5.2 & 4.6 & 5.2 & 5.5 & 6.2\\
\cmidrule(l{3pt}r{3pt}){3-4} \cmidrule(l{3pt}r{3pt}){5-6} \cmidrule(l{3pt}r{3pt}){7-8}
\end{tabular}
  \caption{Empirical size at nominal level $\alpha=5\%$ for the AR(2) process of Eq.~(\ref{eq:AR2}) for the asymptotic test statistic sLMa of Eq.~(\ref{eq:Tn}) and the bootstrap test statistic sLMb of Eq.~(\ref{eq:Tn_boot}). Here the true order of the autoregression is used.}\label{tab:s2}
\end{table}

\begin{table}
  \centering
\begin{tabular}{rrrrrrrr}
\multicolumn{2}{c}{ } & \multicolumn{2}{c}{$n=50$} & \multicolumn{2}{c}{$n=100$} & \multicolumn{2}{c}{$n=200$} \\
\cmidrule(l{3pt}r{3pt}){3-4} \cmidrule(l{3pt}r{3pt}){5-6} \cmidrule(l{3pt}r{3pt}){7-8}
$\phi_{1}$ & $\phi_{2}$ & sLMa & sLMb & sLMa & sLMb & sLMa & sLMb\\
\cmidrule(l{3pt}r{3pt}){3-4} \cmidrule(l{3pt}r{3pt}){5-6} \cmidrule(l{3pt}r{3pt}){7-8}
-0.65 & 0.25 & 22.5 & 4.2 & 18.3 & 4.3 & 15.2 & 5.0\\
-0.95 & -0.25 & 11.8 & 4.3 & 13.0 & 5.7 & 11.3 & 4.4\\
-0.35 & -0.45 & 10.5 & 4.0 & 10.8 & 4.0 & 10.4 & 4.8\\
1.15 & -0.55 & 16.7 & 5.0 & 13.2 & 4.6 & 13.8 & 5.6\\
0.45 & 0.25 & 20.0 & 5.3 & 18.8 & 5.0 & 17.8 & 5.7\\
0.45 & -0.55 & 14.0 & 5.0 & 10.5 & 6.2 & 12.8 & 6.1\\
-0.90 & -0.25 & 12.1 & 4.3 & 11.2 & 5.3 & 11.7 & 4.7\\
\cmidrule(l{3pt}r{3pt}){3-4} \cmidrule(l{3pt}r{3pt}){5-6} \cmidrule(l{3pt}r{3pt}){7-8}
\end{tabular}
  \caption{As Table~\ref{tab:s2} but here the order of the autoregression has been treated as unknown and selected through the AIC.}\label{tab:s2OS}
\end{table}
We study the impact of model selection upon the power of the tests by simulating from the following TAR(2) process:
\begin{equation}\label{eq:TAR2}
 X_t = \phi_{1,0} + \phi_{1,1} X_{t-1} + \phi_{1,2} X_{t-2} + \left(\Psi + \Psi X_{t-1}+ \Psi X_{t-2}\right)I(X_{t-1}\leq 0) + \eps_t.
\end{equation}
where $\phi_{1,0}=0$, $\phi_{1,1}=-0.35$, $\phi_{1,1}=-0.45$ and, as before, $\Psi = (0.0, 0.2, 0.6, 0.8)$ represents the level of departure from $H_0$.  The rejection percentages are shown in Table~\ref{tab:p2}. Here, for $n=50$ the asymptotic test is more powerful than the bootstrap version, whereas the power of the two tests is very similar for $n=100, 200$. Also the size is similar and close to the nominal 5\% level so that the size corrected power reported in the lower panel of the table is very similar to the uncorrected power. Table~\ref{tab:p2OS} reports the empirical power (upper panel) and its size corrected version (lower panel) when the order of the tested model is selected by means of the AIC. As before, this produces an oversize in the asymptotic test and the size-corrected power confirms the superiority of our bootstrap test for small to moderate sample sizes.

\begin{table}
  \centering
\begin{tabular}{rrrrrrr}
\multicolumn{1}{c}{ } & \multicolumn{2}{c}{$n=50$} & \multicolumn{2}{c}{$n=100$} & \multicolumn{2}{c}{$n=200$} \\
\cmidrule(l{3pt}r{3pt}){2-3} \cmidrule(l{3pt}r{3pt}){4-5} \cmidrule(l{3pt}r{3pt}){6-7}
$\Psi$ & sLMa & sLMb & sLMa & sLMb & sLMa & sLMb\\
\cmidrule(l{3pt}r{3pt}){2-3} \cmidrule(l{3pt}r{3pt}){4-5} \cmidrule(l{3pt}r{3pt}){6-7}
 0.0 &  4.2 &  4.8 &  3.8 &  4.8 &   3.9 &   4.5 \\
 0.2 &  5.3 &  5.2 & 10.8 & 11.7 &  20.4 &  21.7 \\
 0.6 & 32.5 & 25.7 & 71.8 & 71.5 &  98.3 &  98.5 \\
 0.8 & 58.6 & 40.6 & 95.5 & 94.7 & 100.0 & 100.0 \\
\cmidrule(l{3pt}r{3pt}){2-3} \cmidrule(l{3pt}r{3pt}){4-5} \cmidrule(l{3pt}r{3pt}){6-7}
\addlinespace
& \multicolumn{6}{c}{size corrected} \\
\addlinespace
$\Psi$ & sLMa & sLMb & sLMa & sLMb & sLMa & sLMb\\
\cmidrule(l{3pt}r{3pt}){2-3} \cmidrule(l{3pt}r{3pt}){4-5} \cmidrule(l{3pt}r{3pt}){6-7}
 0.0 &  5.0 &  4.9 &  5.0 &  5.0 &   5.0 &   5.0 \\
 0.2 &  6.6 &  5.3 & 12.8 & 11.9 &  24.6 &  25.1 \\
 0.6 & 35.3 & 26.1 & 75.9 & 72.8 &  98.9 &  98.7 \\
 0.8 & 60.2 & 41.2 & 96.9 & 95.2 & 100.0 & 100.0 \\
\cmidrule(l{3pt}r{3pt}){2-3} \cmidrule(l{3pt}r{3pt}){4-5} \cmidrule(l{3pt}r{3pt}){6-7}
\end{tabular}
  \caption{Power at nominal level $\alpha=5\%$ for the TAR(2) process of Eq.~(\ref{eq:TAR2}) for the asymptotic test statistic sLMa of Eq.~(\ref{eq:Tn}) and the bootstrap test statistic sLMb of Eq.~(\ref{eq:Tn_boot}). The lower panel reports the size corrected version of the upper panel.}\label{tab:p2}
\end{table}

\begin{table}
  \centering
\begin{tabular}{rrrrrrr}
\multicolumn{1}{c}{ } & \multicolumn{2}{c}{$n=50$} & \multicolumn{2}{c}{$n=100$} & \multicolumn{2}{c}{$n=200$} \\
\cmidrule(l{3pt}r{3pt}){2-3} \cmidrule(l{3pt}r{3pt}){4-5} \cmidrule(l{3pt}r{3pt}){6-7}
$\Psi$ & sLMa & sLMb & sLMa & sLMb & sLMa & sLMb\\
\cmidrule(l{3pt}r{3pt}){2-3} \cmidrule(l{3pt}r{3pt}){4-5} \cmidrule(l{3pt}r{3pt}){6-7}
0.0 & 12.5 &  4.0 & 11.5 &  4.2 & 10.0 & 4.0 \\
0.2 & 13.1 &  5.1 & 19.2 & 10.9 & 28.7 & 20.5\\
0.6 & 25.0 & 18.8 & 55.5 & 50.8 & 90.0 & 90.6\\
0.8 & 30.6 & 23.5 & 59.8 & 64.8 & 91.3 & 94.8\\
\cmidrule(l{3pt}r{3pt}){2-3} \cmidrule(l{3pt}r{3pt}){4-5} \cmidrule(l{3pt}r{3pt}){6-7}
\addlinespace
& \multicolumn{6}{c}{size corrected} \\
\addlinespace
$\Psi$ & sLMa & sLMb & sLMa & sLMb & sLMa & sLMb\\
\cmidrule(l{3pt}r{3pt}){2-3} \cmidrule(l{3pt}r{3pt}){4-5} \cmidrule(l{3pt}r{3pt}){6-7}
0.0 & 5.0  &  5.0 &  5.0 &  4.7 &  5.0 &  5.0\\
0.2 & 5.1  &  6.5 &  9.2 & 11.6 & 16.2 & 23.8\\
0.6 & 6.9  & 22.5 & 38.0 & 52.4 & 83.9 & 92.3\\
0.8 & 10.7 & 26.3 & 46.4 & 66.7 & 82.9 & 95.9\\
\cmidrule(l{3pt}r{3pt}){2-3} \cmidrule(l{3pt}r{3pt}){4-5} \cmidrule(l{3pt}r{3pt}){6-7}
\end{tabular}
  \caption{As Table~\ref{tab:p2} but here the order of the autoregression has been treated as unknown and selected through the AIC.}\label{tab:p2OS}
\end{table}

\section{An application: the effect of warming on populations of larvae}\label{sec:applic}
In this section we analyse a panel of 12 short experimental time series of populations of \emph{Plodia interpunctella}, a pyralid moth which infests at the global level many different stored food. While it is well known that global warming is one of the consequences of climate change, its effects on insects' populations can have important economic consequences and are still not completely clear. One of the main features of many economically important insects and other animals (e.g. salmons) is the appearance of generation cycles. Typically, these are non-seasonal asymmetric cycles linked to delayed density dependence mediated by competition or diet quality. Generation cycles can also be caused by age-specific interactions between the insect and its enemies or cannibalism phenomena between larvae and eggs and pupas. In many different species of insects, the mechanism of generation cycling is similar to that of \emph{P. interpunctella} so that the latter can be taken as a reference model.
\par
The data come from \cite{Lau19}, where the authors established 18 populations of larvae, reflecting different experimental conditions on temperature (27, 30, 33°C)  and food quality (poor, standard/good). Each of the 6 experimental combinations has been replicated 3 times as to obtain 18 series of population counts followed for 82 weeks. The first 10 weeks have been treated as transient and discarded so that the series tested have 71 observations. Since only the populations at 27 and 30°C persisted for the entire time span, we focus on the 12 series corresponding to these two temperature level. The time plots are shown in Figure~\ref{SMfig:1} of the Supplementary Material.
\par
Table~\ref{tab:L1} shows the results of the application of the sLM tests to the panel of 12 series. The first two columns of the table indicate the experimental conditions, the third column indicates the replication, whereas the fourth column contains the value of the sLM statistic. Finally, the last column reports the bootstrap $p$-value of the tests. The threshold is searched between percentiles 25th-75th of the data and the delay parameter is $d=2$ weeks. Despite the small sample size, the bootstrap test is able to reject in 7 out of the 12 series at 90\% level. By using the asymptotic critical values of \cite{And03}, the test rejects in 6 out of 12 series and this confirms the slightly superior power of the bootstrap test in small samples found in the simulation studies.
\begin{table}
  \centering
\begin{tabular}{rrrrr}
\toprule
temp. & diet & repl. & statistic & p.value\\
\midrule
27 & poor & 1 & 6.23 & 0.254\\
27 & poor & 2 & 18.97 & 0.001\\
27 & poor & 3 & 4.27 & 0.521\\
\addlinespace
27 & good & 1 & 3.49 & 0.619\\
27 & good & 2 & 14.78 & 0.010\\
27 & good & 3 & 6.64 & 0.236\\
\addlinespace
30 & poor & 1 & 7.83 & 0.132\\
30 & poor & 2 & 8.70 & 0.094\\
30 & poor & 3 & 9.17 & 0.076\\
\addlinespace
30 & good & 1 & 11.63 & 0.033\\
30 & good & 2 & 10.78 & 0.042\\
30 & good & 3 & 9.79 & 0.066\\
\bottomrule
\end{tabular}
 \caption{Results of the application of the sLM tests to the time series of 12 populations of larvae of the pyralid moth \emph{P. interpunctella} under different experimental conditions: temperature (first column), quality of the diet (second column). For each combination, there are 3 laboratory replications (third column). The fourth column reports the value of the sLM test statistic, whereas the last column contains the bootstrap $p$-values of our test.}\label{tab:L1}
\end{table}

Next, we fit a threshold model to the 4 time series obtained by averaging over the 3 experimental replications. The time plots of series and the results of the tests are presented in Figure~\ref{fig:2} and Table~\ref{tab:L2}, respectively. Clearly, warming has a noticeable effect on the mean of the series (dashed red lines). Figure~\ref{fig:3} reports the power spectral density of the series. The frequency corresponding to the characteristic asymmetric 6-week generation cycle is evidenced with a red dashed line.

\begin{table}
  \centering
\begin{tabular}{rrrr}
\toprule
temp. & diet & statistic & p.value\\
\midrule
27 & poor & 10.67 & 0.055\\
27 & good & 15.68 & 0.006\\
\addlinespace
30 & poor & 22.13 & 0.000\\
30 & good &  8.02 & 0.120\\
\bottomrule
\end{tabular}
 \caption{Results of the application of the sLM tests to the time series of 4 populations of larvae of the pyralid moth \emph{P. interpunctella} under different experimental conditions: temperature (first column), quality of the diet (second column). The fourth column reports the value of the sLM test statistic, whereas the last column contains the bootstrap $p$-values of our test.}\label{tab:L2}
\end{table}
\begin{figure}
  \centering
\includegraphics[width=0.9\linewidth,keepaspectratio]{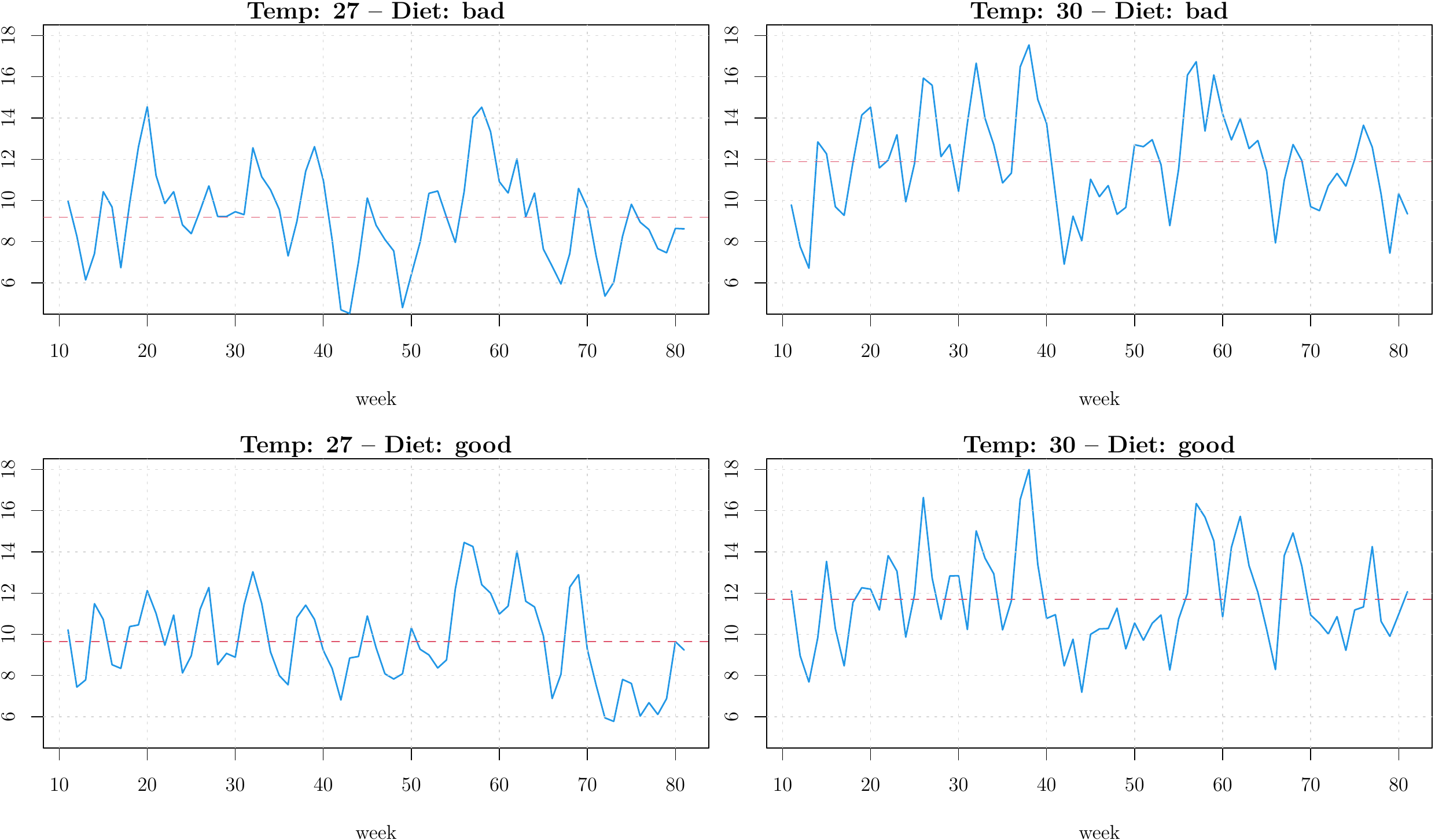}
\caption{Time series of 4 populations of \emph{P. interpunctella} from week 11 to 82 for different experimental conditions. The series have been square-root transformed.}
\label{fig:2}
\end{figure}
\begin{figure}
  \centering
\includegraphics[width=0.6\linewidth,keepaspectratio]{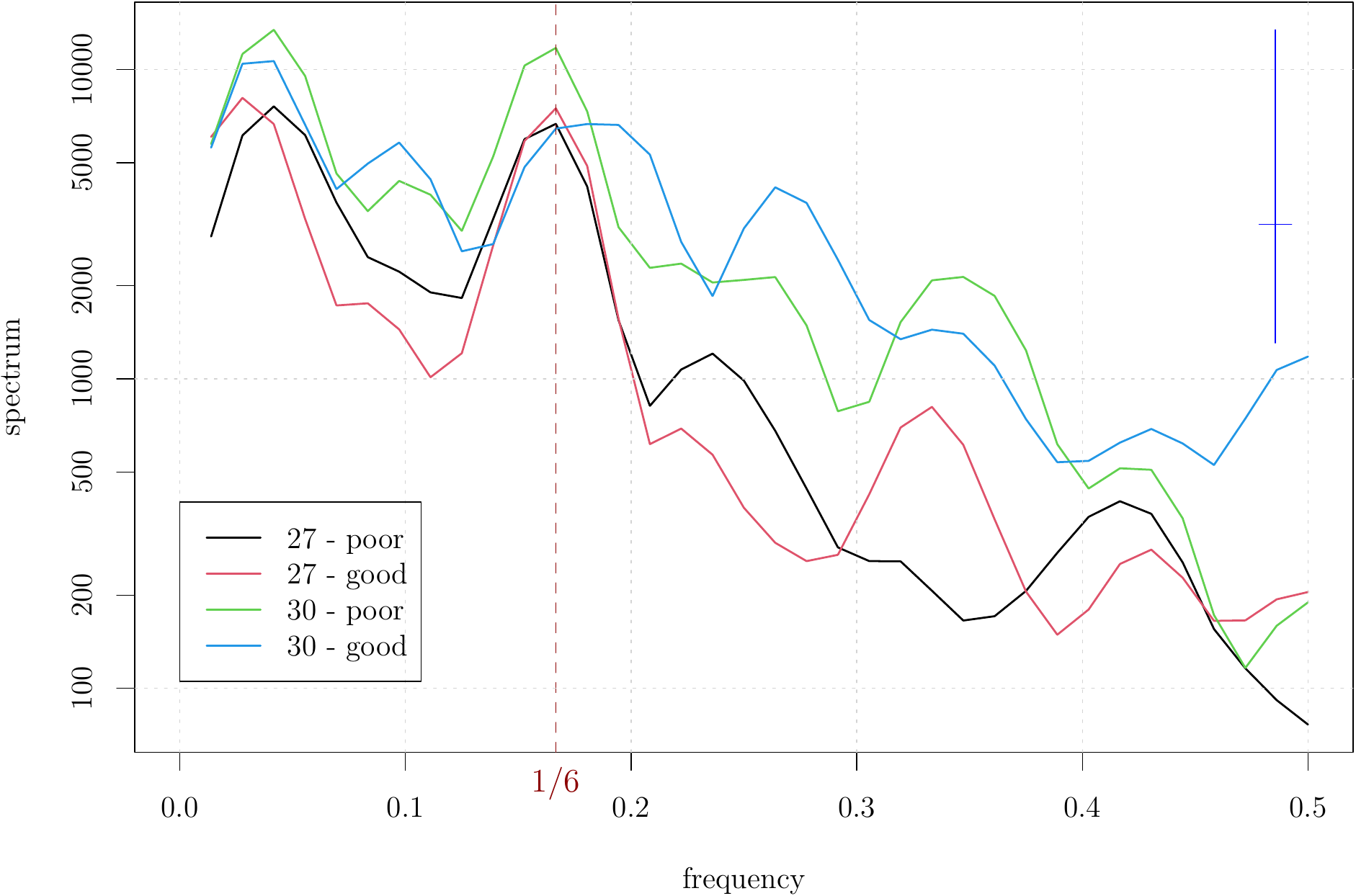}
\caption{Power spectrum of the time series of 4 populations of \emph{P. interpunctella} corresponding to different experimental conditions. The frequency corresponding to the characteristic 6-week cycle is evidenced with a vertical dashed line.}\label{fig:3}
\end{figure}
\noindent
The series are likely to be affected by measurement error so that a threshold ARMA specification is more appropriate than the TAR model (see \cite{Cha21} for the theoretical justification). Typically, the MA parameters greatly enhance the flexibility of the model, while retaining parsimony \cite{Gor20,Gor20a}. Hence, we fit the following TARMA model

\begin{equation}\label{eq:tarma}
  X_t = \begin{cases}
        \phi_{1,0}  + \phi_{1,1} X_{t-1} +  \phi_{1,5} X_{t-5} +  \theta_{1,1} \eps_{t-1} + \eps_t , & \text{ if } X_{t-2}\leq r, \\
        \phi_{2,0}  + \phi_{2,1} X_{t-1} +  \phi_{2,5} X_{t-5} +  \theta_{2,3}\eps_{t-3} + \eps_t, & \text{ otherwise}.
      \end{cases}
\end{equation}

The results are shown in Table~\ref{tab:Lfit}, where the standard error are reported in parenthesis below the estimates. The last column reports the estimated threshold $\hat r$. Due to the length of the series, the standard errors are quite large, still, there are both common features and differences across regimes and for different temperatures. Most importantly, the three-lag specification is consistent with the findings of \cite{Bri00} and manages to reproduce the characteristic 6-week generation cycle, especially for the series corresponding to a temperature of 27°. This is shown in Figure~\ref{fig:4} (left) that shows the power spectral density computed on a series of 100,000 observations simulated from the second fit. The right panel of the figure shows also the histogram of the data with the density of the fitted model, estimated upon the simulated series (blue line). The plots for the other series can be found in the Supplement, Figures~\ref{SMfig:5}--\ref{SMfig:7}. The fourth series (30° - good diet) seems to present different periodicities and the model fit is less satisfactory. This could be an indication of the effect of warming producing a qualitative change in the population dynamics. The diagnostic analysis performed both on the residuals and on the squared residuals of the fitted models does not show any unaccounted dependence, see Figures~\ref{SMfig:diag1}--\ref{SMfig:diag4} of the Supplementary Material. Finally, the Shapiro-Wilk test applied to the residuals (see Table~\ref{SMtab:ntest} of the Supplementary Material) does not show departures from normality, except for the last series whose $p$-value 0.046 somehow confirms that the combined action of warming and diet conditions can alter significantly the population dynamics of larvae.

\begin{table}
  \centering
\begin{tabu}{rrrrrrrrrrr}
temp & diet & $\hat\phi_{1,0}$ & $\hat\phi_{1,1}$ & $\hat\phi_{1,5}$ & $\hat\theta_{1,1}$ & $\hat\phi_{2,0}$ & $\hat\phi_{2,1}$ & $\hat\phi_{2,5}$ & $\hat\theta_{2,3}$ & $\hat r$\\
\cmidrule(lr){3-11}
27 & poor & 1.15 & 0.60 & 0.32 & 0.61 & 0.17 & 0.79 & 0.15 & -0.36 & 9.18\\ \rowfont{\footnotesize}
 &  & (2.38) & (0.21) & (0.15) & (0.21) & (1.37) & (0.13) & (0.09) & (0.18) & \\
\addlinespace
27 & good & 0.73 & 0.69 & 0.31 & 0.24 & 0.24 & 0.72 & 0.19 & -0.06 & 8.89\\ \rowfont{\footnotesize}
 &  & (2.47) & (0.30) & (0.14) & (0.41) & (1.53) & (0.12) & (0.11) & (0.17) & \\
\addlinespace
30 & poor & 3.39 & 0.58 & 0.22 & 0.32 & 1.29 & 0.66 & 0.11 & 0.40 & 11.93\\ \rowfont{\footnotesize}
 &  & (2.64) & (0.19) & (0.12) & (0.23) & (2.29) & (0.12) & (0.12) & (0.20) & \\
\addlinespace
30 & good & 7.00 & 0.10 & 0.32 & 0.40 & 7.62 & 0.44 & -0.09 & -0.18 & 10.72\\ \rowfont{\footnotesize}
 &  & (5.43) & (0.42) & (0.19) & (0.35) & (2.09) & (0.13) & (0.13) & (0.15) & \\
\cmidrule(lr){3-11}
\end{tabu}
 \caption{Estimated parameters for the threshold ARMA model of Eq.~(\ref{eq:tarma}) fitted to the time series of populations of larvae under four different experimental conditions. The standard errors are reported in parenthesis below the estimates.}\label{tab:Lfit}
\end{table}
\begin{figure}
  \centering
\includegraphics[width=0.45\linewidth,keepaspectratio]{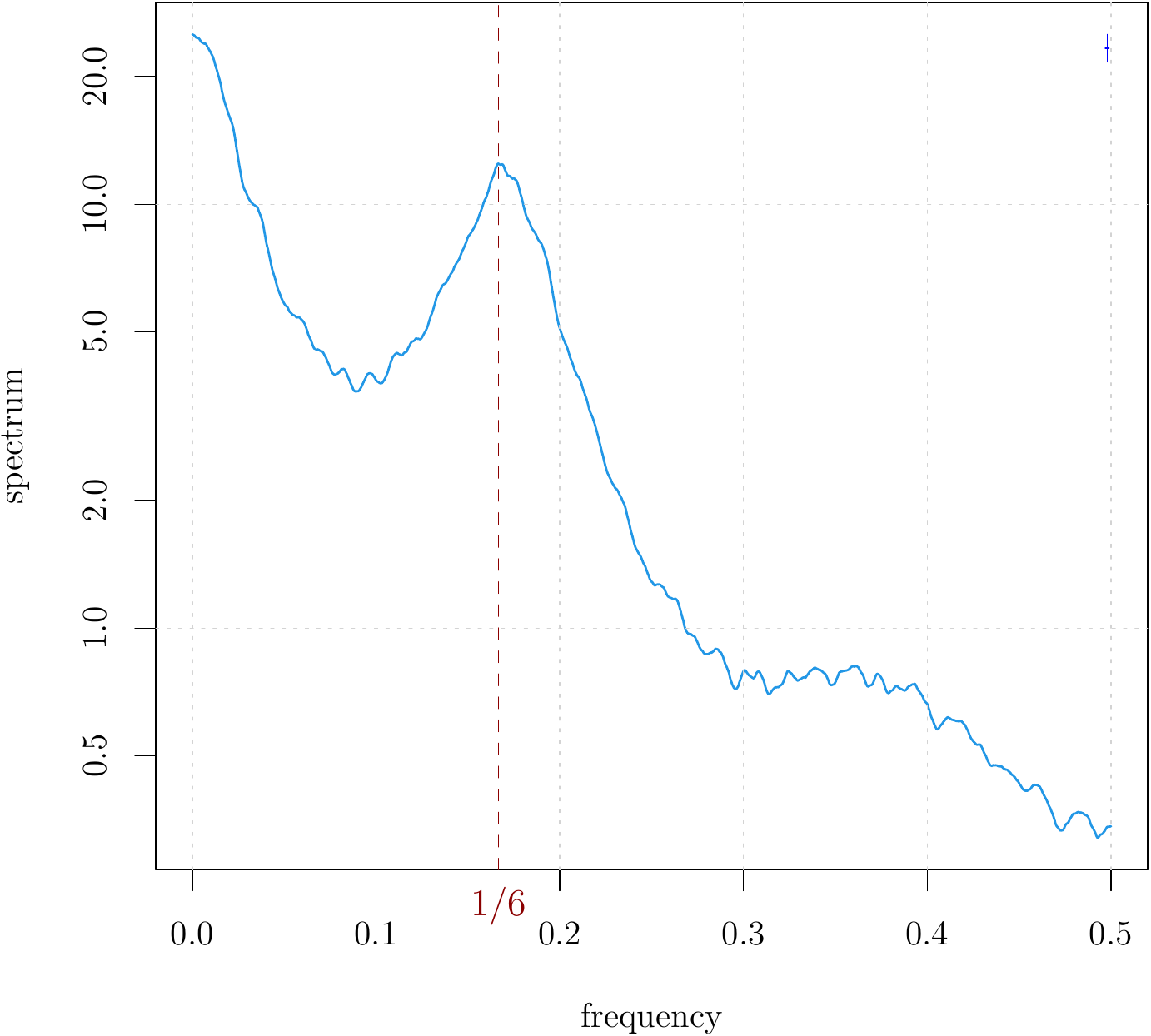}
\includegraphics[width=0.45\linewidth,keepaspectratio]{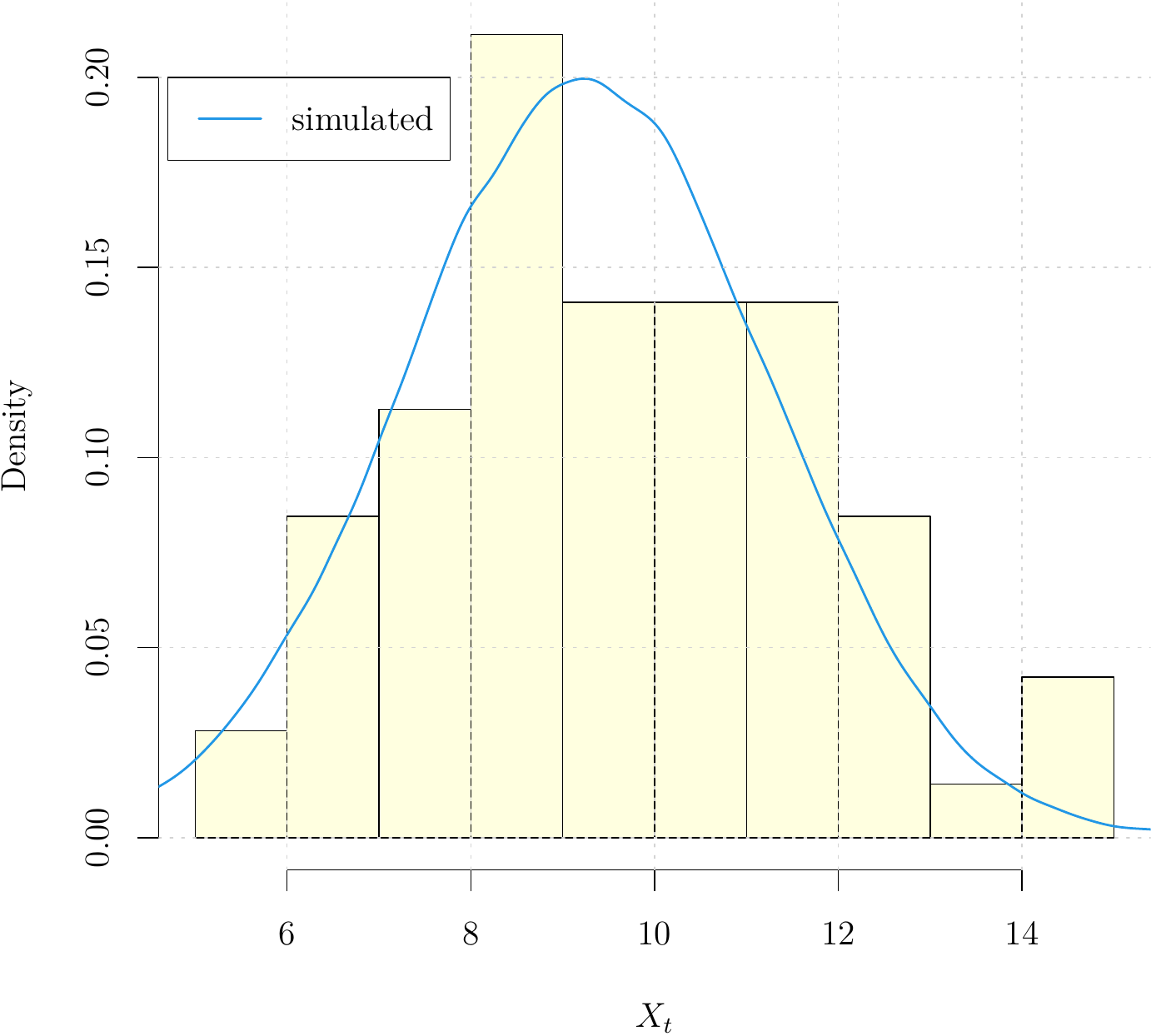}
\caption{(Left) Power spectral density of the simulated time series of 100k observations from the model fit of the second time series (temp: 27°, diet: good). The frequency corresponding to the characteristic 6-week cycle is evidenced with a vertical dashed line.(Right) Histogram of the data (yellow) with the superimposed density of the fitted model, estimated upon the simulated series (blue line).}\label{fig:4}
\end{figure}
%


\section{Proofs}\label{sec:proof}

\subsubsection*{Proof of Proposition~\ref{prop:UBLLN_TARp}}
 \textbf{PART 1.} The proof is divided in two parts: first, we show Eq.~(\ref{eq:UBLLN_TAR1}) for a given $r$ and then we prove that the result holds also uniformly for $r \in [r_{L},r_{U}]$. \par
  \textit{Pointwise convergence.} We assume  $r$ to  be fixed and, for each $\eta>0$, we show that
\begin{equation}\label{eq:UBLLN_point}
  P^*\left(\left|\frac{1}{n}\sum_{t=1}^{n}X^{*u}_{t}I(X^*_{t}\leq r)-E[X^u_{t}I(X_{t}\leq r)]\right|>2\eta\right)\xrightarrow[n\to\infty]{p}0.
\end{equation}
Since the indicator function $I(y\leq r)$ is not differentiable, standard methods based upon Taylor's expansion cannot be applied. We exploit the fact that the function is discontinuous only at $r$. By extending the approach used in \cite{Cha20}, we approximate the step function with a sequence of continuous and differentiable functions $G_{\alpha}(y)$, parameterized by $\alpha \geq 0$:
\begin{equation}\label{eq:Gy}
G_{\alpha}(y) =
\begin{cases}
  \frac{1}{2} + \frac{1}{\pi}\arctan(\frac{r-y}{\alpha}) & \mbox{if } y \neq r \\
  \frac{1}{2}& \mbox{if } y = r
\end{cases}
\end{equation}
In Figure~\ref{fig:G} we show the plot of $G_{\alpha}(y)$ for three values of $\alpha$, together with the limit value $\alpha=0$ for which $G_{\alpha}(y) = I(y \leq r)$ almost surely.
\begin{figure}
  \centering
  \includegraphics[keepaspectratio,width=0.5\linewidth]{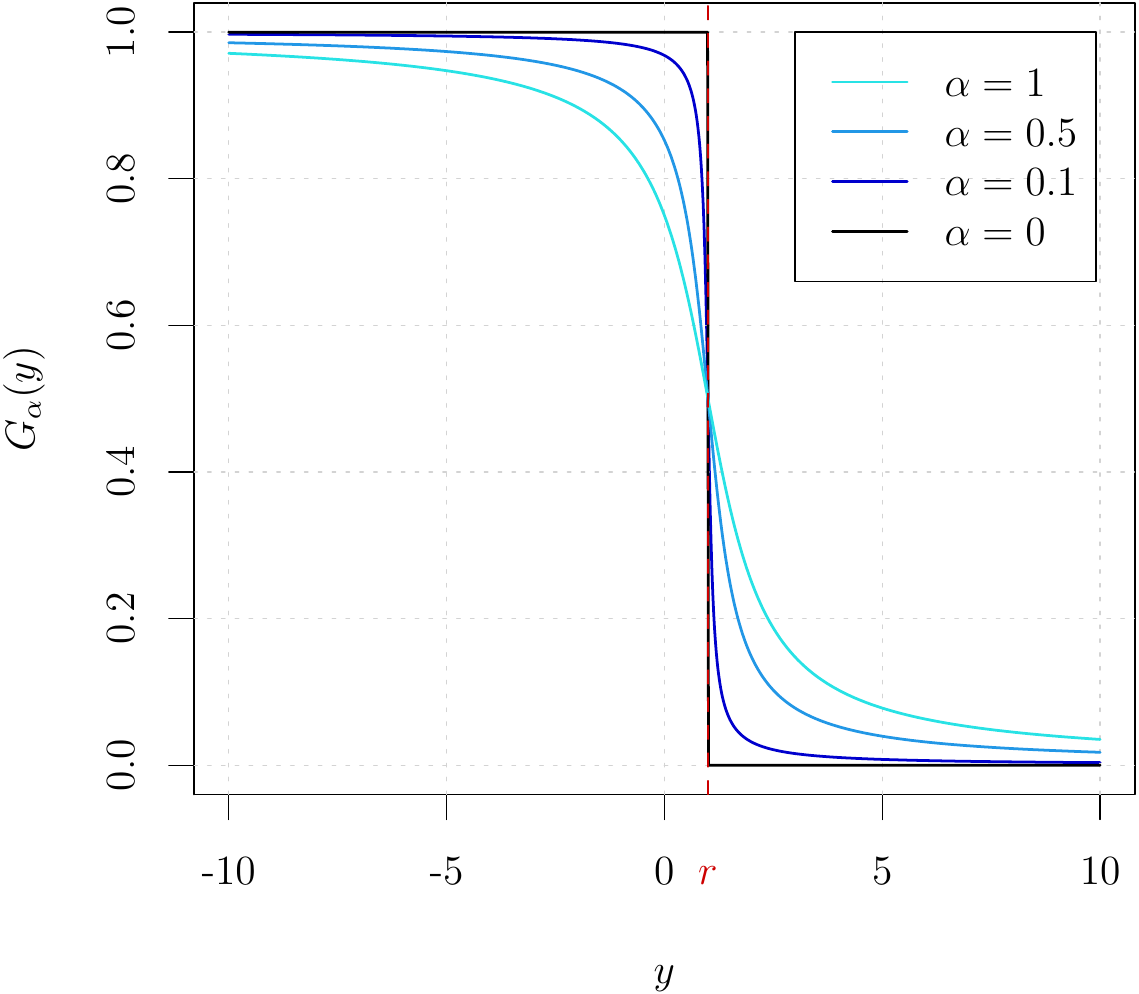}
  \caption{$G_{\alpha}(y)$ for $\alpha=1,0.5,0.1$, together with the limit value $\alpha=0$ for which $G_{\alpha}(y) = I(y \leq r)$ almost surely.}\label{fig:G}
\end{figure}
For each $\delta>0$, define the interval
\begin{equation}\label{eq:L_U_interval}
  [L_{\alpha,\delta},U_{\alpha,\delta}]:= r \pm q_{\alpha,\delta},
\end{equation}
where $q_{\alpha,\delta} = \alpha \tan(\pi(\delta-1/2))$. This implies:
\begin{align}
\left|I(y\leq r)-G_\alpha(y)\right|<\delta&\text{ if } y\notin[L_{\alpha,\delta},U_{\alpha,\delta}];\label{PWC_ind_2}\\
\left|I(y\leq r)-G_\alpha(y)\right|<1&\text{ if } y\in[L_{\alpha,\delta},U_{\alpha,\delta}].\label{PWC_ind_3}
\end{align}
Conditions Eq.~(\ref{eq:L_U_interval})--(\ref{PWC_ind_3}) assure that, when $\alpha$ and $\delta$ approach zero, the interval $[L_{\alpha,\delta},U_{\alpha,\delta}]$ collapses on $r$ and the distance between $G_\alpha(\cdot)$ and $I(\cdot \leq r)$, which is bounded by $\delta$, vanishes.
Now, it holds that:
\begin{align}
&P^*\left(\left|\frac{1}{n}\sum_{t=1}^{n}X^{*u}_{t}I(X^*_{t}\leq r)-E[X^{u}_{t}I(X_{t}\leq r)]\right|>2\eta\right)\nonumber\\
&\leq P^*\left(\left|\frac{1}{n}\sum_{t=1}^{n}X^{*u}_{t}I(X^*_{t}\leq r)-\frac{1}{n}\sum_{t=1}^{n}X^{*u}_{t}G_\alpha(X^*_{t})\right.\right.\nonumber\\
&\phantom{\left(\frac{1}{n}\right)}\left.\left.\phantom{\frac{1}{2}}-E[X^{u}_{t}I(X_{t}\leq r)]+E[X^{u}_{t}G_\alpha(X_{t})]\right|>\eta\right)\label{eq:bound1}\\
&+P^*\left(\left|\frac{1}{n}\sum_{t=1}^{n}X^{*u}_{t}G_\alpha(X^*_{t})-E[X^{u}_{t}G_\alpha(X_{t})]\right|>\eta\right).\label{eq:bound2}
\end{align}
Markov's inequality implies that, in order to prove that Eq.~(\ref{eq:bound1}) is $o_{p^*}(1)$, it suffices to show that the following two expectations vanish in probability:
\begin{align}
&E^*\left[\left|E[X^{u}_{t}I(X_{t}\leq r)]-E[X^{u}_{t}G_\alpha(X_{t})]\right|\right],\label{eqn_main2.1}\\
&E^*\left[\left|\frac{1}{n}\sum_{t=1}^{n}X^{*u}_{t}I(X^*_{t}\leq r)-\frac{1}{n}\sum_{t=1}^{n}X^{*u}_{t}G_\alpha(X^*_{t})\right|\right].\label{eqn_main2.2}
\end{align}
As for  Eq.~(\ref{eqn_main2.1}): let $f_X(\cdot)$ be the stationary probability density function of the AR$(p)$ process $\{X_t\}$. Since it is continuous \citep[see e.g.][theorem 1.3]{And00} and $E[|X_t|^u]<\infty$, by using the same argument developed in \citet{Lin05}, it is possible to show that there exists a positive finite constant, say $M$,  such that $\sup_{x\in\mathds{R}}|x|^uf_X(x)<M$. It holds that:
\begin{align*}
&E^*\left[\left|E[X^{u}_{t}I(X_{t}\leq r)]-E[X^{u}_{t}G_\alpha(X_{t})]\right|\right]=\left|E[X^{u}_{t}I(X_{t}\leq r)]-E[X^{u}_{t}G_\alpha(X_{t})]\right|\\
&\leq E[|X^{u}_{t}|\cdot|I(X_{t}\leq r)-G_\alpha(X_{t})|]\\
&=E[|X^{u}_{t}|\cdot|I(X_{t}\leq r)-G_\alpha(X_{t})|I(X_t\notin[L_{\alpha,\delta},U_{\alpha,\delta}])]\\
&+E[|X^{u}_{t}|\cdot|I(X_{t}\leq r)-G_\alpha(X_{t})|I(X_t\in[L_{\alpha,\delta},U_{\alpha,\delta}])]\\
&\leq \delta E[|X_t|^u]+ M(U_{\alpha,\delta}-L_{\alpha,\delta}),
\end{align*}
where the last inequality follows from Eq.~(\ref{eq:L_U_interval}), Eq.~(\ref{PWC_ind_2}) and Eq.~(\ref{PWC_ind_3}). Hence, Eq.~(\ref{eqn_main2.1}) can be made arbitrarily small in probability by choosing $\alpha$ and $\delta$ sufficiently small. A similar argument handles Eq.~(\ref{eqn_main2.2}):
\begin{align*}
&E^*\left[\left|\frac{1}{n}\sum_{t=1}^{n}X^{*u}_{t}I(X^*_{t}\leq r)-\frac{1}{n}\sum_{t=1}^{n}X^{*u}_{t}G_\alpha(X^*_{t})\right|\right]\\
&\leq E^*\left[\frac{1}{n}\sum_{t=1}^{n}|X^{*u}_{t}|\cdot|I(X^*_{t}\leq r)-G_\alpha(X^*_{t})|\right]\\
&=E^*\left[\frac{1}{n}\sum_{t=1}^{n}|X^{*u}_{t}|\cdot|I(X^*_{t}\leq r)-G_\alpha(X^*_{t})|I(X^*_t\notin[L_{\alpha,\delta},U_{\alpha,\delta}])\right]\\
&+E^*\left[\frac{1}{n}\sum_{t=1}^{n}|X^{*u}_{t}|\cdot|I(X^*_{t}\leq r)-G_\alpha(X^*_{t})|I(X^*_t\in[L_{\alpha,\delta},U_{\alpha,\delta}])\right]\\
&\leq \delta\frac{1}{n}\sum_{t=1}^{n}E^*[|X^{*u}_{t}|]+\mathcal{M}\frac{1}{n}\sum_{t=1}^{n}P^*(X^*_t\in[L_{\alpha,\delta},U_{\alpha,\delta}]),
\end{align*}
where $\mathcal{M}=\max\{|L_{\alpha,\delta}|^u,|U_{\alpha,\delta}|^u,1\}$. Lemma~\ref{lemma:prob} implies that $n^{-1}\sum_{t=1}^{n}P^*(X^*_t\in[L_{\alpha,\delta},U_{\alpha,\delta}])\leq n^{-1}$ with probability 1, hence Eq.~(\ref{eqn_main2.2}) is $o_{p^*}(1)$. Lastly, in order to show that also Eq.~(\ref{eq:bound2}) is $o_{p^*}(1)$, we use the following two expansions:
\begin{align}
  &G_\alpha(X_t^*)=G_\alpha(X_t)+g_\alpha(Y_t^*)(X_t^*-X_t)\label{eq:taylor1};\\
  &G_\alpha(X_t)=G_\alpha(q_{\alpha,\delta}+2r)+g_\alpha(Y_t)(X_t- q_{\alpha,\delta} -2 r)\label{eq:taylor2}
\end{align}
where $q_{\alpha,\delta}$ is defined in Eq.~(\ref{eq:L_U_interval}), $Y_t^*=\lambda_{1,t} X^*_t+(1-\lambda_{1,t})X_t$ and $Y_t=\lambda_{2,t} X_t+(1-\lambda_{2,t})(q_\delta+2r)$ for some $\lambda_{j,t}$ with $0\leq \lambda_{j,t}\leq 1$ and $j=1,2$; moreover,
$$
g_{\alpha}(y) = \frac{\partial G_{\alpha}(y)}{\partial y} =
\begin{cases}
  -\frac{\alpha}{\pi(\alpha^2 +(r-y)^2)} & \mbox{if } y \neq r \\
  0 & \mbox{if } y = r.
\end{cases}
 $$
Note that $g_{\alpha}(y) \xrightarrow[\alpha \rightarrow 0]{} 0\text{ for each } y$. Since the ergodicity of $\{X_t\}$ implies that
 $$\frac{1}{n}\sum_{t=1}^{n}X^u_{t}G_\alpha(X_{t}) \xrightarrow[n\to\infty]{p}E[X^u_{t}G_\alpha(X_{t})],$$
 it suffices to prove that
\begin{align}
 P^*\left(\left|\frac{1}{n}\sum_{t=1}^{n}X^{*u}_{t}G_\alpha(X^*_{t})-\frac{1}{n}\sum_{t=1}^{n} X^u_{t}G_\alpha(X_{t})\right|>\eta/2\right) \xrightarrow[n\to\infty]{p}0, \label{eq:taylor1a}
\end{align}
which can be achieved by using Markov's inequality. Indeed, by using Eq.~(\ref{eq:taylor1}), Eq.~(\ref{eq:taylor2}) and since $G_\alpha(q_\delta+2r)=\delta$  we have
\begin{align*}
&E^*\left[\left|\frac{1}{n}\sum_{t=1}^{n}X^{*u}_{t}G_\alpha(X^*_{t}) -\frac{1}{n}\sum_{t=1}^{n}X^u_{t}G_\alpha(X_{t})\right|\right]\\
&\leq E^*\left[\left|\frac{1}{n}\sum_{t=1}^{n}(X^{*u}_{t}-X^{u}_{t})\left\{\delta+g_\alpha(Y_t)(X_t-q_\delta-2r)\right\}\right|\right]\\
&+ E^*\left[\left|\frac{1}{n}\sum_{t=1}^{n}X^{*u}_{t}g_\alpha(Y_t^*)(X^*_{t}- X_{t})\right|\right]\\
\end{align*}
which can be made arbitrarily small in probability by taking $\alpha$ and $\delta$ sufficiently small and this completes the proof.\par
  \textit{Uniform convergence.} By deploying arguments similar to \citet{Cav17}, we show that for each $\eta>0$
\begin{equation}\label{eq:BLLN_unif}
  P^*\left(\sup_{r\in[r_L,r_U]}\left|\Delta_n^*(r)\right|>2\eta\right)\xrightarrow [n\to\infty]{p}0,
\end{equation}
where $\Delta_n^*(r)=n^{-1}\sum_{t=1}^{n}X^{*u}_{t}I(X^*_{t}\leq r)-E[X^{u}_{t}I(X_{t}\leq r)]$. Since $[r_L,r_U]$ is a compact subset of $\mathds{R}$, for any  $c>0$, there exists a finite coverage $\{[r_{i-1},r_i];i=1,\dots,m\}$, with $m$ being a constant, such that $r_L=r_0<r_1<\ldots < r_{m-1}<r_m=r_U$ ad $r_{i}-r_{i-1}\leq c$, for each $i=1,\ldots,m$. Therefore, it holds that
$$\sup_{r\in[r_L,r_U]}\left|\Delta_n^*(r)\right|\leq \max_{i=0,\ldots,m}\left|\Delta_n^*(r_{i})\right|+\max_{i=1,\ldots,m}\sup_{r\in[r_{i-1},r_i]}\left|\Delta_n^*(r)-\Delta_n^*(r_{i-1})\right|$$
which implies:
\begin{align}
   &P^*\left(\sup_{r\in[r_L,r_U]}\left|\Delta_n^*(r)\right|>2\eta\right)\nonumber\\
   &\leq P^*\left(\max_{i=0,\ldots,m}\left|\Delta_n^*(r_{i})\right|>\eta\right) +P^*\left(\max_{i=1,\ldots,m}\sup_{r\in[r_{i-1},r_i]}\left|\Delta_n^*(r)-\Delta_n^*(r_{i-1})\right|>\eta\right)\label{eq:prob_bound}
\end{align}
By combining Bonferroni's inequality, the pointwise convergence and the finiteness of $m$, we have that
$$P^*\left(\max_{i=1,\ldots,m}\left|\Delta_n^*(r_{i})\right|>\eta\right) \leq \sum_{i=1}^{m}P^*\left(\left|\Delta_n^*(r_{i})\right|>\eta\right)\xrightarrow[n\to\infty]{p}0.$$
It remains to show that the second term of the RHS of Eq.~(\ref{eq:prob_bound}) converges to zero in probability (in probability), which is the case because:
\begin{align*}
&E^*\left[\sup_{r\in[r_{i-1},r_i]}\left|\Delta_n^*(r)-\Delta_n^*(r_{i-1})\right|\right]\\
&\leq E^*\left[\sup_{r\in[r_{i-1},r_i]}\frac{1}{n}\sum_{t=1}^{n}|X_t^{*u}|I(r_{i-1}<X_t^*\leq r)\right.\\
&\phantom{E^*\left[\right.}\phantom{\frac{1}{2}}+\left.\sup_{r\in[r_{i-1},r_i]}E[|X_t|^uI(r_{i-1}<X_t\leq r)]\right]\\
&\leq E^*\left[\frac{1}{n}\sum_{t=1}^{n}|X_t^{*u}|I(r_{i-1}-c<X_t^*\leq r_{i-1}+c)+E[|X_t|^uI(r_{i-1}<X_t\leq r_i)]\right]\\
&\leq \frac{\mathcal{M}_1}{n}\sum_{t=1}^{n}P^*(r_{i-1}-c<X_t^*\leq r_{i-1}+c)+\mathcal{M}_2P(r_{i-1}<X_t\leq r),
\end{align*}
with $\mathcal{M}_1=\max\{ |r_{i-1}-c|^u, |r_{i-1}+c|^u,1 \}$ and $\mathcal{M}_2=\max\{ |r_{i-1}|^u, |r_{i}|^u,1 \}$.
By combining Lemma~\ref{lemma:prob} and Markov's inequality, the proof is completed since $c$ can be chosen arbitrarily small.
\par
\textbf{PART 2.} The proof follows via the same arguments used in 1. and, hence, it is omitted.
\subsubsection*{Proof of Proposition~\ref{prop:matrix_TARp}}
By routine algebra it holds that
$I^*_{n,22}(r)=I^*_{n,12}(r)=I^{*\;\intercal}_{n,21}(r)$ are $(p+1)\times (p+1)$ symmetric matrices whose $(i+1,j+1)$th element is

  \begin{align*}
    &\sum_{t=1}^{n}I(X^*_{t-d}\leq r),  &\mbox{if } i=0,j=0 \\
    &\sum_{t=1}^{n}X^*_{t-j}I(X^*_{t-d}\leq r),  &\mbox{if } i=0,j\neq0 \\
    &\sum_{t=1}^{n}X^*_{t-i}X^*_{t-j}I(X^*_{t-d}\leq r),  &\mbox{if } i\neq0,j\neq0
  \end{align*}
  and $I^*_{n,11}=I^*_{n,22}(\infty)$.
The results readily follows by combining Proposition~\ref{prop:UBLLN_TARp} with $u=0,1,2$ for point 1 and standard results of bootstrap asymptotic analysis.

\subsubsection*{Proof of Proposition~\ref{prop:score_TARp}}

The proof is based upon verifying the following two equalities:
\begin{align}
\sqrt{n}(\tilde{\boldsymbol\phi}-\tilde{\boldsymbol{\phi}}^*)&= -\left(\frac{I^*_{n,11}}{n}\right)^{-1}\frac{1}{\sqrt{n}}\frac{\partial\ell_n^*}{\partial\boldsymbol{\phi}} \label{score_1}\\
\frac{1}{\sqrt{n}}\frac{\partial\tilde\ell_n^*}{\partial\boldsymbol{\Psi}}(r)&=\frac{1}{\sqrt{n}}\frac{\partial\ell_n^*}{\partial\boldsymbol{\Psi}}(r)+ \frac{I^*_{n,21}(r)}{n}\sqrt{n}(\tilde{\boldsymbol{\phi}}-\tilde{\boldsymbol{\phi}}^*),\label{score_2}
\end{align}
where $\frac{\partial\ell_n^*}{\partial\boldsymbol{\phi}}$, $\frac{\partial\ell_n^*}{\partial\boldsymbol{\Psi}}(r)$ and $\frac{\partial\tilde\ell_n^*}{\partial\boldsymbol{\Psi}}(r)$ are defined in Eq.~(\ref{eq:boot_score2}) whereas $I^*_{n,11}$ and $I^*_{n,21}(r)$ in Eq.~(\ref{eq:inf_matrix_boot}).
As previously state we use $\boldsymbol\phi$ as to refer to a generic parameter and let $\frac{\partial\ell_n^*}{\partial\boldsymbol\phi}(\boldsymbol\phi)$ be the partial derivative of the bootstrap log-likelihood  computed under the null hypothesis, i.e.:
$$\left.\frac{\partial\ell_n^*}{\partial\boldsymbol\phi}(\boldsymbol\phi) =\frac{\partial\ell_n^*(\boldsymbol\eta,r)}{\partial\boldsymbol\eta}\right|_ {\boldsymbol\Psi=\boldsymbol 0,\sigma^2=\tilde{\sigma}^2}.$$
Next, we derive two first order Taylor expansions of the function $\frac{\partial\ell_n^*}{\partial\boldsymbol\phi}(\boldsymbol\phi)$: one at the true bootstrap value $\tilde{\boldsymbol\phi}$ ad the other at the bootstrap  MLE $\tilde{\boldsymbol\phi}^*$. Note that, since the $\nu'$ partial derivatives of $\ell_n^*(\boldsymbol\eta,r)$ are zero for $\nu>2$, the Taylor expansion of $\frac{\partial\ell_n^*}{\partial\boldsymbol\phi}(\boldsymbol\phi)$ coincides with its first-order Taylor polynomial; moreover, the Jacobian matrix of $\frac{\partial\ell^*}{\partial\boldsymbol\phi}(\boldsymbol\phi)$ is $-I^*_{n,11}$, defined in Eq.~(\ref{eq:inf_matrix_boot}), which does not depend on $\boldsymbol\phi$. Hence it results that:
\begin{align}
\frac{\partial\ell_n^*}{\partial\boldsymbol\phi}(\boldsymbol\phi)&= \frac{\partial{\ell_n}^*}{\partial\boldsymbol{\phi}}-I_{n,11}(\boldsymbol\phi-\tilde{\boldsymbol\phi}),\label{eq:Taylor1}\\
\frac{\partial\ell_n^*}{\partial\boldsymbol\phi}(\boldsymbol\phi)&= \frac{\partial\tilde{\ell}_n^*}{\partial\boldsymbol{\phi}}-I_{n,11}(\boldsymbol\phi-\tilde{\boldsymbol\phi}^*)\label{eq:Taylor2}
\end{align}
\noindent
with $\frac{\partial\tilde{\ell}_n^*}{\partial\boldsymbol{\phi}}$ and $\frac{\partial{\ell}_n^*}{\partial\boldsymbol{\phi}}$ being defined in Eq.~(\ref{eq:boot_score}). By subtracting Eq.~(\ref{eq:Taylor2}) from Eq.~(\ref{eq:Taylor1}) and dividing by $\sqrt{n}$, we get
\begin{equation}\label{eq:Taylor_diff}
  \frac{1}{\sqrt{n}}\frac{\partial\tilde{\ell}_n^*}{\partial\boldsymbol{\phi}}= \frac{1}{\sqrt{n}}\frac{\partial{\ell}_n^*}{\partial\boldsymbol{\phi}} +\frac{I_{n,11}^*}{n}\sqrt{n}(\tilde{\boldsymbol{\phi}}-\tilde{\boldsymbol{\phi}}^*).
\end{equation}
Since $\tilde{\boldsymbol{\phi}}^*$ is the bootstrap MLE obtained under the null hypothesis, $\frac{\partial\tilde{\ell}_n^*}{\partial\boldsymbol{\phi}}=0$ thence Eq.~(\ref{eq:Taylor_diff}) implies
\begin{align*}
&\frac{I_{n,11}^*}{n}\sqrt{n}(\tilde{\boldsymbol{\phi}}^*-\tilde{\boldsymbol{\phi}}) =-\frac{1}{\sqrt{n}}\frac{\partial\ell_n^*}{\partial\boldsymbol{\phi}}.
\end{align*}
and hence Eq.~(\ref{score_1}) follows. We prove Eq.~(\ref{score_2}) componentwise. We detail below the argument only for the  first component since it can be easily adapted to the other ones. Therefore, we show that:
\begin{align}
 \frac{1}{\sqrt{n}}\frac{\partial\tilde\ell_n^*}{\partial{\Psi_0}}(r)&= \frac{1}{\sqrt{n}}\frac{\partial\ell_n^*}{\partial{\Psi_0}}(r) +\sqrt{n}({\tilde\phi_0}-{\tilde\phi_0^*})\frac{1}{n}\sum_{t=1}^{n}I(X^*_{t-d}\leq r)\nonumber\\
 &+\sum_{i=1}^{p}\sqrt{n}({\tilde\phi_i}-{\tilde\phi_i^*})\frac{1}{n}\sum_{t=1}^{n}X^*_{t-i}I(X^*_{t-d}\leq r).\label{score_2_2}
\end{align}
Let $\tilde\eps_t^*$ be the residuals obtained from the ML fit upon the bootstrap sample $\{X^*_t,t=1\dots,n\}$, i.e.:
\begin{align*}
\tilde{\eps}^*_t&=X^*_t-\tilde{\phi}^*_0-\sum_{i=1}^{p}\tilde{\phi}^*_iX^*_{t-i} =(\tilde{\phi}_{0}-\tilde{\phi}^*_0)+\sum_{i=1}^{p}(\tilde{\phi}_{i}-\tilde{\phi}^*_i)X_{t-i}+\eps^*_t.
\end{align*}
Clearly:
\begin{equation}
\eps_t^*-\tilde\eps_t^*=({\tilde\phi_0^*}-{\tilde\phi_0})+\sum_{i=1}^{p}({\tilde\phi_i^*}-{\tilde\phi_i})X_{t-i}^*.\label{eq:eps_dif}
\end{equation}
Note that $\eps^*_t(\boldsymbol\eta,r)$, defined in Eq.~(\ref{eq:eps_boot}), does not depend on $\sigma^2$ and  $\tilde{\eps}^*_t$ and ${\eps}^*_t$ correspond to the function $\eps^*_t(\boldsymbol\eta,r)$ evaluated at $\boldsymbol\phi=\tilde{\boldsymbol\phi}^*$, $\bPsi=\boldsymbol 0$ and $\boldsymbol\phi=\tilde{\boldsymbol\phi}$, $\bPsi=\boldsymbol 0$, respectively. Consider the partial derivatives of the function $\eps^*_t(\boldsymbol\eta,r)$  and denote:
\begin{equation*}
\left.D_{\Psi_0t}^*(r) =\frac{\partial\eps^*_t(\boldsymbol\eta,r)}{\partial\Psi_0}\right|_ {\boldsymbol\phi=\tilde{\boldsymbol\phi},\bPsi=\boldsymbol 0}
,\quad
\left.\tilde{D}_{\Psi_0t}^*(r)=\frac{\partial\eps^*_t(\boldsymbol\eta,r)}{\partial\Psi_0}\right|_ {\boldsymbol\phi=\tilde{\boldsymbol\phi}^*,\bPsi=\boldsymbol 0}.
\end{equation*}
Note that:
$$D_{\Psi_0t}^*(r)=\tilde{D}_{\Psi_0t}^*(r)=-I(X^*_{t-d}\leq r)$$
therefore, we get:
\begin{align*}
\frac{1}{\sqrt{n}}\frac{\partial\tilde\ell_n^*}{\partial{\Psi_0}}(r)&= -\frac{1}{\sqrt{n}}\sum_{t=1}^{n}\tilde{\eps}_t^*\tilde{D}_{\Psi_0t}^*(r) =-\frac{1}{\sqrt{n}}\sum_{t=1}^{n}\tilde{\eps}_t^*{D}_{\Psi_0t}^*(r)\\
&=-\frac{1}{\sqrt{n}}\sum_{t=1}^{n}\tilde{\eps}_t^*{D}_{\Psi_0t}^*(r)
 -\frac{1}{\sqrt{n}}\sum_{t=1}^{n}{\eps}_t^*{D}_{\Psi_0t}^*(r) +\frac{1}{\sqrt{n}}\sum_{t=1}^{n}{\eps}_t^*{D}_{\Psi_0t}^*(r)\\
&=\frac{1}{\sqrt{n}}\frac{\partial\ell_n^*}{\partial{\Psi_0}}(r) +\frac{1}{\sqrt{n}}\sum_{t=1}^{n}({\eps}_t^*-\tilde{\eps}_t^*){D}_{\Psi_0t}^*(r).
\end{align*}
The expression of $({\eps}_t^*-\tilde{\eps}_t^*)$ in Eq.~(\ref{eq:eps_dif}) implies that
\begin{align*}
\frac{1}{\sqrt{n}}\frac{\partial\tilde\ell_n^*}{\partial{\Psi_0}}(r)&= \frac{1}{\sqrt{n}}\frac{\partial\ell_n^*}{\partial{\Psi_0}}(r) +\frac{1}{\sqrt{n}}\sum_{t=1}^{n}\left\{({\tilde\phi_0^*}-{\tilde\phi_0}) +\sum_{i=1}^{p}({\tilde\phi_i^*}-{\tilde\phi_i})X_{t-i}^*\right\}{D}_{\Psi_0t}^*(r)\\
&=\frac{1}{\sqrt{n}}\frac{\partial\ell_n^*}{\partial{\Psi_0}}(r) \\ &+\frac{1}{\sqrt{n}}\sum_{t=1}^{n}\left\{({\tilde\phi_0^*}-{\tilde\phi_0})+\sum_{i=1}^{p}({\tilde\phi_i^*}-{\tilde\phi_i})X_{t-i}^*\right\} \left\{-I(X^*_{t-d}\leq r)\right\}\\
&=\frac{1}{\sqrt{n}}\frac{\partial\ell_n^*}{\partial{\Psi_0}}(r)+ \sqrt{n}({\tilde\phi_0}-{\tilde\phi_0^*})\frac{1}{n}\sum_{t=1}^{n}I(X^*_{t-d}\leq r)\\
& +\sum_{i=1}^{p}\sqrt{n}({\tilde\phi_i}-{\tilde\phi_i^*})\frac{1}{n}\sum_{t=1}^{n}X^*_{t-i}I(X^*_{t-d}\leq r)
\end{align*}
and this completes the proof.

\subsubsection*{Proof of Proposition~\ref{prop:BCLT_TARp}}
Since $\frac{\partial\ell_n^*}{\partial\boldsymbol\eta}(r)=-\sum_{t=1}^{n}\eps^*_t D^*_{t-1}(r)$, with $D^*_{t}(r)$ being defined in Eq.~(\ref{eq:epsder}), forms a sequence of martingale difference arrays with respect to the filtration $\mathcal{F}_{t-1}^*:=\sigma\{X^*_{t-1}, X^*_{t-2},\dots\}$, the result holds upon proving, uniformly on $r$, the following two conditions:
  \begin{align}
&\frac{1}{n}\sum_{t=1}^{n}E^{*}\left[\eps^{*2}_t(D^*_{t-1}(r))(D^*_{t-1}(r))^\intercal|\mathcal{F}^{*}_{t-1}\right] \xrightarrow[n\to\infty]{p^*}_p\sigma^2 I_\infty(r);\label{BCLT_1}\\
&\frac{1}{n}\sum_{t=1}^{n}E^{*}\left[\eps^{*2}_t\Lambda^{*2}_{t-1}(r) I\left(\left|\eps^{*}_t\Lambda^{*}_{t-1}(r)\right|>\eta\sqrt{n}\right)|\mathcal{F}^{*}_{t-1}\right]\xrightarrow[n\to\infty]{p^*}_p0,\label{BCLT_2}
  \end{align}
  with $\Lambda^{*}_{t}(r):=(\lambda_1,\dots,\lambda_{2(p+1)})\cdot D_t^*(r)$, with $\lambda_i$, $i=1,\dots,2(p+1)$, being real numbers. In order to prove Eq.~(\ref{BCLT_1}) note that the independence between $\eps_t^*$ and $X_{t-j}^*$, $j\geq1$ implies that
  \begin{align*}
  &\frac{1}{n}\sum_{t=1}^{n}E^{*}\left[\eps^{*2}_t(D^*_{t-1}(r))(D^*_{t-1}(r))^\intercal|\mathcal{F}^{*}_{t-1}\right] 
  =E^{*}\left[\eps_t^{*2}\right]\frac{1}{n}\sum_{t=1}^{n}D_{t-1}^*(r)(D_{t-1}^*(r))^\intercal,
  \end{align*}
  which converges in probability (in probability) to $\sigma^2 I_\infty(r)$ uniformly on $r$ by Lemma~\ref{lemma:LLN_eps} and Proposition~\ref{prop:matrix_TARp}. As for Eq.~(\ref{BCLT_2}) first observe that, by using Jensen's inequality and Proposition~\ref{prop:UBLLN_TARp}, $n^{-1}\sum_{t=1}^{n}\Lambda^{*2}_{t-1}(r)$ is bounded by
  \begin{equation}\label{eq:lambda2}
   \frac{2(p+1)}{n}\sum_{t=1}^{n}\left[\left(\lambda_1^2+\lambda_{p+2}^2\right) +\sum_{i=2}^{p+1}\left(\lambda_{i}^2+\lambda_{i+p+1}^2\right)X^{*2}_{t-i+1}\right]=O_{p^*}(1),
  \end{equation}
  whereas $n^{-1}\sum_{t=1}^{n}\Lambda^{*4}_{t-1}(r)$ is bounded by
  \begin{equation}\label{eq:lambda4}
   \frac{8(p+1)^3}{n}\sum_{t=1}^{n}\left[\left(\lambda_1^4+\lambda_{p+2}^4\right) +\sum_{i=2}^{p+1}\left(\lambda_{i}^4+\lambda_{i+p+1}^4\right)X^{*4}_{t-i+1}\right]=O_{p^*}(1).
  \end{equation}
  Now, since $|xy|\leq x^2+y^2$, it follows that
  \begin{align*}
  &\frac{1}{n}\sum_{t=1}^{n}E^{*}\left[\eps^{*2}_t\Lambda^{*2}_{t-1}(r) I\left(\left|\eps^{*}_t\Lambda^{*}_{t-1}(r)\right|>\eta\sqrt{n}\right)|\mathcal{F}^{*}_{t-1}\right]\\
  &\leq \frac{1}{n}\sum_{t=1}^{n}E^{*}\left[\eps^{*2}_t\Lambda^{*2}_{t-1}(r) I\left(\Lambda^{2*}_{t-1}(r)>2^{-1}\eta\sqrt{n}\right)|\mathcal{F}^{*}_{t-1}\right]\\
  &+\frac{1}{n}\sum_{t=1}^{n}E^{*}\left[\eps^{*2}_t\Lambda^{*2}_{t-1}(r) I\left(\eps^{*2}_t>2^{-1}\eta\sqrt{n}\right)|\mathcal{F}^{*}_{t-1}\right]\\
  &\leq \frac{2}{\eta\sqrt{n}}\left\{\frac{1}{n}\sum_{t=1}^{n}\Lambda^{*4}_{t-1}(r)E^*[\eps^{*2}_t] +\frac{1}{n}\sum_{t=1}^{n}\Lambda^{*2}_{t-1}(r)E^*[\eps^{*4}_t]\right\}
  \end{align*}
  which is $o_{p^*}(1)$ by combing Eq.~(\ref{eq:lambda2}), Eq.~(\ref{eq:lambda4}) and Lemma~\ref{lemma:LLN_eps}.
\subsubsection*{Proof of Theorem~\ref{thm:boot}}
In view of Proposition~\ref{prop:BCLT_TARp} and Theorem 18.14, p. 261 of \citet{Vaa98}, it suffices to prove the stochastic equicontinuity of $\frac{\partial\ell_n^*}{\partial\boldsymbol\eta}(r)=-\sum_{t=1}^{n}\eps^*_t D^*_{t-1}(r)$, where $D^*_{t-1}(r)$ is defined in Eq.~(\ref{eq:epsder}). The envelope of $\eps^*_t D^*_{t-1}(r)$ is $\mathcal{L}^2$ integrable in probability:
  \begin{align*}
  &\frac{1}{n}\sum_{t=1}^{n}E^*\left[\sup_{r\in[r_L,r_U]}\|\eps^*_t x^*_{t-1}(r)\|^2\right]= \frac{1}{n}\sum_{t=1}^{n}E^*\left[\sup_{r\in[r_L,r_U]}(\eps^*_t x^*_{t-1}(r))^\intercal (\eps^*_t x^*_{t-1}(r))\right]\\
  &=\frac{1}{n}\sum_{t=1}^{n}E^*\left[\sup_{r\in[r_L,r_U]}\eps_t^{*2}\left\{1+\sum_{i=1}^{p}X^{*2}_{t-i}+I(X^{*}_{t-d}\leq r)+\sum_{i=1}^{p}X^{*2}_{t-i}I(X^{*}_{t-d}\leq r)\right\}\right]\\
  &\leq\frac{2}{n}\sum_{t=1}^{n}E^*\left[\eps_t^{*2}\left\{1+\sum_{i=1}^{p}X^{*2}_{t-i}\right\}\right] =\frac{2}{n}\sum_{t=1}^{n}E^*\left[E^*\left[\eps_t^{*2}\left\{1+\sum_{i=1}^{p}X^{*2}_{t-i}\right\}|\mathcal{F}^{*}_{t-1}\right]\right]\\
  &=\frac{2}{n}\sum_{t=1}^{n}E^*\left[1+\sum_{i=1}^{p}X^{*2}_{t-i}\right]E^*\left[\eps_t^{*2}|\mathcal{F}^{*}_{t-1}\right] =E^*\left[\eps_t^{*2}\right]E^*\left[\frac{2}{n}\sum_{t=1}^{n}\left\{1+\sum_{i=1}^{p}X^{*2}_{t-i}\right\}\right]\\
  &=O_{p^*}(1).
  \end{align*}
Define the norms:
\begin{align*}
\rho^*_n(r_1,r_2)&=\left\|\frac{1}{\sqrt{n}}\left(\frac{\partial\ell^*}{\partial\boldsymbol\eta}(r_2) -\frac{\partial\ell^*}{\partial\boldsymbol\eta}(r_1)\right)\right\|_2\\
\text{and } \rho(r_1,r_2)&=\left\|\eps_tD_{t-1}(r_2)-\eps_tD_{t-1}(r_1)\right\|_2,
\end{align*}
where, in analogy with Eq.~(\ref{eq:epsder}), $D_t(r)$ is the first-order derivative of the function $\eps_t(\boldsymbol\eta,r)$ defined in Eq.~(\ref{eq:eps}), i.e.:
\begin{align*}
D_t(r)&=\left(-1,-X_t,\dots,-X_{t-p+1},\right.\\
&\left.-I(X_{t-d+1}\leq r),-X_tI(X_{t-d+1}\leq r),\dots,-X_{t-p+1}I(X_{t-d+1}\leq r)\right)^\intercal.
\end{align*}
  It holds that
  \begin{align*}
  &\rho^{*2}_n(r_1,r_2)=E^*\left\|\frac{1}{\sqrt{n}}\left(\frac{\partial\ell^*}{\partial\boldsymbol\eta}(r_2) -\frac{\partial\ell^*}{\partial\boldsymbol\eta}(r_1)\right)\right\|\\
  &= E^*\left[\frac{1}{n}\sum_{t=1}^{n}\eps_t^{*2}\left(I(r_1<X^{*}_{t-d}\leq r_2)+\sum_{i=1}^{p}X^{*2}_{t-i}I(r_1<X^{*}_{t-d}\leq r_2)\right)\right].
  \end{align*}
  By using the law of iterated expectations and Proposition~\ref{prop:UBLLN_TARp}, $\rho^*_n(r_1,r_2)$ converges uniformly to
  \begin{align*}
  \left\{\sigma^2P(r_1<X_t\leq r_2)+\sigma^2\sum_{i=1}^{p}E[X_{t-i}^2I(r_1<X_{t-d}\leq r_2)]\right\}=\rho^2(r_1,r_2).
  \end{align*}
  Thence the same argument of Theorem 2 of \citet{Han96} holds, and this completes the proof.

\bibliographystyle{plainnat}
\bibliography{bootstrapTARMA_V3}


\newpage
\setcounter{table}{0}
\setcounter{section}{0}
\setcounter{page}{1}
\renewcommand{\thesection}{\Alph{section}}

\begin{center}
\section*{\Large Supplement for:\\ The validity of bootstrap testing in the threshold framework}
\end{center}

\bigskip

\begin{center}
\Large
{Simone Giannerini, Greta Goracci, Anders Rahbek}
\end{center}

\bigskip\bigskip

\noindent
\textbf{Abstract}\\
\smallskip
This Supplement has 3 sections. In Section~\ref{SMsec:lemmas} we present auxiliary technical lemmas used in the proofs. Section~\ref{SMsec:MC} contains supplementary results from the simulation study. Section~\ref{SMsec:applic} presents supplementary results on the analysis of larvae population dynamics under the effect of warming.

\bigskip\bigskip

\section{Auxiliary Lemmas}\label{SMsec:lemmas}

\begin{lemma}\label{lemma:prob}
  Let $\{X_t^*,t=1,\dots,n\}$ be defined in Eq.~(\ref{eq:boot_proc}) and assume $b_1(c)$ and $b_2(c)$ to be two continuous functions in $c$ such that
  $$\lim_{c\to 0}b_1(c)=\lim_{c\to 0}b_2(c)=\mathrm{C},$$
  with $\mathrm{C}$ being a real number. Then, for each $\gamma>0$ we can choose $c$ sufficiently small such that
  \begin{equation}\label{eq:lemma}
    P^*(b_1(c)\leq X^*_{t}\leq b_2(c))\leq \frac{1}{n}+\gamma
  \end{equation}
  with probability one.
\end{lemma}
\begin{proof}
  From the definition of limit, for each $\gamma>0$ we can choose $c$ sufficiently small such that
  $$P^*(b_1(c)\leq X^*_{t}\leq b_2(c))\leq P^*(X_t^*=\mathrm{C})+\gamma;$$
  hence it remains to show that  $P^*(X_t^*=\mathrm{C})\leq 1/n$ in probability. Define $\mathcal{A}^*_t$ to be the set of values that $X_t^*$ can assume conditionally to the data. By using the fact that: $(i)$ for any real number $\kappa\in \mathds{R}$,  $P^*(\eps_t^*=\kappa)\leq 1/n$ with probability one  and  $(ii)$ $\sum_{a\in\mathcal{A}^*_{s}}P^*(X^*_{s} =a)=1$, for any integer $s$, it holds that
  \begin{align*}
  &P^*(X^*_t =\mathrm{C})\\
  &=\sum_{a\in\mathcal{A}^*_{t-1}}P^*(X^*_t =\mathrm{C}|X^*_{t-1}=a)P^*(X^*_{t-1} =a)\\
  &=\sum_{a\in\mathcal{A}^*_{t-1}}P^*(\eps^*_t =\mathrm{C}-\tilde{\phi}_0-\tilde{\phi}_1 a)P^*(X^*_{t-1} =a)\\
  &\leq \frac{1}{n}\sum_{a\in\mathcal{A}^*_{t-1}}P^*(X^*_{t-1} =a)=\frac{1}{n}
  \end{align*}
  and the proof is completed.
\end{proof}
\begin{lemma}\label{lemma:LLN_eps}
  \textbf{(LLN)} Let $\{\eps_t^{*}\}$ be defined in Section~\ref{sec:boot}. Under Assumption~\ref{ass} , it holds that:
  $$E^*[\eps_t^{*2}]\xrightarrow[n\to\infty]{p}\sigma^2 \quad \text{and} \quad E^*[\eps_t^{*4}]\xrightarrow[n\to\infty]{p}\kappa.$$
\end{lemma}
\begin{proof}
Since
$$\tilde{\eps}_t  =(\phi_{0,0}-\tilde{\phi}_0)+\sum_{i=1}^{p}(\phi_{i,0}-\tilde{\phi}_i)X_{t-i}+\eps_t$$
and $(\tilde{\boldsymbol\phi}-\boldsymbol\phi_0)=O_p(n^{-1/2})$, it follows that
\begin{align*}
\bar{\tilde{\eps}}:=\frac{1}{n}\sum_{t=1}^{n}\tilde{\eps}_t =\frac{1}{n}\sum_{t=1}^{n}\left[(\phi_{0,0}-\tilde{\phi}_0) +\sum_{i=1}^{p}(\phi_{i,0}-\tilde{\phi}_i)X_{t-i}+\eps_t\right]
\end{align*}
converges in probability to zero. Similarly, by routine algebra, it is possible to show that $E^*[\eps_t^{*2}]=\frac{1}{n}\sum_{t=1}^{n}\left(\tilde{\eps}_t-\bar{\tilde{\eps}}\right)^2$ and $E^*[\eps_t^{*4}]=\frac{1}{n}\sum_{t=1}^{n}\left(\tilde{\eps}_t-\bar{\tilde{\eps}}\right)^4$ converge in probability to $\sigma^2$ and $\kappa$ respectively.
\end{proof}

\section{Supplementary Monte Carlo results}\label{SMsec:MC}
\subsection{Empirical (uncorrected) power of the tests}\label{SMsec:power}
Table~\ref{tab:p1} reports the (uncorrected) power of the tests at nominal level $\alpha=5\%$ for the TAR(1) process of Eq.~(\ref{eq:TAR1}).
\begin{table}[H]
  \centering
\begin{tabular}{rrrrrrr}
\multicolumn{1}{c}{ } & \multicolumn{2}{c}{$n=50$} & \multicolumn{2}{c}{$n=100$} & \multicolumn{2}{c}{$n=200$} \\
\cmidrule(l{3pt}r{3pt}){2-3} \cmidrule(l{3pt}r{3pt}){4-5} \cmidrule(l{3pt}r{3pt}){6-7}
$\Psi$ & ARa & ARi & ARa & ARi & ARa & ARi\\
\cmidrule(l{3pt}r{3pt}){2-3} \cmidrule(l{3pt}r{3pt}){4-5} \cmidrule(l{3pt}r{3pt}){6-7}
0.0 & 9.1 & 4.9 & 6.8 & 6.2 & 5.2 & 5.2\\
0.3 & 12.0 & 9.0 & 17.4 & 18.3 & 42.1 & 44.3\\
0.6 & 26.3 & 27.3 & 59.8 & 62.4 & 94.1 & 94.4\\
0.9 & 52.5 & 54.8 & 91.3 & 92.2 & 100.0 & 100.0\\
\addlinespace
0.0 & 3.5 & 4.4 & 3.8 & 5.4 & 3.9 & 5.2\\
0.3 & 4.3 & 5.0 & 5.9 & 8.2 & 7.3 & 8.6\\
0.6 & 6.5 & 8.0 & 12.4 & 14.9 & 33.8 & 36.8\\
0.9 & 10.1 & 12.1 & 32.1 & 34.1 & 65.8 & 68.1\\
\cmidrule(l{3pt}r{3pt}){2-3} \cmidrule(l{3pt}r{3pt}){4-5} \cmidrule(l{3pt}r{3pt}){6-7}
\end{tabular}
\caption{Empirical power at nominal level $\alpha=5\%$ for the TAR(1) process of Eq.~(\ref{eq:TAR1}).}\label{tab:p1}
\end{table}

\section{Supplementary results from the real application: the effect of warming on populations of larvae}\label{SMsec:applic}
The data come from the study published in \cite{Lau19} and are publicly available at \cite{Lau19b}.

\begin{figure}[H]
  \centering
\includegraphics[width=0.9\linewidth,keepaspectratio]{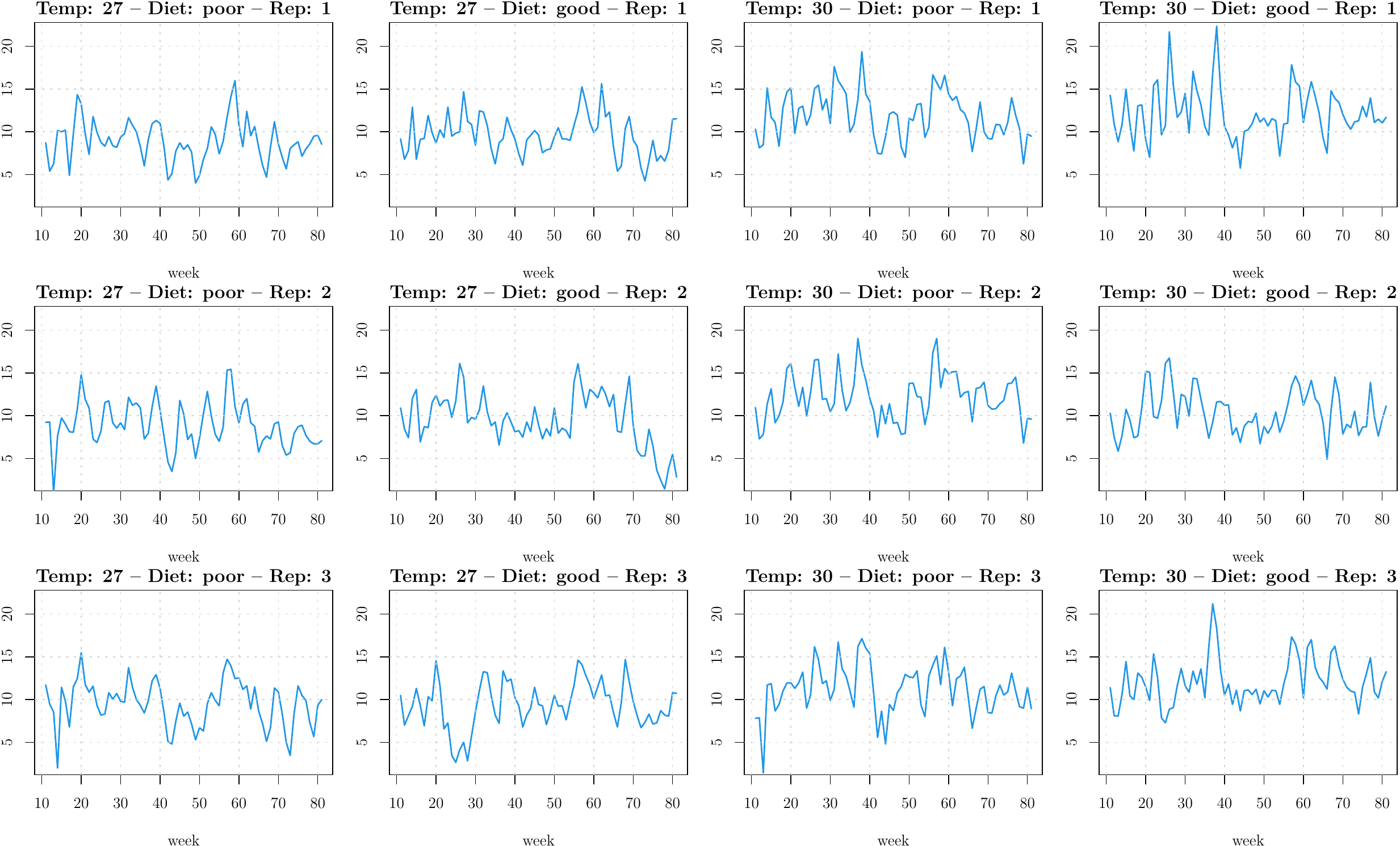}
\caption{Time series of 12 populations of \emph{P. interpunctella} from week 11 to 82 for different experimental conditions. The series have been square-root transformed.}\label{SMfig:1}
\end{figure}

\begin{figure}[H]
  \centering
\includegraphics[width=0.45\linewidth,keepaspectratio]{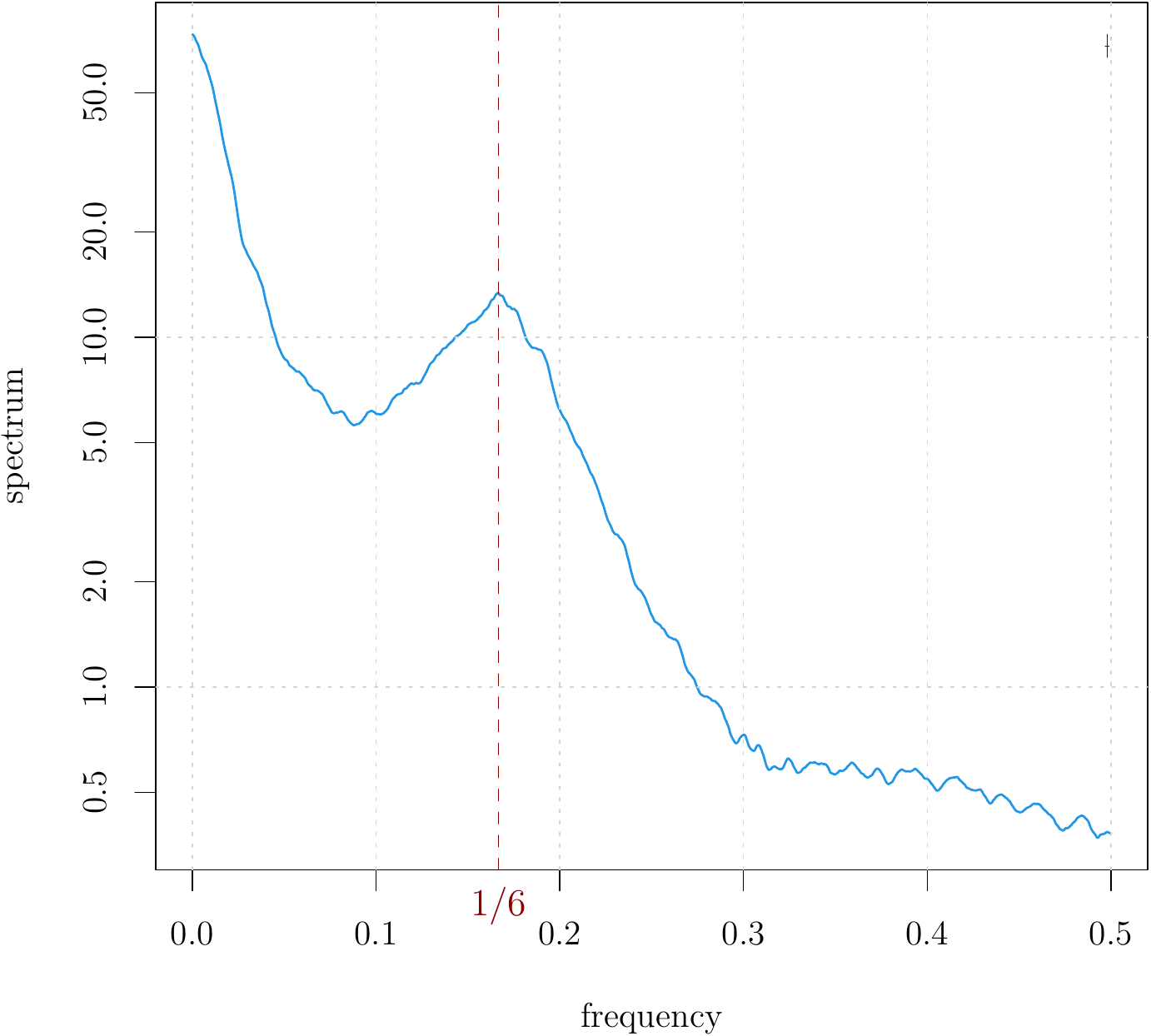}
\includegraphics[width=0.45\linewidth,keepaspectratio]{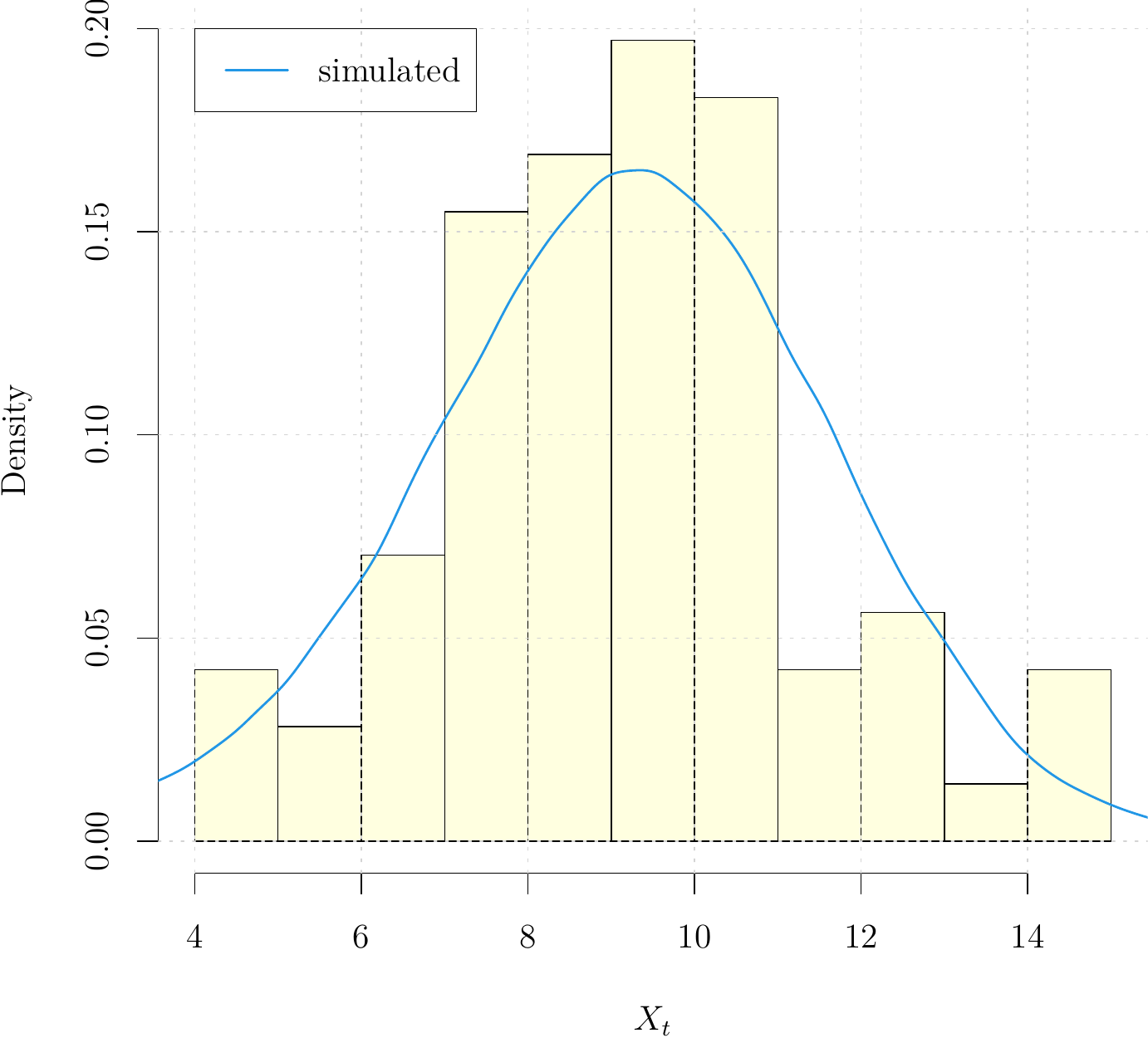}

\caption{(Left) Power spectral density of the simulated time series of 100k observations from the model fit of the first time series (temp: 27°, diet: poor). The frequency corresponding to the characteristic 6-week cycle is evidenced with a vertical dashed line.(Right) Histogram of the data (yellow) with the superimposed density of the fitted model, estimated upon the simulated series (blue line).}\label{SMfig:5}
\end{figure}

\begin{figure}[H]
  \centering
\includegraphics[width=0.45\linewidth,keepaspectratio]{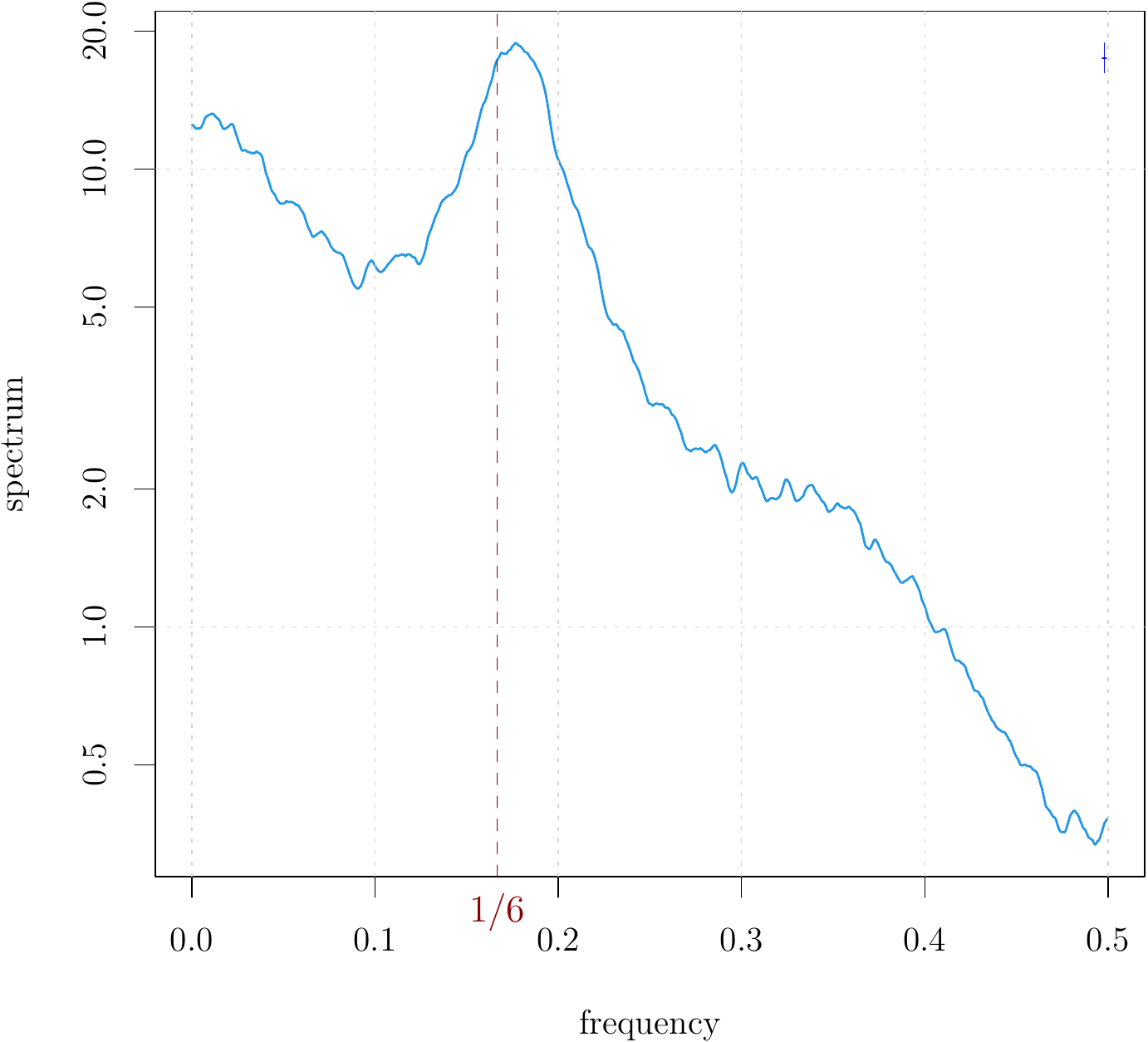}
\includegraphics[width=0.45\linewidth,keepaspectratio]{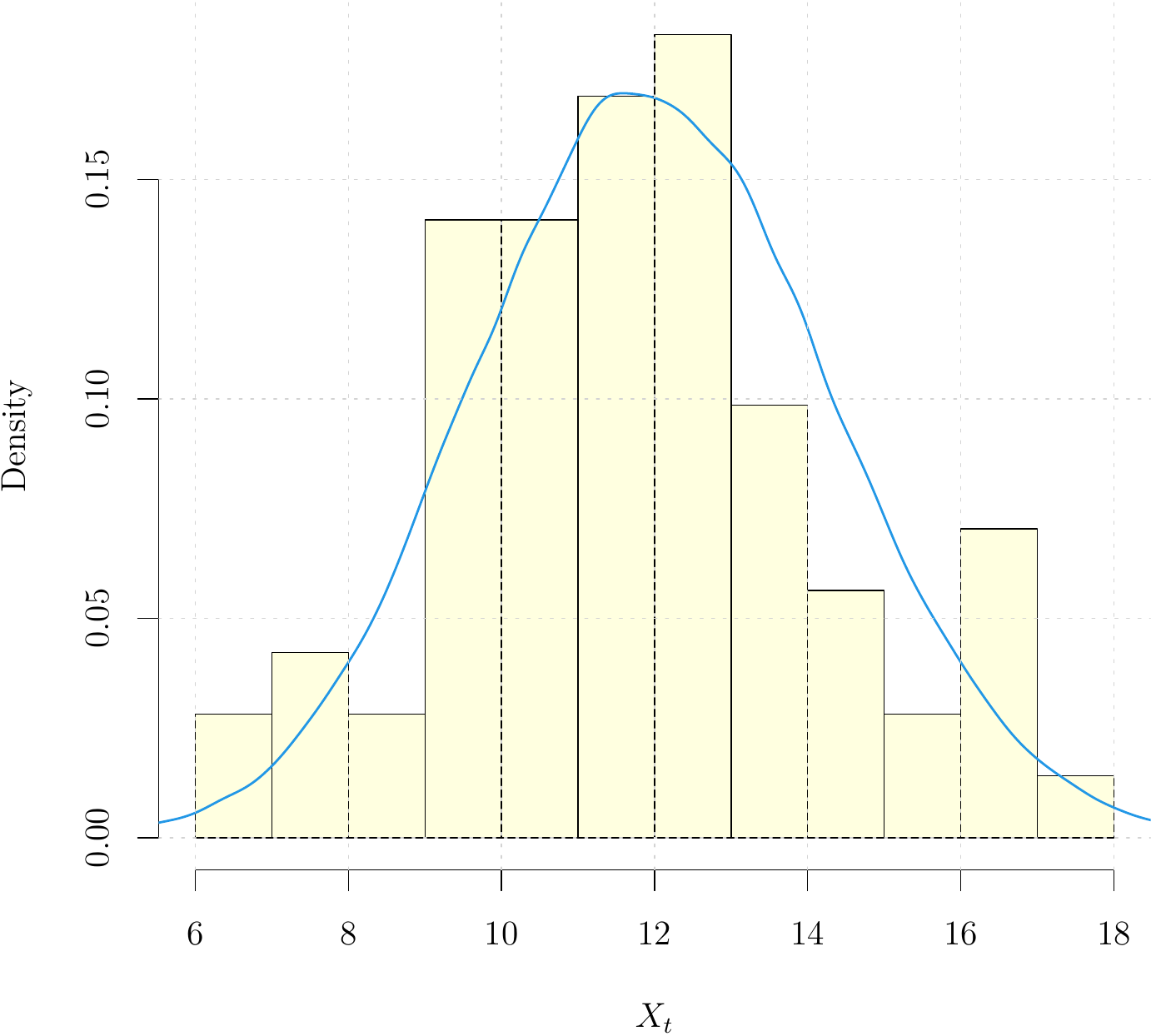}
\caption{(Left) Power spectral density of the simulated time series of 100k observations from the model fit of the third time series (temp: 30°, diet: poor). The frequency corresponding to the characteristic 6-week cycle is evidenced with a vertical dashed line.(Right) Histogram of the data (yellow) with the superimposed density of the fitted model, estimated upon the simulated series (blue line).}\label{SMfig:6}
\end{figure}

\begin{figure}[H]
  \centering
\includegraphics[width=0.45\linewidth,keepaspectratio]{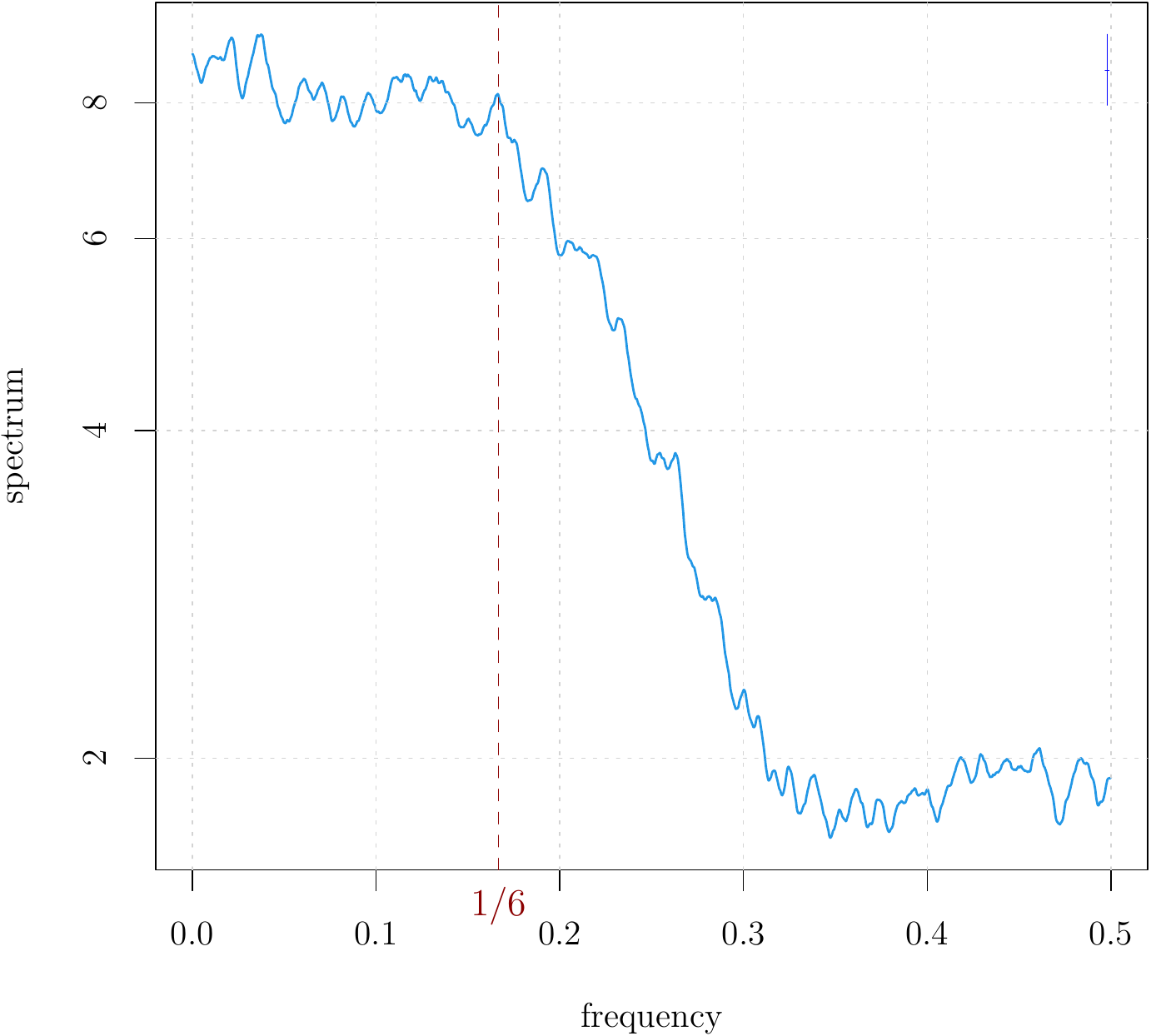}
\includegraphics[width=0.45\linewidth,keepaspectratio]{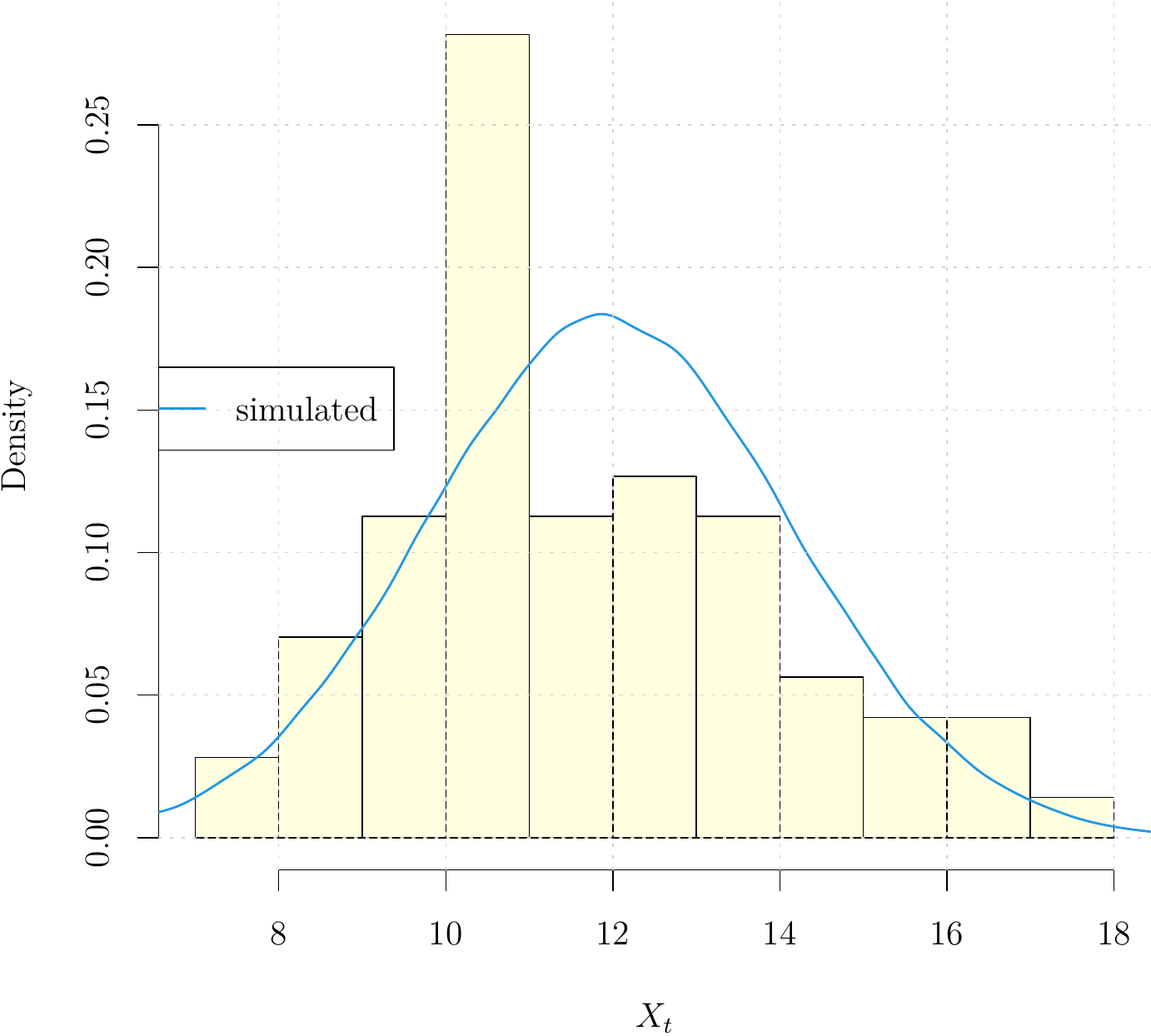}
\caption{(Left) Power spectral density of the simulated time series of 100k observations from the model fit of the third time series (temp: 30°, diet: poor). The frequency corresponding to the characteristic 6-week cycle is evidenced with a vertical dashed line.(Right) Histogram of the data (yellow) with the superimposed density of the fitted model, estimated upon the simulated series (blue line).}\label{SMfig:7}
\end{figure}

\begin{figure}[H]
  \centering
\includegraphics[width=0.45\linewidth,keepaspectratio]{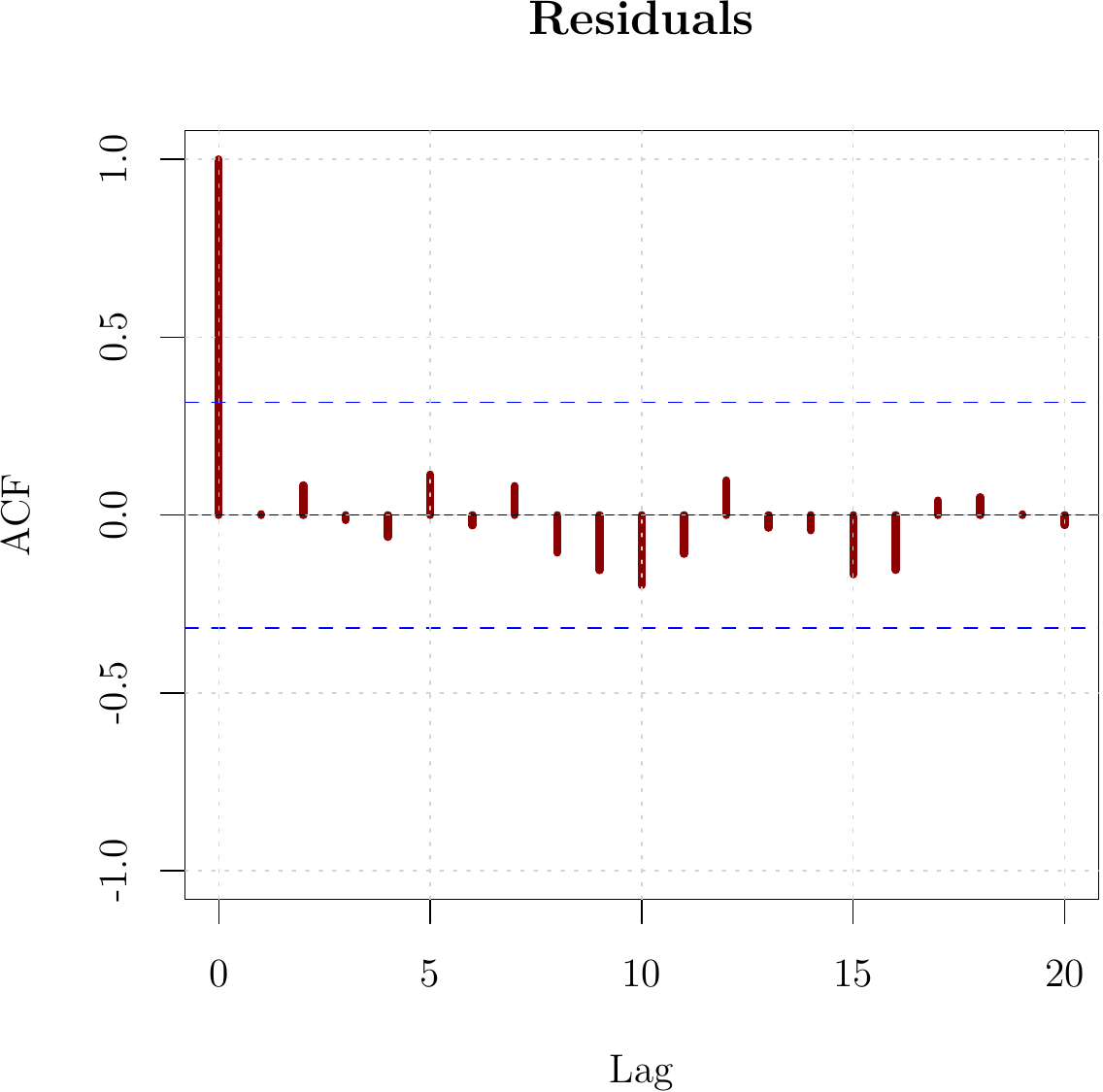}
\includegraphics[width=0.45\linewidth,keepaspectratio]{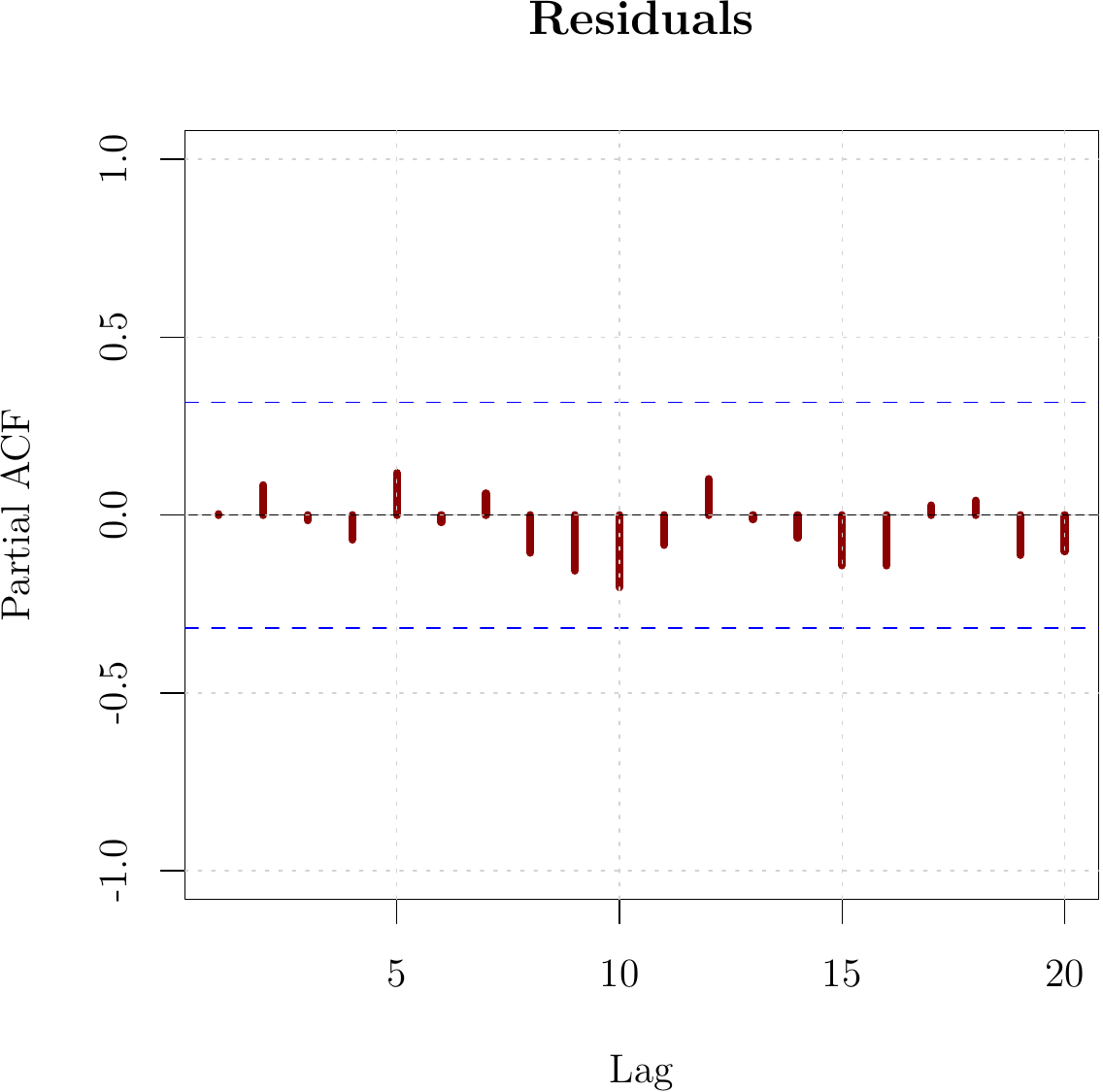}
\includegraphics[width=0.45\linewidth,keepaspectratio]{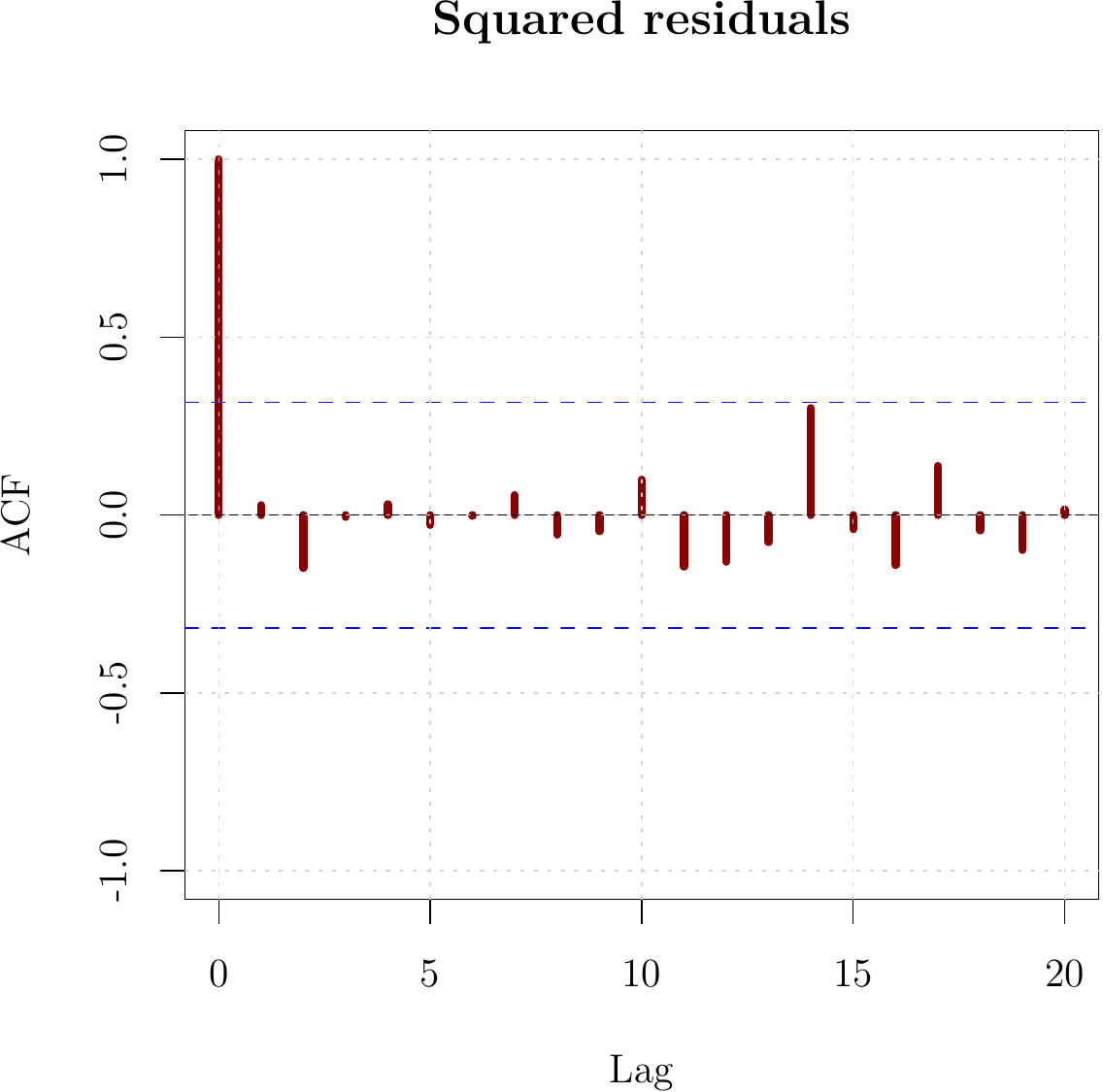}
\includegraphics[width=0.45\linewidth,keepaspectratio]{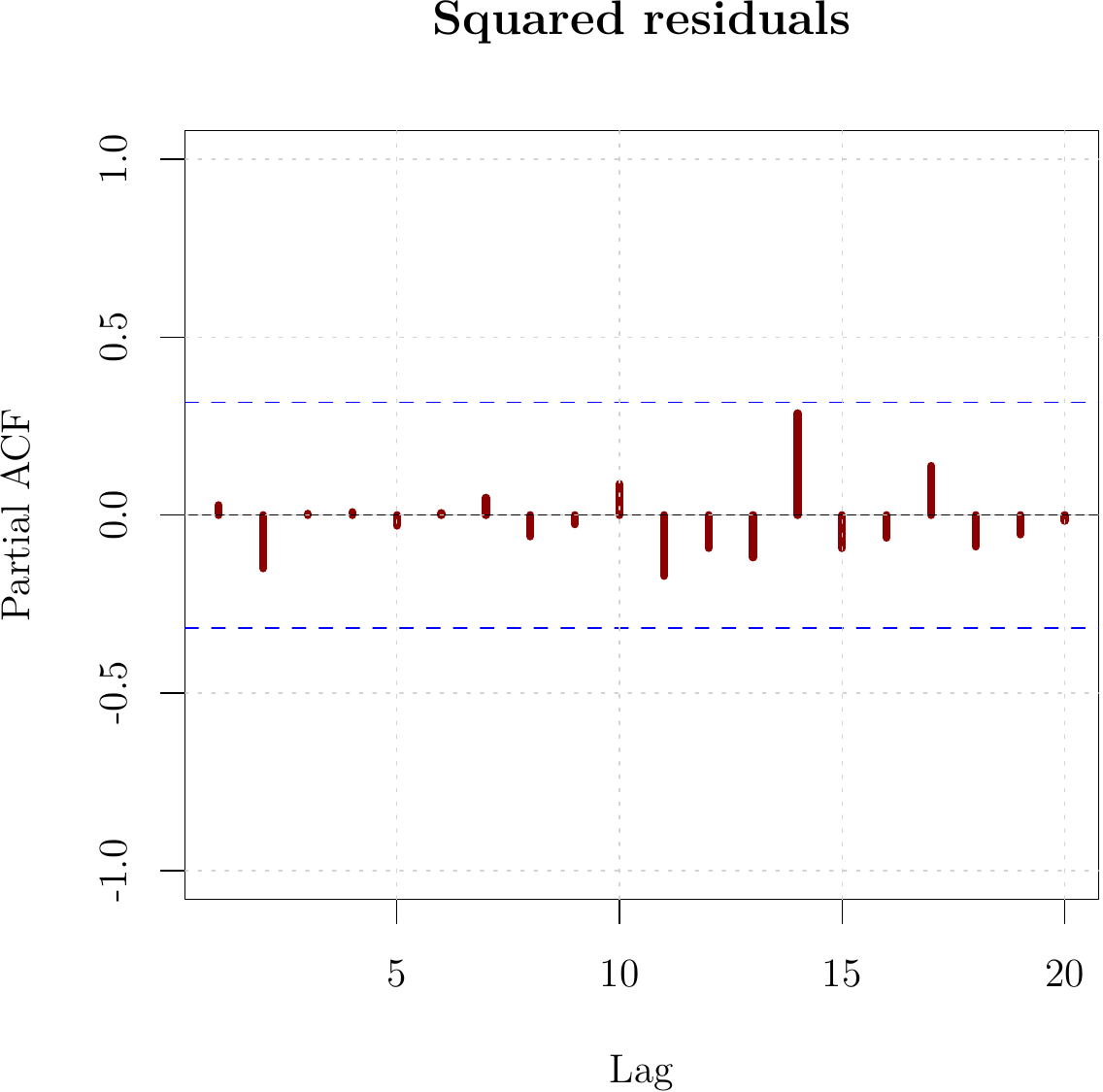}
\caption{Temperature: 27° - Diet: poor. Correlograms of the residuals of the TARMA model of Eq.~(\ref{eq:tarma}). Autocorrelation function (left) and partial autocorrelation function (right). The blue dashed lines indicate the rejection bands at 99\% level.}\label{SMfig:diag1}
\end{figure}

\begin{figure}[H]
  \centering
\includegraphics[width=0.45\linewidth,keepaspectratio]{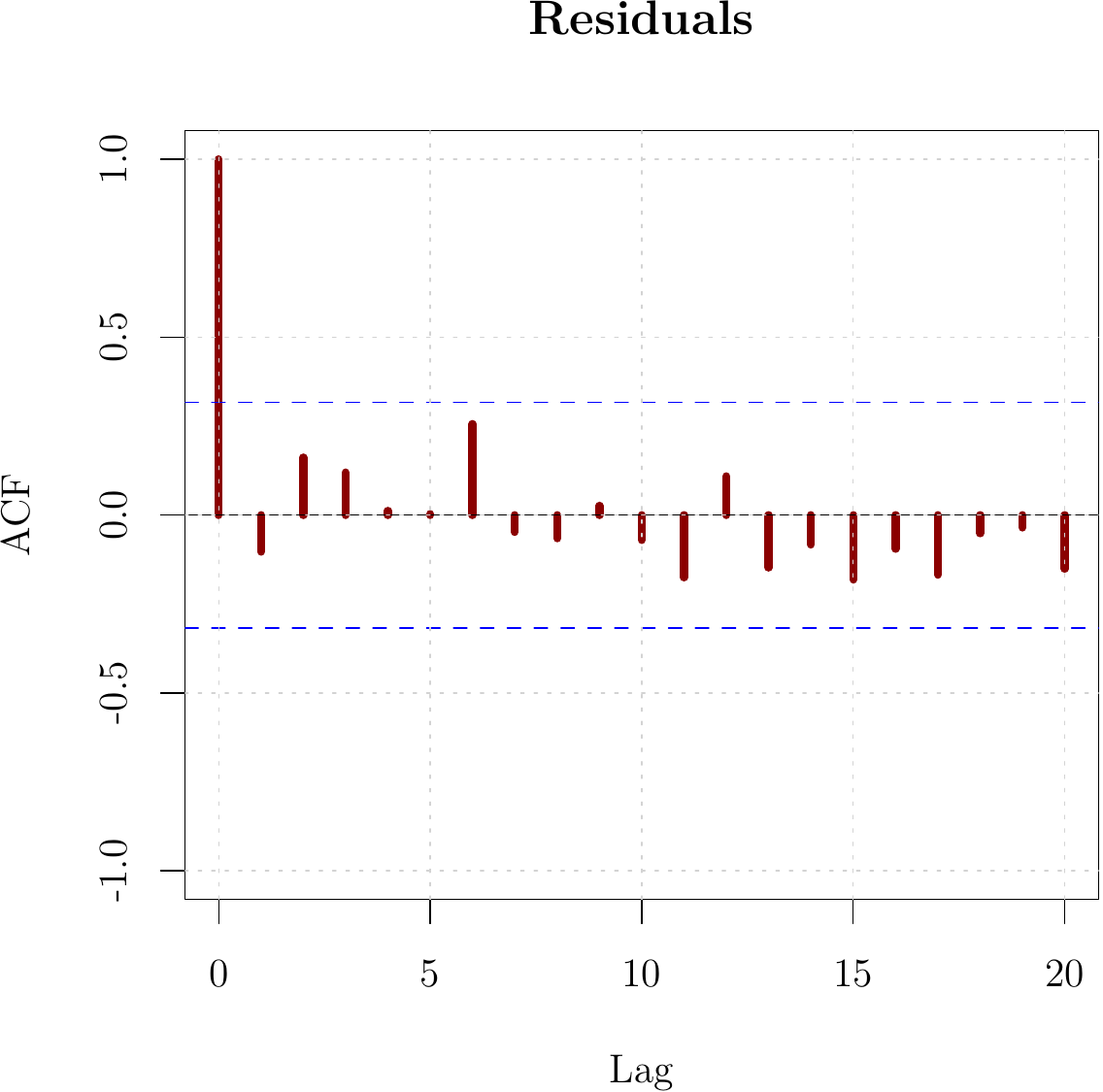}
\includegraphics[width=0.45\linewidth,keepaspectratio]{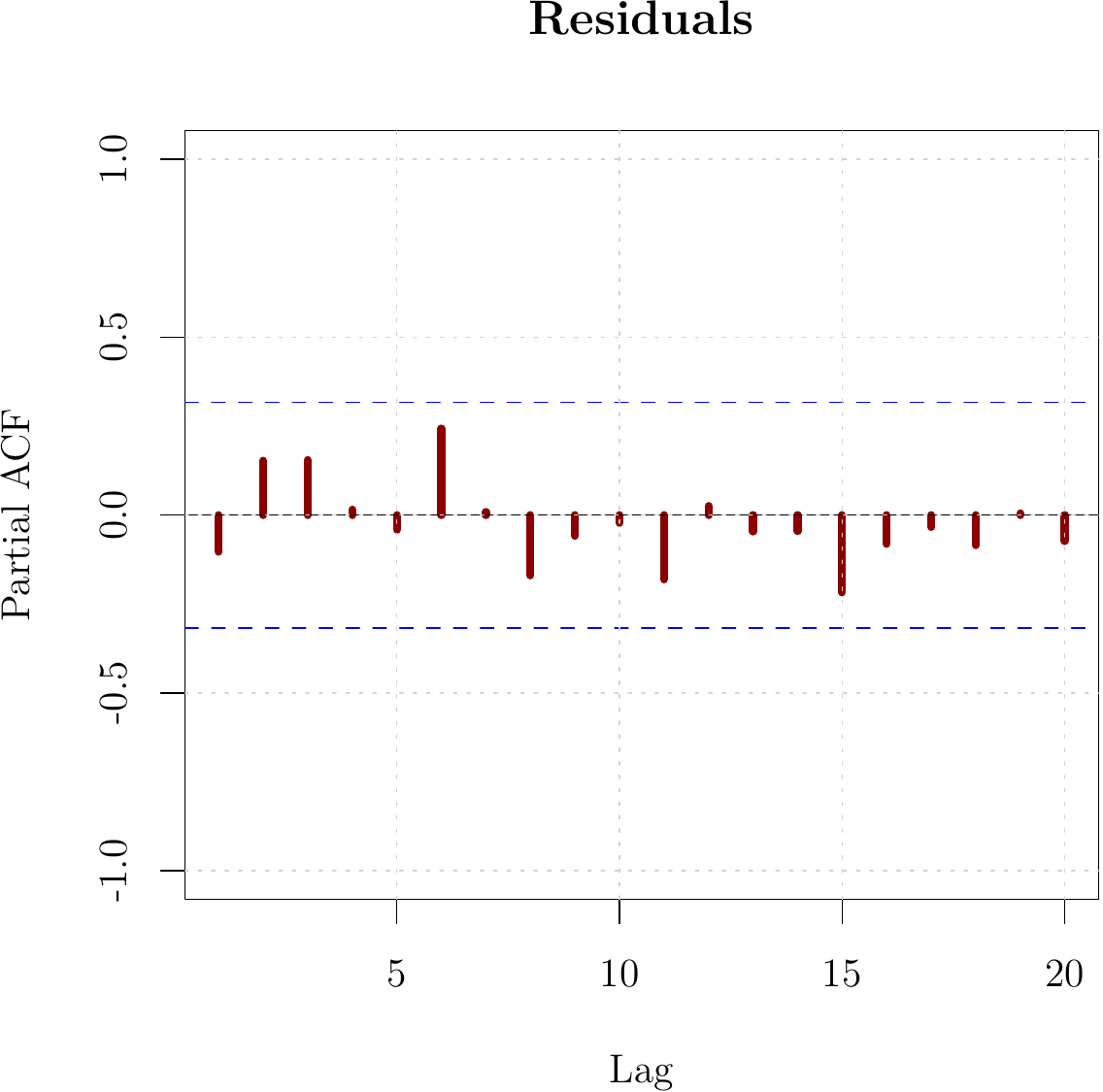}
\includegraphics[width=0.45\linewidth,keepaspectratio]{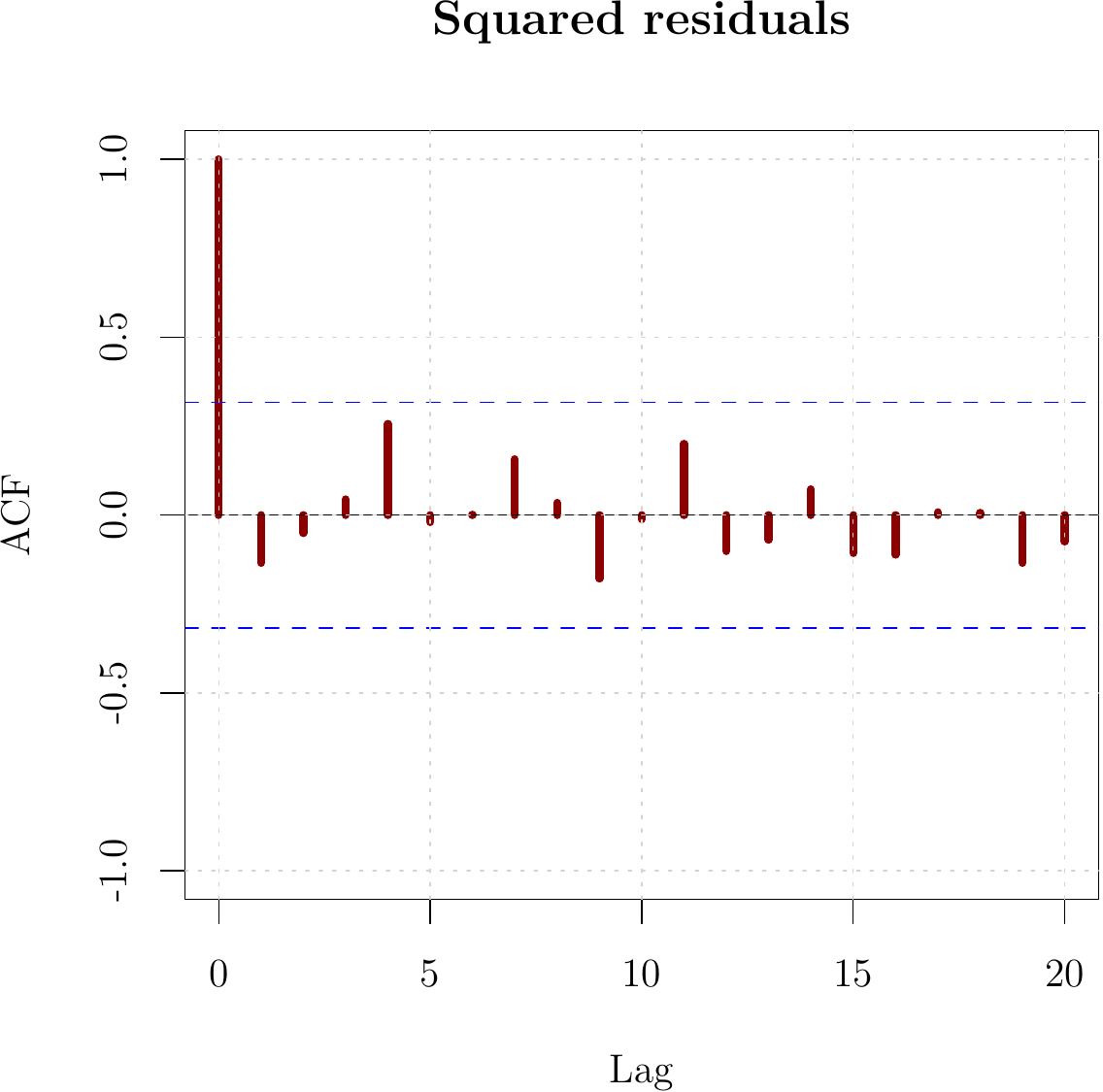}
\includegraphics[width=0.45\linewidth,keepaspectratio]{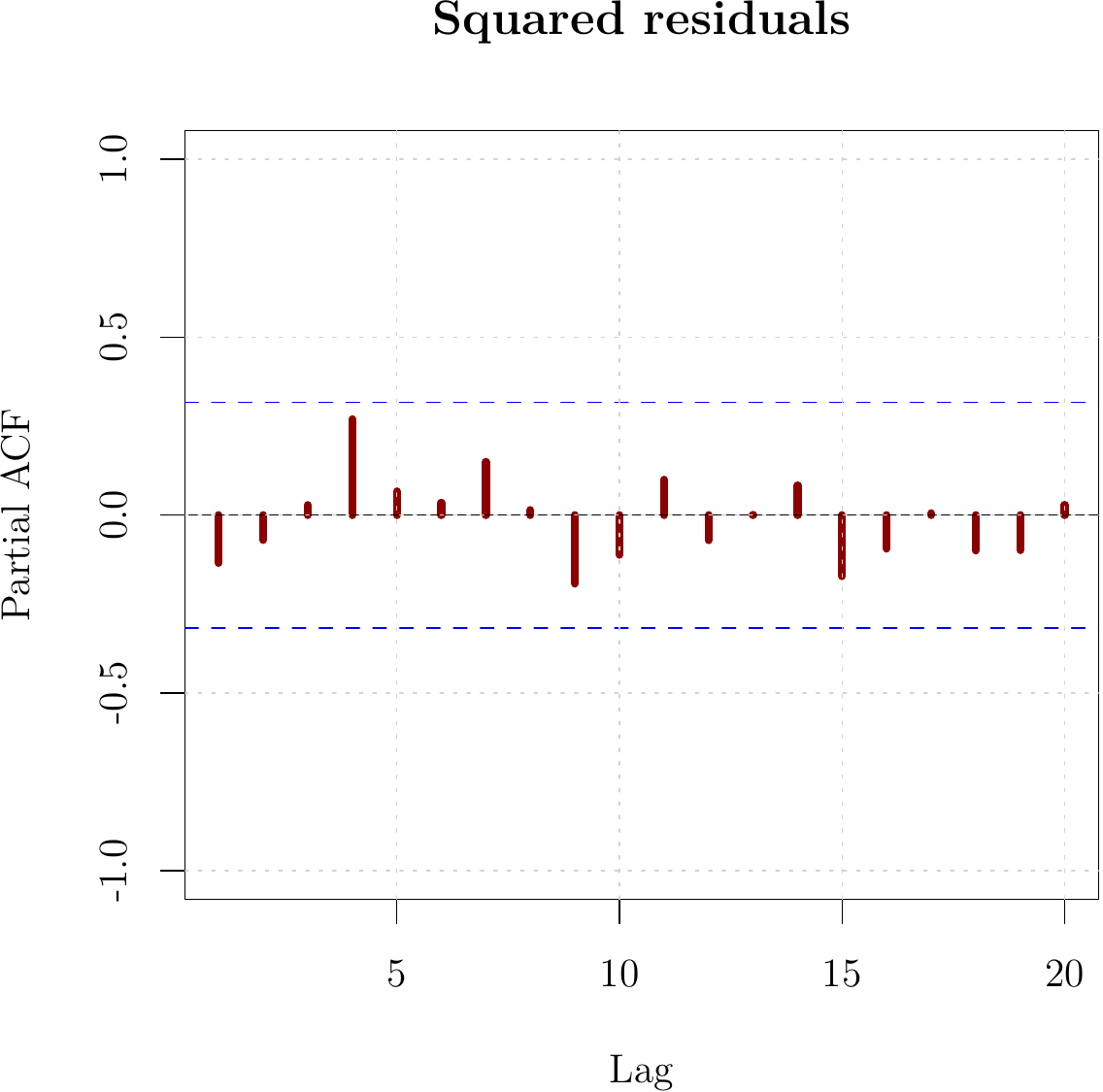}
\caption{Temperature: 27° - Diet: good. Correlograms of the residuals of the TARMA model of Eq.~(\ref{eq:tarma}). Autocorrelation function (left) and partial autocorrelation function (right). The blue dashed lines indicate the rejection bands at 99\% level.}\label{SMfig:diag2}
\end{figure}

\begin{figure}[H]
  \centering
\includegraphics[width=0.45\linewidth,keepaspectratio]{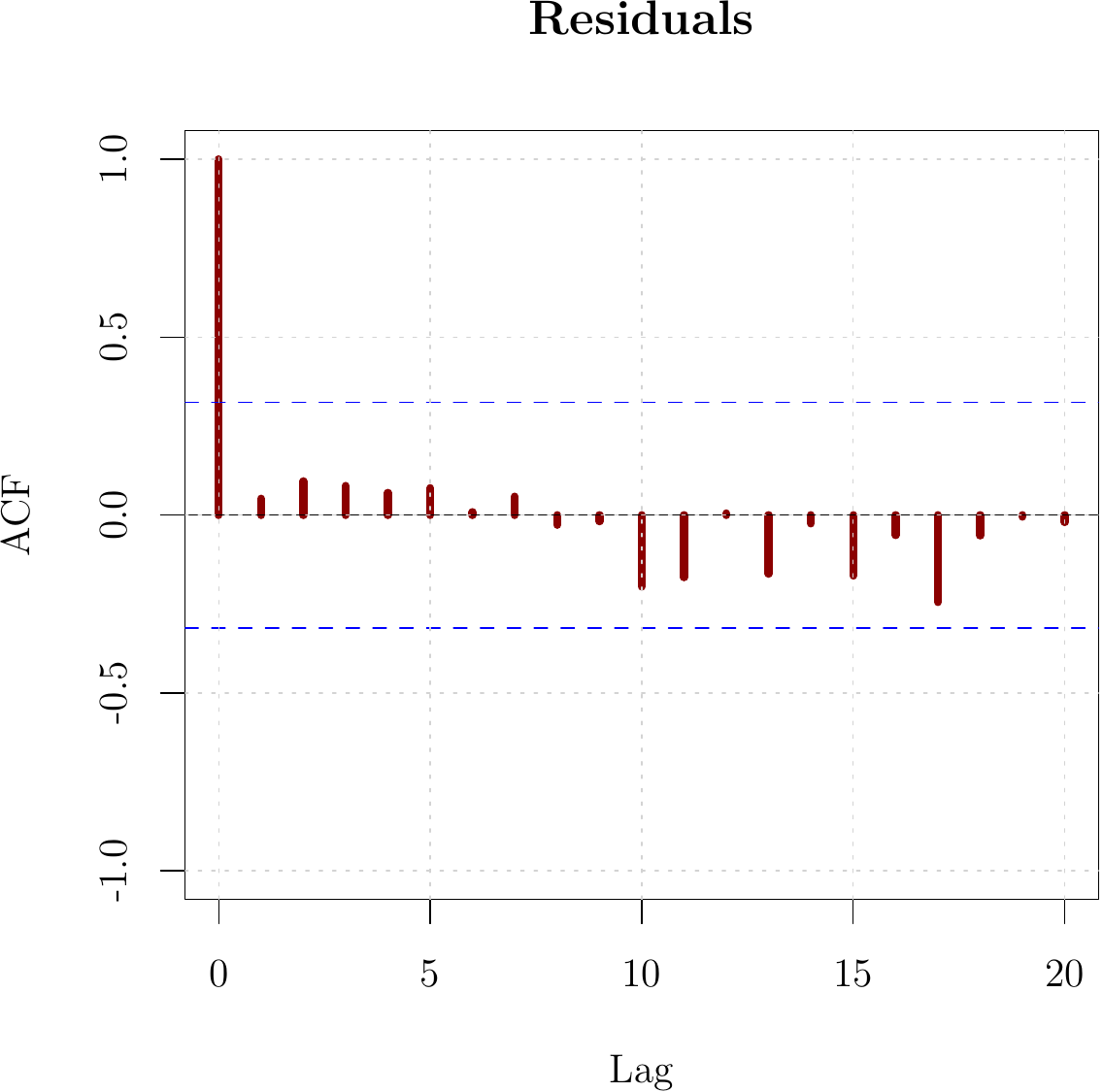}
\includegraphics[width=0.45\linewidth,keepaspectratio]{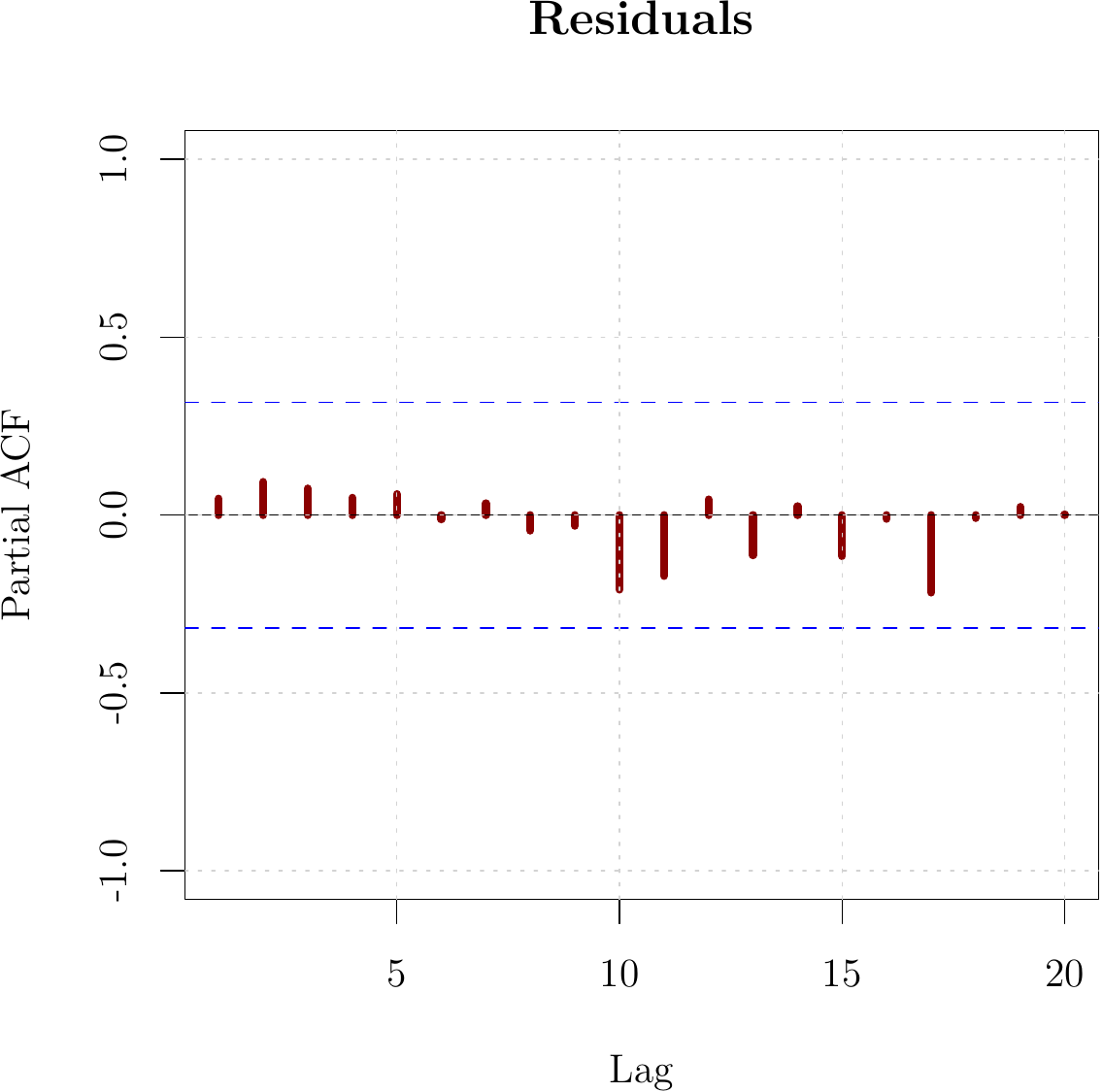}
\includegraphics[width=0.45\linewidth,keepaspectratio]{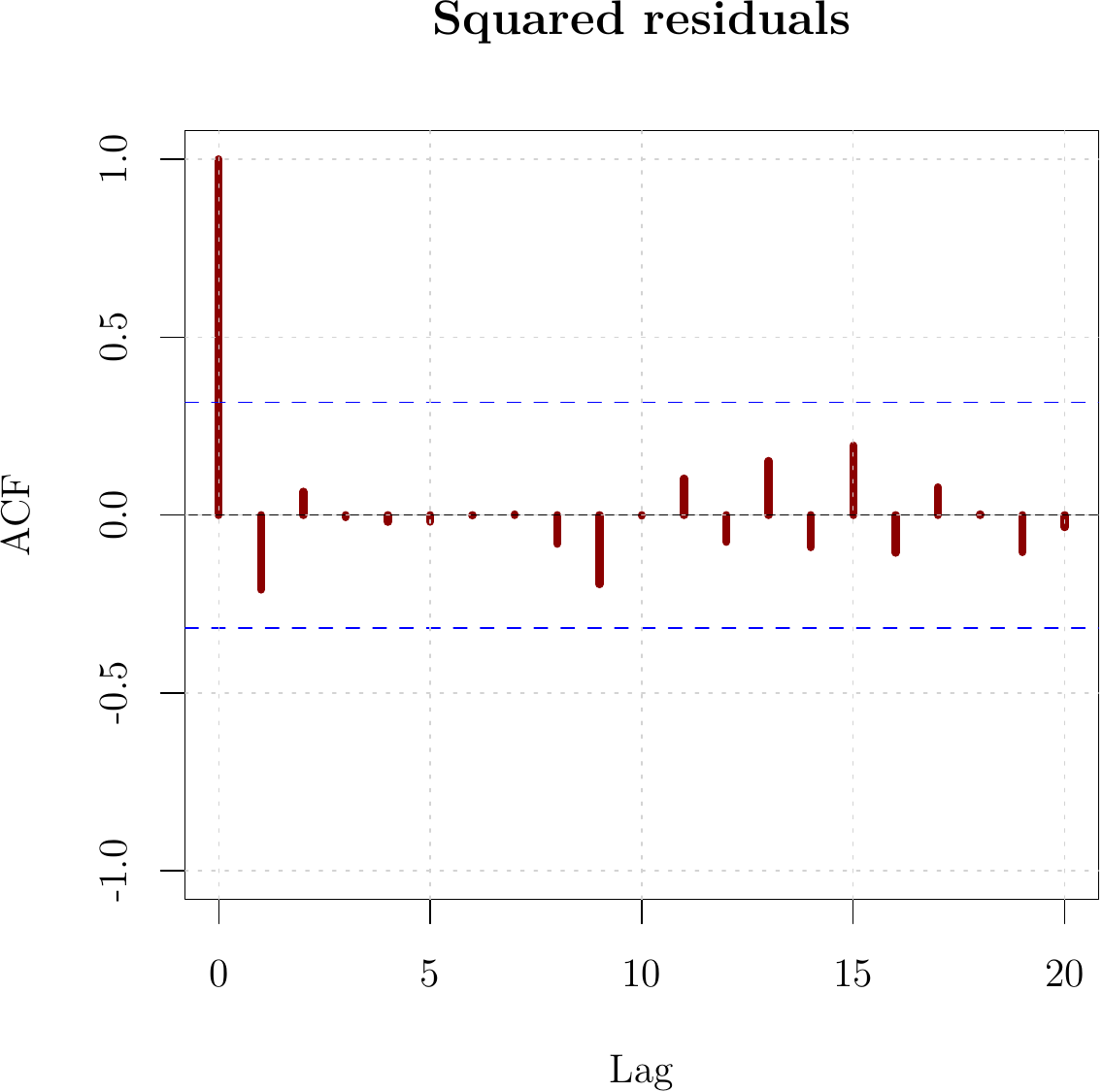}
\includegraphics[width=0.45\linewidth,keepaspectratio]{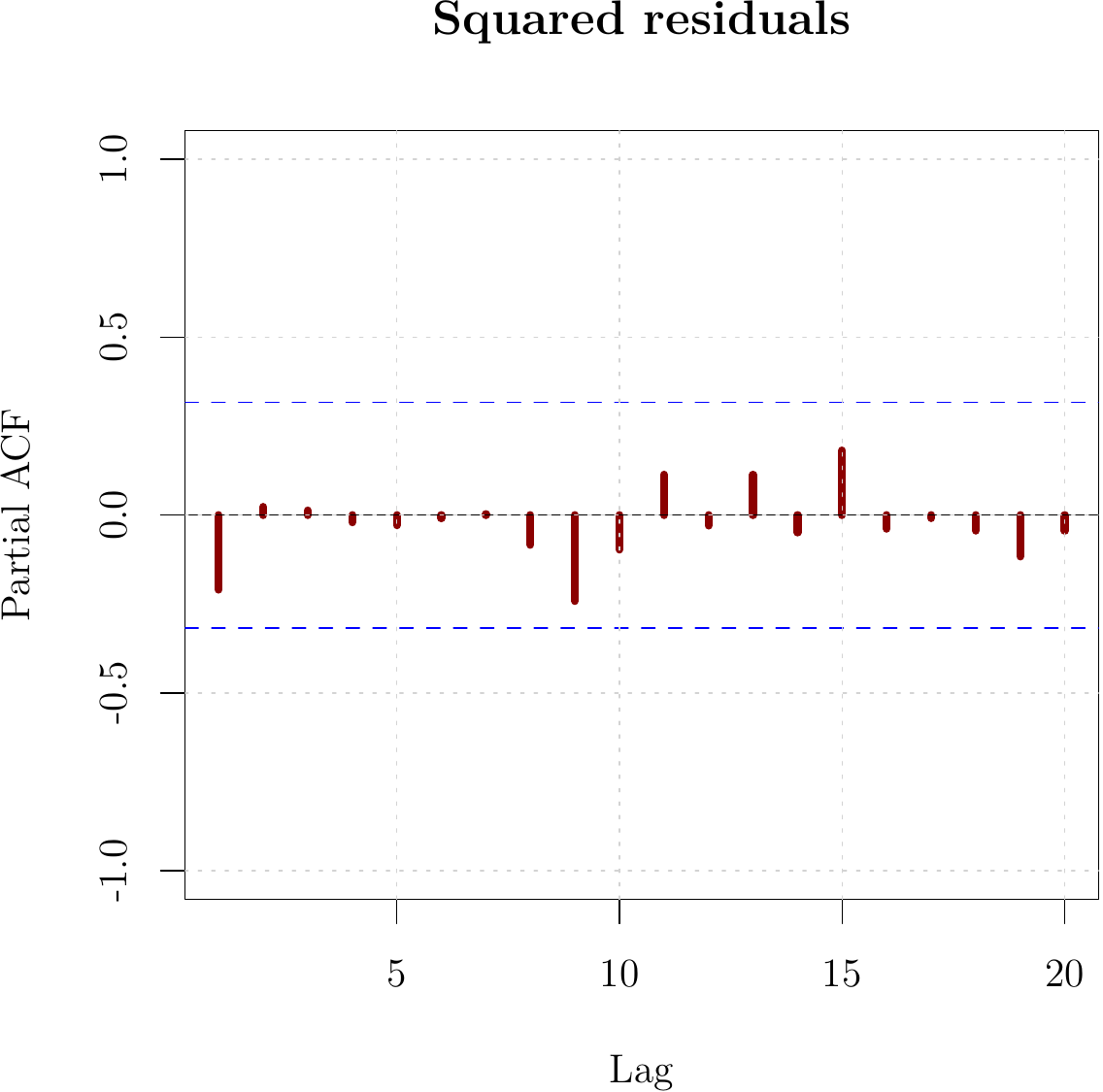}
\caption{Temperature: 30° - Diet: poor. Correlograms of the residuals of the TARMA model of Eq.~(\ref{eq:tarma}). Autocorrelation function (left) and partial autocorrelation function (right). The blue dashed lines indicate the rejection bands at 99\% level.}\label{SMfig:diag3}
\end{figure}

\begin{figure}[H]
  \centering
\includegraphics[width=0.45\linewidth,keepaspectratio]{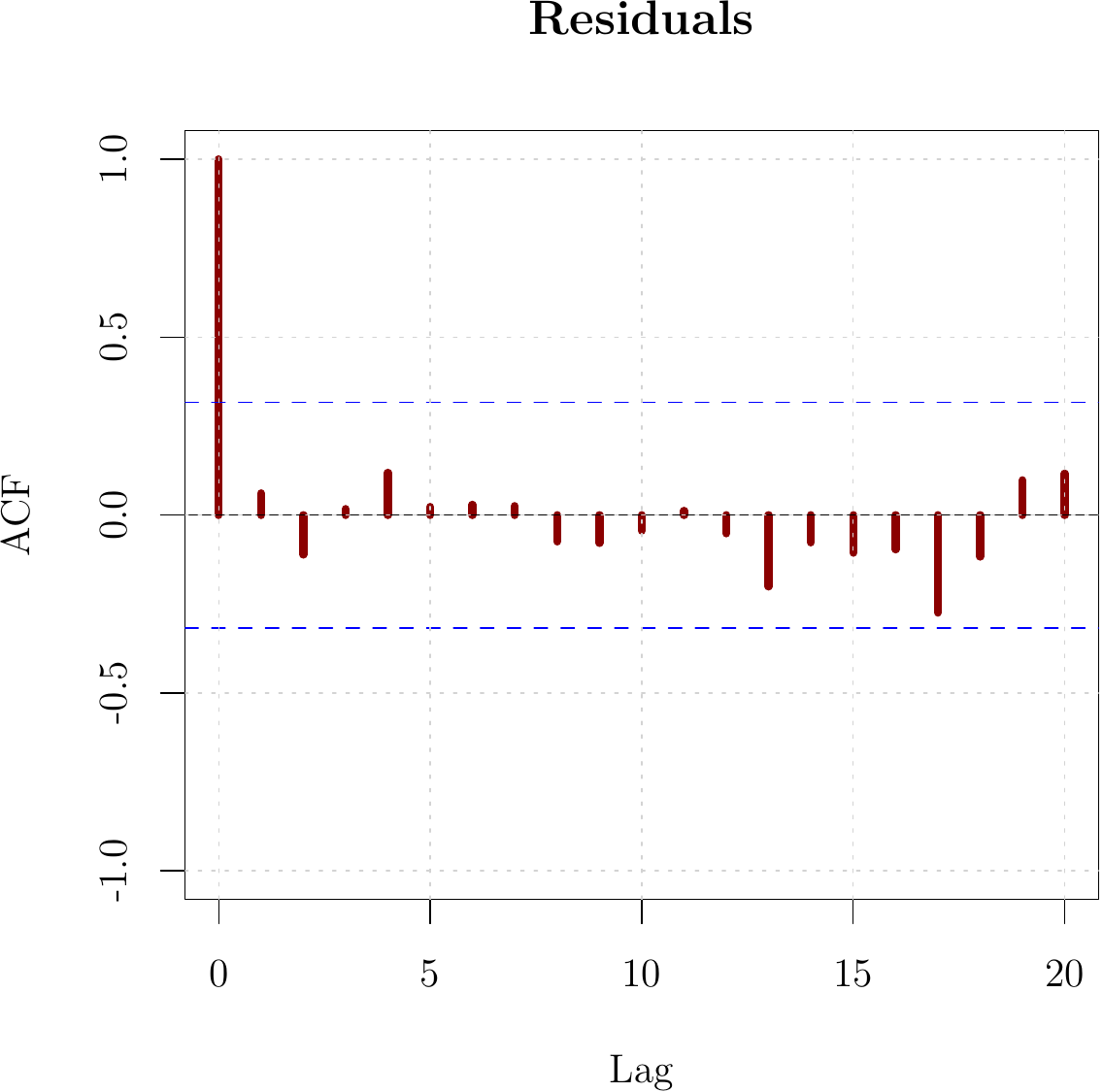}
\includegraphics[width=0.45\linewidth,keepaspectratio]{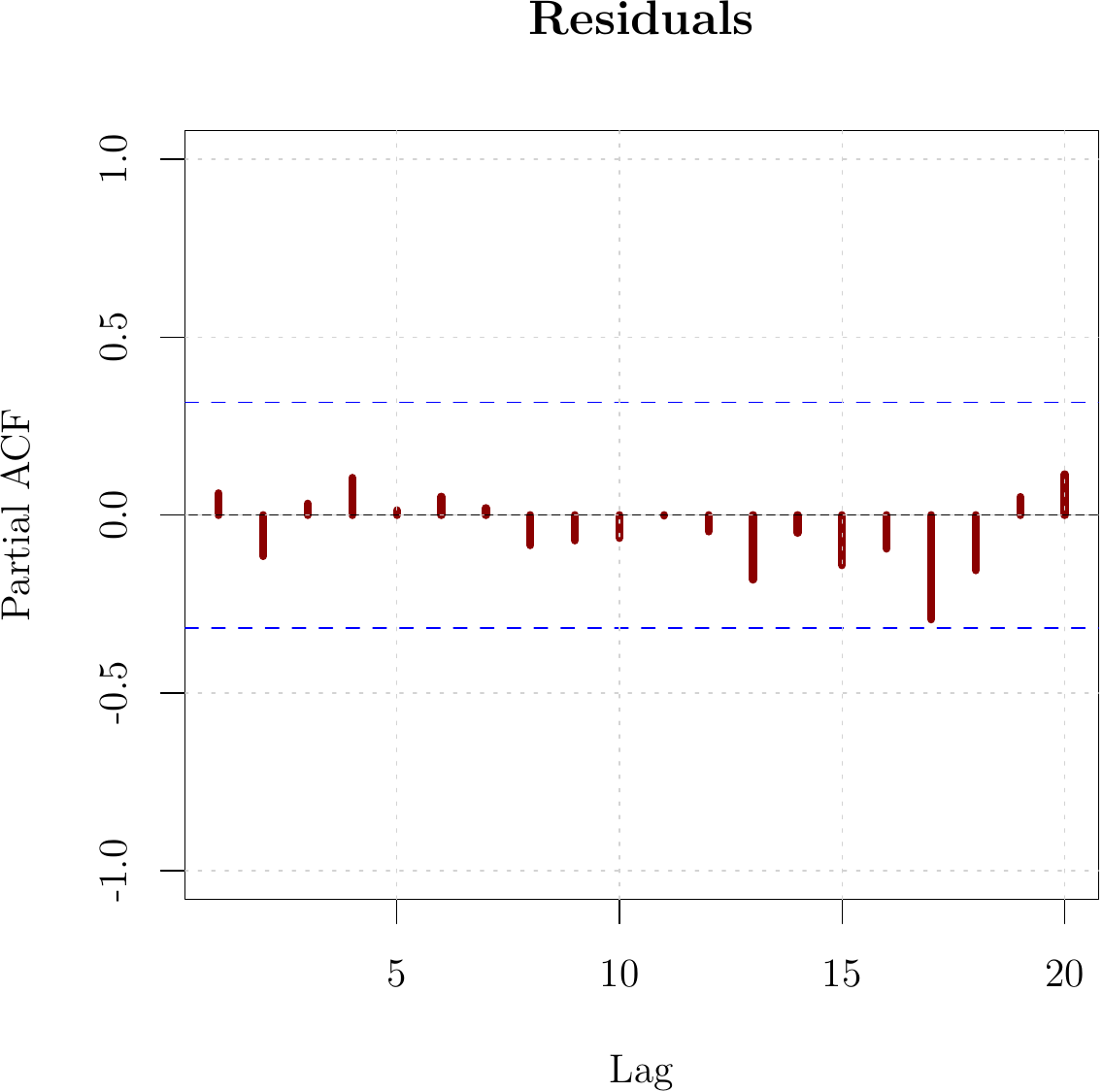}
\includegraphics[width=0.45\linewidth,keepaspectratio]{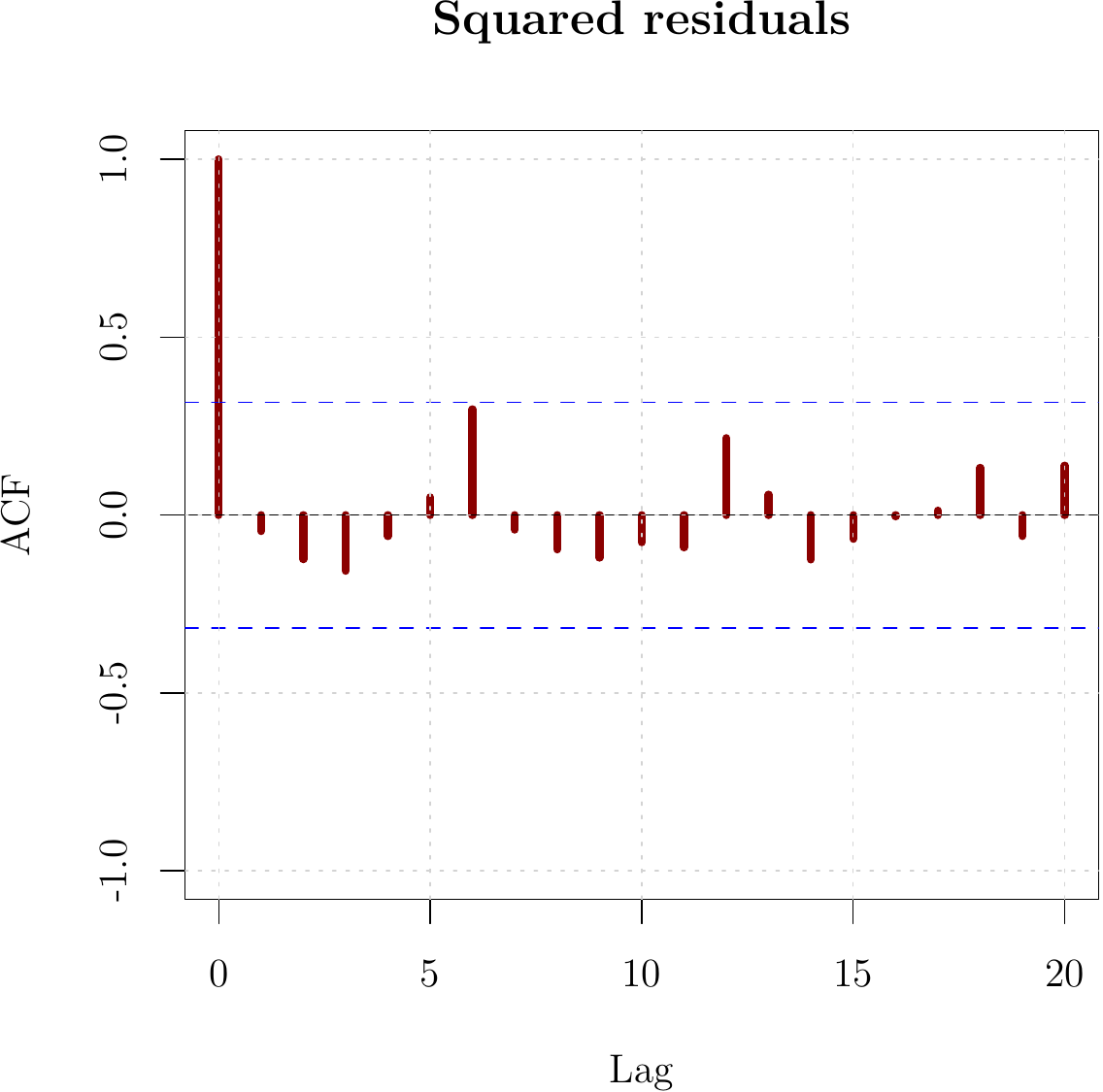}
\includegraphics[width=0.45\linewidth,keepaspectratio]{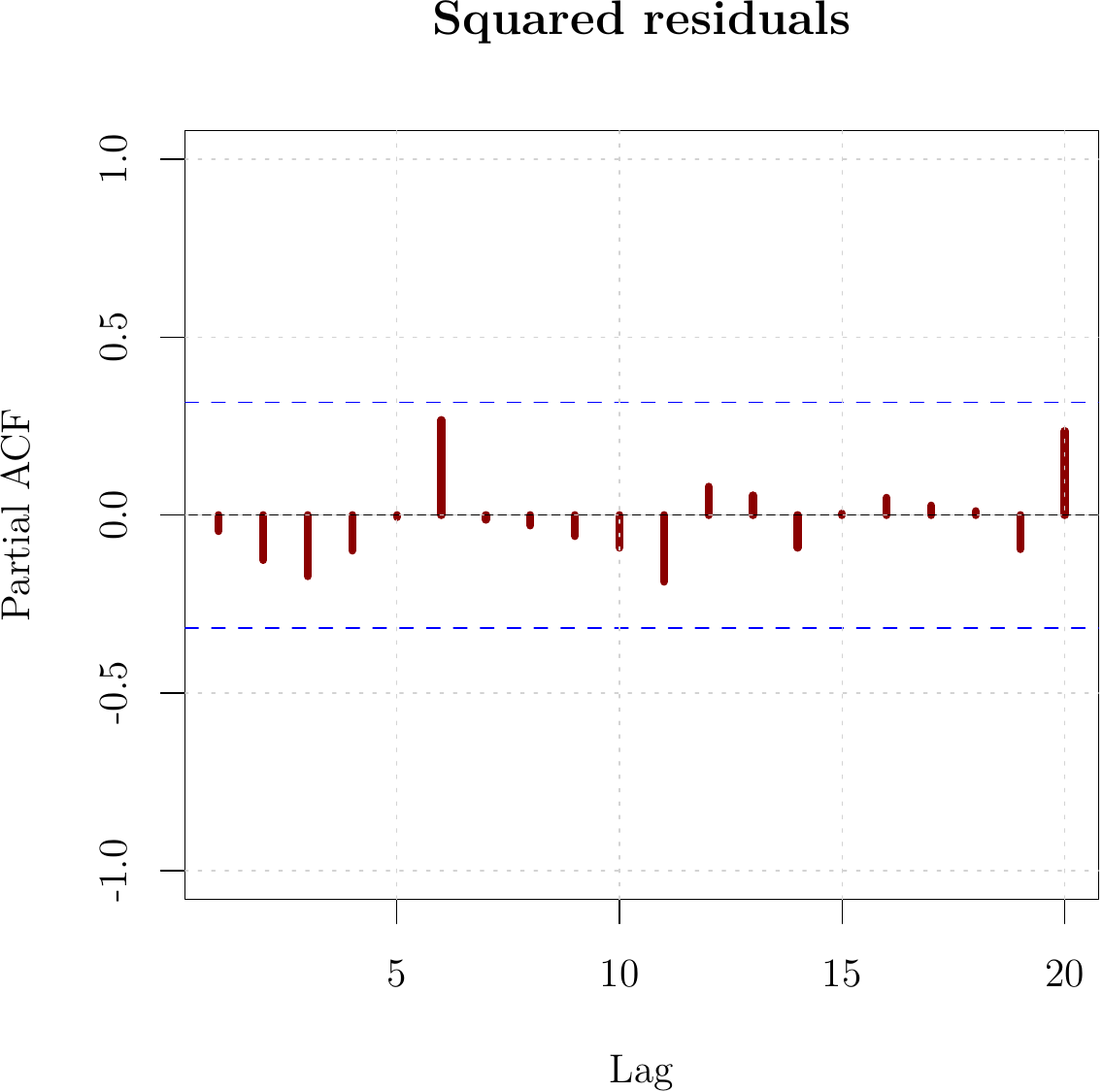}

\caption{Temperature: 30° - Diet: good. Correlograms of the residuals of the TARMA model of Eq.~(\ref{eq:tarma}). Autocorrelation function (left) and partial autocorrelation function (right). The blue dashed lines indicate the rejection bands at 99\% level.}\label{SMfig:diag4}
\end{figure}

\begin{table}[H]
  \centering
\begin{tabular}{rrrr}
\toprule
temp. & diet & W & p-value\\
\midrule
27 & poor & 0.989 & 0.849\\
27 & good & 0.984 & 0.578\\
\addlinespace
30 & poor & 0.987 & 0.730\\
30 & good & 0.963 & 0.046\\
\bottomrule
\end{tabular}
\caption{Shapiro Wilk normality test statistic and $p$-values for the residuals of the fitted threshold ARMA models.}\label{SMtab:ntest}
\end{table}

\end{document}